\documentclass[10pt, a4paper]{article} 
\usepackage[square,numbers,sort&compress]{natbib}
\usepackage[shortlabels]{enumitem}
\usepackage{setspace}
\usepackage{subfloat}
\usepackage{mathtools}
\usepackage{amsmath,amsfonts,amssymb,amsthm,mathrsfs, bm}
\usepackage[utf8]{inputenc}
\usepackage[swedish, english]{babel}
\usepackage{csquotes}
\usepackage{subfig}
\usepackage{graphicx}
\usepackage{supertabular}                                             
\usepackage[table]{xcolor} 
\usepackage{textgreek} 
\usepackage{graphicx}
\usepackage{epstopdf}
\usepackage{microtype}
\usepackage{multirow}
\usepackage{xcolor}
\usepackage{authblk}
\usepackage{datetime2}	
\usepackage[margin=1in]{geometry} % Change of margins
\usepackage{setspace}	% for double spacing
\theoremstyle{plain}% Theorem-like structures provided by amsthm.sty
\newtheorem{theorem}{Theorem}[section]
\newtheorem{lemma}{Lemma}[section]
\newtheorem{corollary}{Corollary}[section]

\usepackage{amsmath,amsfonts,amssymb,amsthm,mathrsfs, bm}
\usepackage{stmaryrd}

\usepackage{graphicx}
\usepackage{tikz}
\usepackage[utf8]{inputenc}
\usepackage{pgfplots} 
\usepackage{pgfgantt}
\usepackage{pdflscape}
\pgfplotsset{compat=newest} 
\pgfplotsset{plot coordinates/math parser=false}

\usepackage{endnotes}
\usepackage{commath}
\usepackage{gensymb}
\usepackage{mathtools}
\usepackage{tikz} 
\usetikzlibrary{calc,quotes,angles}
\usepackage{tkz-euclide}
\usetikzlibrary{automata,positioning} 

\usepackage{xcolor}
\usepackage{tikz,tkz-euclide,siunitx}
\usepackage{etoolbox} % for \patchcmd

\usetikzlibrary{quotes,arrows.meta,angles,calc,shadings,positioning,babel}

\newtheorem{definition}{Definition}[section]

\newtheorem{remark}{Remark}[section]

\usepackage{endnotes}
\usepackage{commath}
\usepackage{gensymb}
\usepackage{tcolorbox}
\makeatletter
\patchcmd{\tkz@DrawLine}{\begingroup}{\begingroup\makeatletter}{}{}
\makeatother

\allowdisplaybreaks

\DeclareMathOperator{\argmax}{argmax}
\DeclareMathOperator{\esup}{ess\, sup}

\makeatletter
\newcommand\makebig[2]{%
  \@xp\newcommand\@xp*\csname#1\endcsname{\bBigg@{#2}}%
  \@xp\newcommand\@xp*\csname#1l\endcsname{\@xp\mathopen\csname#1\endcsname}%
  \@xp\newcommand\@xp*\csname#1r\endcsname{\@xp\mathclose\csname#1\endcsname}%
}
\makeatother

\providecommand*{\ped}[1]{%
\ensuremath{_\textnormal{#1}}}
\providecommand*{\ap}[1]{%
\ensuremath{^\textnormal{#1}}}
\providecommand*{\eu}%
{\ensuremath{\mathrm{e}}}
\providecommand*{\im}%
{\ensuremath{\mathrm{i}}}
\providecommand*{\GammaF}%
{\ensuremath{\mathrm{\Gamma}}}
\providecommand*{\BetaF}%
{\ensuremath{\mathrm{\Beta}}}

\makebig{biggg} {3.0}
\makebig{Biggg} {3.5}
\makebig{bigggg}{4.0}
\makebig{Bigggg}{4.5}
\makebig{biggggg}{5.0}
\makebig{Biggggg}{5.5}
\usepackage{longtable}
\usepackage{xcolor}
\DeclareMathSymbol{\Gamma}{\mathalpha}{letters}{"00}
\DeclareMathSymbol{\Delta}{\mathalpha}{letters}{"01}
\DeclareMathSymbol{\Theta}{\mathalpha}{letters}{"02}
\DeclareMathSymbol{\Lambda}{\mathalpha}{letters}{"03}
\DeclareMathSymbol{\Xi}{\mathalpha}{letters}{"04}
\DeclareMathSymbol{\Pi}{\mathalpha}{letters}{"05}
\DeclareMathSymbol{\Sigma}{\mathalpha}{letters}{"06}
\DeclareMathSymbol{\Upsilon}{\mathalpha}{letters}{"07}
\DeclareMathSymbol{\Phi}{\mathalpha}{letters}{"08}
\DeclareMathSymbol{\Psi}{\mathalpha}{letters}{"09}
\DeclareMathSymbol{\Omega}{\mathalpha}{letters}{"0A}

\definecolor{matblue}{rgb}{0.0000,0.4470,0.7410}
\definecolor{matred}{rgb}{0.8500,0.3250,0.0980}
\definecolor{matyellow}{rgb}{0.9290,0.6940,0.1250}
\definecolor{matpurple}{rgb}{0.4940,0.1840,0.5560}
\definecolor{matgreen}{rgb}{0.4660,0.6740,0.1880}
\definecolor{matcyan}{rgb}{0.3010,0.7450,0.9330}
\definecolor{matmaroon}{rgb}{0.6350,0.0780,0.1840}

\usepackage{hyperref}
\hypersetup{
    colorlinks = true,
    linkbordercolor = {white},
            linkcolor=blue,
            bookmarks=true,
            filecolor=blue,
            urlcolor=blue,
            citecolor=blue
}

\newtcolorbox[auto counter]{modelbox}[2][]{%
  colback=white, colframe=black,
  fonttitle=\bfseries,
  title=Model~\thetcbcounter: #2,
  label=#1
}

\newcommand{\modref}[1]{Model~\ref{#1}}

\begin{document}

\title{Two-dimensional FrBD friction models for rolling contact\footnote{This document is a corrected author version of the article published in \emph{Nonlinear Dynamics} 114, 444 (2026) \url{https://doi.org/10.1007/s11071-026-12298-x}. It incorporates corrections made after publication (see Errata) and is distributed in accordance with the CC BY licence.}}
\date{}
\author[a,b]{Luigi Romano\thanks{Corresponding author. Email: luigi.romano@liu.se.}}
%\author[b]{Ole Morten Aamo}
%\author[a]{Jan Aslund}
%\author[a]{Erik Frisk}
\affil[a]{\footnotesize{Division of Vehicular Systems, Department of Electrical Engineering, Linköping University, SE-581 83 Linköping, Sweden}}
%\affil[b]{\footnotesize{Department of Engineering Cybernetics, Norwegian University of Science and Technology, O. S. Bragstads plass 2, NO-7034, Trondheim, Norway}}
\affil[b]{\footnotesize{Control Systems Technology Group, Department of Mechanical Engineering, Eindhoven University of Technology, Groene Loper 1, 5612 AZ Eindhoven, the Netherlands}}

\maketitle

\begin{abstract}
This paper develops a comprehensive two-dimensional generalisation of the recently introduced Friction with Bristle Dynamics (FrBD) framework for rolling contact problems. The proposed formulation extends the one-dimensional FrBD model to accommodate simultaneous longitudinal and lateral slips, spin, and arbitrary transport kinematics over a finite contact region. The derivation combines a rheological representation of the bristle element with an analytical local sliding-friction law. By relying on an application of the Implicit Function Theorem, the notion of sliding velocity is then eliminated, and a fully dynamic friction model, driven solely by the rigid relative velocity, is obtained. Building upon this local model, three distributed formulations of increasing complexity are introduced, covering standard linear rolling contact, as well as linear and semilinear rolling in the presence of large spin slips. For the linear formulations, well-posedness, stability, and passivity properties are investigated under standard assumptions. In particular, the analysis reveals that the model preserves passivity under almost any parametrisation of practical interest. Numerical simulations illustrate steady-state action surfaces, transient relaxation phenomena, and the effect of time-varying normal loads. The results provide a unified and mathematically tractable friction model applicable to a broad class of rolling contact systems.
\end{abstract}
\section*{Keywords}
Rolling contact; friction; friction modelling; contact mechanics; distributed parameter systems; semilinear systems

\section{Introduction}\label{intro}
Rolling and sliding contact phenomena are central to a wide range of engineering systems \cite{KinematicsMio,Flores,Flores2}, including wheel-rail interaction \cite{Knothe,KalkerBook,2000RCP}, tyre-road dynamics \cite{Guiggiani,Pacejka2,LibroMio,Gauterin,Gauterin2}, spherical robots \cite{Sphere1,Sphere2,Sphere3,Sphere4,SphereNoSlip1,SphereNoSlip2,SphereNoSlip3}, and general mechanical and tribological components such as continuous transmissions 
\cite{CarboneTrans,TransModel,Belt1,Frendo1,Frendo2,Frendo3} and bearing elements \cite{bearing1,bearing2,bearing3,bearing4,bearing5}.

Historically, the first analytical theory of rolling contact was developed by Kalker, who combined kinematic relationships with a local Coulomb-Amontons friction model. Focusing on similar elastic cylinders undergoing longitudinal slip\footnote{Or \emph{creepage}, as better known within the railway community.} under dry friction conditions, Kalker was able to analyse both steady-state and transient phenomena \cite{Kalker51,Kalker5} within the theoretical framework offered by his full nonlocal theory of elasticity \cite{KalkerPhD}. The solutions obtained in these seminal works are now regarded as classical results in contact mechanics and are reported in several authoritative references \cite{KalkerBook,Johnson,Goryacheva,Barber}. However, pure longitudinal slip remains the only case for which a closed-form analytical solution is available. The intrinsic complexities associated with unsteady rolling have confined other exact analyses, particularly those involving combined slip conditions, to the numerical domain \cite{Vollebregt1,Vollebregt2,Vollebregt3,Vollebregt4}, whereas several authors have proposed approximate analytical solutions \cite{Nielsen,Alfredsson}.

In fact, more involved cases featuring combined translational and spin slips can be studied, at least qualitatively, using Kalker's simplified model \cite{KalkerBook,Johnson,KalkerSimp}, which employs a Winkler-type approximation for the elastic relationship between deformations and stresses. Closely related formulations, based on local constitutive laws, quickly became popular in vehicle dynamics, where they are known as \emph{brush models}. A key advantage of these simplified models is that they replace the convolution integrals arising from full elasticity theory with linear algebraic relations, thereby enabling both steady-state and transient analyses for a wide range of wheel-rail and tyre-road operating conditions. Analytical solutions derived from brush models can be found, for example, in \cite{Gross,Alonso1,Alonso2,Ciavarella1,Ciavarella2,Ciavarella3,Al-Bender,USB,LibroMio}. More recent extensions addressing large spin slips and omnidirectional rolling are presented in \cite{Meccanica2,SphericalWheel}, whilst \cite{Vollebregt1,Vollebregt2,Vollebregt3,Vollebregt4} discusses adaptations to account for lubricated friction and third-body layers.

Despite their usefulness and wide adoption, classical brush models suffer from an important limitation: they require an explicit distinction between adhesion and sliding regions within the contact patch. When the rolling contact process is formulated as a spatially distributed problem, this feature introduces insurmountable obstacles to obtaining general analytical results whilst complicating rigorous mathematical treatment, as well as integration with control or estimation algorithms. To address these drawbacks, several alternative theories of rolling and spinning friction have been proposed.
For example, the local approaches presented in \cite{Zhuravlev1,Zhuravlev2,Kireenkov1,Kireenkov2} estimate the friction forces and moments acting on spinning or sliding bodies by relying on simplified representations of rolling kinematics, permitting the derivation of compact expressions for the global quantities of interest. Dynamic friction models, such as the Dahl and LuGre formulations \cite{Astrom1,Olsson,Astrom2}, have also been extended to rolling contact applications, particularly in tyre-road interaction \cite{Sorine,TsiotrasConf,Tsiotras1,Tsiotras3,Deur0,Deur1,Deur2}. These descriptions do not differentiate between stick and slip regimes, since they only model the average deflection of bristle elements within the contact area. As such, they possess a mathematical structure that is more amenable to analysis and control design. Conceptually, similar descriptions may also be adapted starting from other dynamic friction formulations, including Leuven-type \cite{Integrated,Leuven} and elastoplastic models  \cite{Elasto1,Elasto2}. However, most of these approaches remain empirical in nature and therefore provide only limited physical insight into the underlying mechanisms of friction.

Recently, as a variant of LuGre, the Friction with Bristle Dynamics (FrBD) formulation was introduced in \cite{FrBD} as a first-order approximation to a rheological description of bristle-like elements attached to the rolling components. Being physically motivated, the FrBD model provides a consistent dynamic description of sliding in the presence of dry and lubricated friction. However, the variant developed in \cite{FrBD} was conceived for line contact conditions and, albeit able to capture essential transient and steady-state effects, its applicability to rolling and spinning contact, which is inherently a two-dimensional distributed phenomenon, remains limited. The present paper addresses this intrinsic limitation by presenting a full two-dimensional generalisation of the FrBD model, suitable to describe rolling contact processes occurring between a deformable body and a rigid substrate, or between similarly elastic bodies, in the presence of dry or lubricated friction. As for its one-dimensional counterpart developed in \cite{FrBD}, the main ingredient for the refined formulation is a physically motivated constitutive relationship postulating the bristle force as a function of both deformation and deformation rate. This local constitutive relationship is coupled with an analytical law for sliding friction through a nonlinear mapping involving a generalised friction coefficient matrix. The combination leads to an implicit differential equation for the sliding velocity, which is approximated analytically using a constructive version of the Implicit Function Theorem. The resulting dynamic friction law depends solely on the rigid relative velocity and bristle deflection, eliminating the concept of sliding velocity and yielding a nonlinear \emph{ordinary differential equation} (ODE) applicable to both sliding and rolling contexts. In particular, to move from a local dynamic model to a spatially distributed rolling contact formulation, the paper adopts an Eulerian approach: through the definition of a transport-velocity field and a change of variables from time to travelled distance, the bristle dynamics is expressed, in very general terms, as a first-order nonlinear \emph{partial differential equation} (PDE) defined over a possibly time-dependent contact domain. This reformulation enables the friction state to vary within the contact area, capturing important dynamical effects such as loading and unloading, relaxation from leading to trailing edge, and spin-induced asymmetries.

Specifically, three distributed rolling contact models are derived from this general PDE. The first is essentially a standard description, similar to those already encountered, for instance, in \cite{Tsiotras1,Tsiotras3,Deur0,Deur1}, and qualifies as the FrBD analogue of the classical brush models. The second model variant accounts more explicitly for the presence of large spin slips, which are incorporated exactly into the expressions for the relative and transport velocities, yielding a rather involved semilinear PDE. Finally, the third variant is introduced as an approximation to the semilinear one, and only considers first-order effects that are relevant for strongly spin-dominated regimes. It shares some similarities with the advanced brush and LuGre-brush models proposed in \cite{Meccanica2,SphericalWheel} and \cite{LuGreSpin}, albeit being substantially more complex. For the linear formulations introduced in the paper, well-posedness is established by invoking standard results, ensuring existence and uniqueness of mild and classical solutions under appropriate regularity conditions on the transport velocity, slip inputs, and matrices of coefficients.
The mathematical analysis further investigates important notions such as stability -- including \emph{input-to-state stability} (ISS) and \emph{input-to-output stability} (IOS) --, and passivity, which, apart from having important physical interpretations, are crucial for control-oriented applications. A central finding is that, in contrast to LuGre-based formulations, the linear models delivered by the paper are (almost) always passive, as inherited from the underlying rheological structure of the bristle-like elements.

The theoretical findings of the paper are complemented with numerical results illustrating both steady-state and transient phenomena in rolling contact. Action surfaces are computed to characterise steady force-slip relationships, whilst transient simulations focus on relaxation processes and on the effect of time-varying contact regions or normal forces.
Overall, the paper provides the first mathematically rigorous, two-dimensional dynamic friction model tailored specifically to rolling contact systems. It unifies rheological bristle dynamics, nonlinear friction laws, and PDE representations of distributed contact, offering a versatile and physically grounded framework for future developments in wheel-rail, tyre, and general rolling contact modelling.
%For convenience, the technical results established in the paper are instead confined to Appendix~\ref{app:proofs}, whereas Appendix~\ref{App:transferFunc} reports the analytical expressions derived for the transfer functions of the models with lumped and distributed tyre relaxation.

The remainder of this manuscript is organised as follows. Section~\ref{sect:2Dext} details the derivation of the extended two-dimensional FrBD formulation combining physical-oriented models with sound analytical arguments.  Then, Sect.~\ref{sect:models} particularises the general equations obtained in Sect.~\ref{sect:2Dext} for rolling contact systems, delivering three different model variants with increasing level of complexity. The linear models introduced in Sect.~\ref{sect:models} are analysed in Sect.~\ref{sect:math} in terms of stability and passivity. These notions are also instrumental in highlighting key mathematical features connected with friction-related phenomena. A rich variety of numerical results are then reported and discussed in Sect.~\ref{sect:numer}, concerning both steady-state and transient behaviours. Finally, the main conclusions of the paper are summarised in Sect.~\ref{sect:conclusion}, where possible directions for future research are also indicated. Technical results and additional modelling details are instead confined to Appendices~\ref{app:1},~\ref{app:Patch}, and~\ref{app:details}.

\subsection*{Notation}
In this paper, $\mathbb{R}$ denotes the set of real numbers; $\mathbb{R}_{>0}$ and $\mathbb{R}_{\geq 0}$ indicate the set of positive real numbers excluding and including zero, respectively. The set of positive integer numbers is indicated with $\mathbb{N}$, whereas $\mathbb{N}_{0}$ denotes the extended set of positive integers including zero, i.e., $\mathbb{N}_{0} = \mathbb{N} \cup \{0\}$; the set of positive even integers is indicated with $2\mathbb{N}$.
The set of $n\times m$ matrices with values in $\mathbb{F}$ ($\mathbb{F} = \mathbb{R}$, $\mathbb{R}_{>0}$, $\mathbb{R}_{\geq0}$) is denoted by $\mathbf{M}_{n\times m}(\mathbb{F})$ (abbreviated as $\mathbf{M}_{n}(\mathbb{F})$ whenever $m=n$). $\mathbf{Sym}_n(\mathbb{R})$ represents the group of symmetric matrices with values in $\mathbb{R}$; the identity matrix on $\mathbb{R}^n$ is indicated with $\mathbf{I}_n$. A positive-definite matrix is noted as $\mathbf{M}_n(\mathbb{R}) \ni \mathbf{Q} \succ \mathbf{0}$; a positive semidefinite one as $\mathbf{M}_n(\mathbb{R}) \ni \mathbf{Q} \succeq \mathbf{0}$. 
The standard Euclidean norm on $\mathbb{R}^n$ is indicated with $\norm{\cdot}_2$; operator norms are simply denoted by $\norm{\cdot}$.
Given a domain $\Omega$ with closure $\overline{\Omega}$, $L^p(\Omega;\mathcal{Z})$ and $C^k(\overline{\Omega};\mathcal{Z})$ ($p, k \in \{1, 2, \dots, \infty\}$) denote respectively the spaces of $L^p$-integrable functions and $k$-times continuously differentiable functions on $\overline{\Omega}$ with values in $\mathcal{Z}$ (for $T = \infty$, the interval $[0,T]$ is identified with $\mathbb{R}_{\geq 0}$). In particular, $L^2(\Omega;\mathbb{R}^n)$ denotes the Hilbert space of square-integrable functions on $\Omega$ with values in $\mathbb{R}^n$, endowed with inner product $\langle \bm{\zeta}_1, \bm{\zeta}_2 \rangle_{L^2(\Omega;\mathbb{R}^n)} = \int_\Omega \bm{\zeta}_1^{\mathrm{T}}(\bm{x})\bm{\zeta}_2(\bm{x}) \dif \bm{x}$ and induced norm $\norm{\bm{\zeta}(\cdot)}_{L^2(\Omega;\mathbb{R}^n)}$. The Hilbert space $H^1(\Omega;\mathbb{R}^n)$ consists of functions $\bm{\zeta}\in L^2(\Omega;\mathbb{R}^n)$ whose weak derivative also belongs to $L^2(\Omega;\mathbb{R}^n)$. For a function $f :\Omega \to \mathbb{R}$, the sup norm is defined as $\norm{f(\cdot)}_\infty \triangleq \esup_{\Omega} \abs{f(\cdot)}$; $f : \Omega \to \mathbb{R}$ belongs to the space $L^\infty(\Omega;\mathbb{R})$ if $\norm{f(\cdot)}_\infty < \infty$. A function $f \in C^0(\mathbb{R}_{\geq 0}; \mathbb{R}_{\geq 0})$ belongs to the space $\mathcal{K}$ if it is strictly increasing and $f(0) = 0$; $f \in \mathcal{K}$ belongs to the space $\mathcal{K}_\infty$ if it is unbounded. Finally, a function $f \in C^0(\mathbb{R}_{\geq 0}^2; \mathbb{R}_{\geq 0})$ belongs to the space $\mathcal{KL}$ if $f(\cdot,t) \in \mathcal{K}$ and is strictly decreasing in its second argument, with $\lim_{t\to \infty}f(\cdot,t) = 0$. 

%Finally, for a scalar $\omega \in \mathbb{R}$, the associated matrix $[\omega]_{\times} \in \mathbf{M}_2(\mathbb{R})$ is defined as
%\begin{align}
%\begin{bmatrix}\omega\end{bmatrix}_{\times} \triangleq \begin{bmatrix} 0 & -\omega \\ \omega & 0\end{bmatrix}.
%\end{align}

\section{2D extension of the FrBD friction model}\label{sect:2Dext}

This section extends the FrBD friction model to cover two-dimensional contact cases, where relative rolling and sliding motions may occur in both the longitudinal and lateral directions. The model is derived by mixing the right doses of physical intuition and algebraic dexterity. For this, three main ingredients are needed: a rheological description of a bristle-like element, a local friction model, and a sufficiently convincing mathematical artifice. The first two modelling ingredients are introduced in Sect.~\ref{sect:rhelANdFr}, and then mixed mathematically in Sect.~\ref{sect:DynamicDer}.

\subsection{Rheological and friction models}\label{sect:rhelANdFr}
As already anticipated, two key constituents in the derivation of the extended FrBD model are a rheological description of a bristle element, and an analytical description of the local sliding friction process. These are discussed in detail in Sects.~\ref{sect:rheol} and~\ref{sect:frictionModel}.

\subsubsection{Rheological model}\label{sect:rheol}
The first step in deriving a two-dimensional FrBD friction model is to specify an appropriate rheological representation for the bristle element. This may be achieved by analysing separately the two configurations illustrated in Fig.~\ref{fig:LumpModel}, where a body slides over a flat substrate. The problem is studied in a reference frame $(O;x,y,z)$ oriented as follows: the $x$-axis (longitudinal) is usually directed as the main direction of motion of the upper body, the $z$-axis (vertical) points downward, into the lower body, and the $y$-axis (lateral) is oriented so to have a right-handed system.

First, the situation depicted in Fig.~\ref{fig:LumpModel}(a) is considered: the upper body travels with rigid relative velocity $\mathbb{R}^2 \ni \bm{v}\ped{r} = [v_{\textnormal{r}x}\; v_{\textnormal{r}y}]^{\mathrm{T}}$ with respect to the rigid substrate (lower body). To the lower boundary of the upper body, deformable massless bristles protrude, whose deflection is denoted by $\mathbb{R}^2 \ni \bm{z} = [z_x\; z_y]^{\mathrm{T}}$ (not to be confounded with the vertical axis $z$). These bristles, which are interpreted as material particles in the paper, consistently with the classic theories of contact mechanics \cite{KalkerBook}, are assumed to have zero vertical dimension, and only undergo deformations along the $x$ and $y$ directions. The roots of the bristles, which are attached to the upper body, have rigid relative velocity $\bm{v}\ped{r}$ with respect to the substrate. On the other hand, the total sliding velocity between the tip of the bristle and the lower body reads
\begin{align}\label{eq:slidingS}
\bm{v}\ped{s}(\dot{\bm{z}},\bm{v}\ped{r}) = \bm{v}\ped{r}+ \dot{\bm{z}}.
\end{align}

\begin{figure}
\centering
\includegraphics[width=1\linewidth]{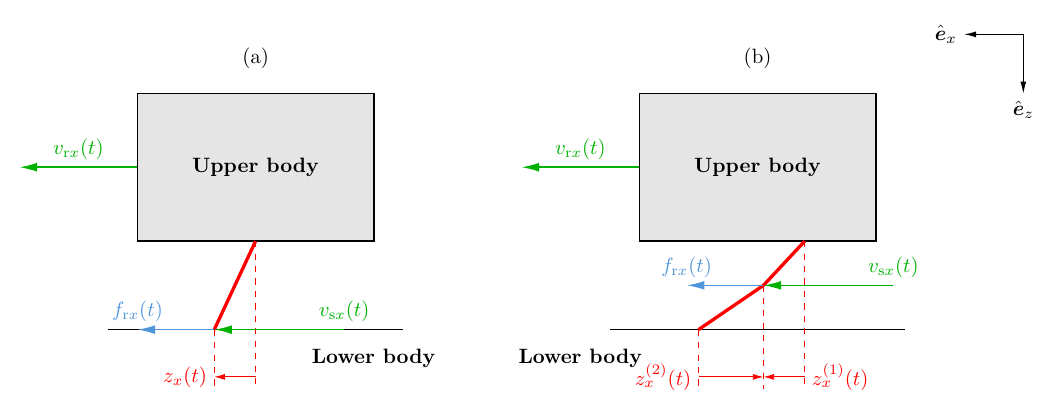} 
\caption{A schematic representation of the friction model: (a) configuration with a rigid substrate; (b) configuration with a deformable substrate. The problem is studied in a right-handed reference frame $(O;x,y,z)$ with unit vectors $(\hat{\bm{e}}_x, \hat{\bm{e}}_y, \hat{\bm{e}}_z)$.}
\label{fig:LumpModel}
\end{figure}
As illustrated in Fig.~\ref{fig:LumpModelForces}, during its sliding over the rigid substrate, the bristle generates a nondimensional force\footnote{Essentially, a force per unit of vertical load, which may be interpreted either as a friction potential or as a dynamic friction coefficient.}, $\mathbb{R}^2 \ni \bm{f} = [f_x\; f_y]^{\mathrm{T}}$, which, in the absence of inertial effects, counteracts the friction force per unit of normal load $p \in \mathbb{R}_{>0}$ acting on its tip, $\mathbb{R}^2 \ni \bm{f}\ped{r} = [f_{\textnormal{r}x}\; f_{\textnormal{r}y}]^{\mathrm{T}}$. A simple but rather general constitutive equation, based on a Kelvin-Voigt rheological representation of the bristle, is
\begin{align}\label{eq:rheol1}
\bm{f}(\dot{\bm{z}},\bm{z}) = \mathbf{\Sigma}_0\bm{z} + \mathbf{\Sigma}_1\dot{\bm{z}},
\end{align}
where the matrices $\mathbf{Sym}_2(\mathbb{R}) \ni \mathbf{\Sigma}_0 \succ \mathbf{0}$ and $\mathbf{Sym}_2(\mathbb{R}) \ni \mathbf{\Sigma}_1 \succeq \mathbf{0}$ collect the normalised micro-stiffness and micro-damping coefficients: 
\begin{subequations}\label{eq:Sigmas1}
\begin{align}
\mathbf{\Sigma}_0 & = \begin{bmatrix} \sigma_{0xx} & \sigma_{0xy} \\ \sigma_{0xy} & \sigma_{0yy} \end{bmatrix}, \\
\mathbf{\Sigma}_1 & = \begin{bmatrix} \sigma_{1xx} & \sigma_{1xy} \\ \sigma_{1xy} & \sigma_{1yy} \end{bmatrix}, 
\end{align}
\end{subequations}
or, in the case of diagonal matrices:
\begin{subequations}\label{eq:Sigmas2}
\begin{align}
\mathbf{\Sigma}_0 & = \begin{bmatrix} \sigma_{0x} & 0 \\ 0 & \sigma_{0y} \end{bmatrix}, \label{eq:Sigmas02}\\
\mathbf{\Sigma}_1 & = \begin{bmatrix} \sigma_{1x} & 0 \\ 0 & \sigma_{1y} \end{bmatrix}. 
\end{align}
\end{subequations}
In Eqs.~\eqref{eq:rheol1}-\eqref{eq:Sigmas2}, $\mathbf{\Sigma}_0$ and $\mathbf{\Sigma}_1$ have been assumed constant for simplicity, but the generalisation to nonlinear stiffness and damping terms is straightforward. It is also worth emphasising that, in contrast to the LuGre model, Eq.~\eqref{eq:rheol1} does not include viscous contributions, which will be instead incorporated correctly in the local model friction model detailed in Sect.~\ref{sect:frictionModel}. This strategy is coherent with the observations reported in several recent publications (see, for instance, \cite{Trib1,Rill}).

\begin{figure}
\centering
\includegraphics[width=0.5\linewidth]{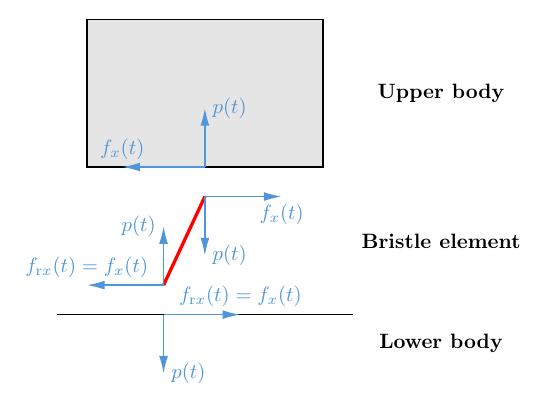} 
\caption{Free-body diagram of the bristle element in the $x$ and $z$-directions, along with its reaction on the upper and lower bodies in the absence of inertial effects.}
\label{fig:LumpModelForces}
\end{figure}

In the second scenario illustrated in Fig.~\ref{fig:LumpModel}(b), deformable bristles $\mathbb{R}^2 \ni \bm{z}\ap{(1)} = [z_x\ap{(1)}\; z_y\ap{(1)}]^{\mathrm{T}}$ and $\mathbb{R}^2 \ni \bm{z}\ap{(2)} = [z_x\ap{(2)}\; z_y\ap{(2)}]^{\mathrm{T}}$ protrude from both bodies (the superscripts $\ap{(1)}$ and $\ap{(2)}$ refer to body 1 and 2, respectively). In this case, by defining $\mathbb{R}^2 \ni \bm{z} \triangleq \bm{z}\ap{(1)}-\bm{z}\ap{(2)}$, the total sliding velocity between the tips of the bristles attached to the upper and lower body may be computed again as in Eq.~\eqref{eq:slidingS}. Furthermore, a similar relationship to Eq.~\eqref{eq:rheol1} may be postulated concerning each body, that is, 
\begin{subequations}\label{eq:fbSedndss}
\begin{align}
\bm{f}\ap{(1)}\bigl(\dot{\bm{z}}\ap{(1)},\bm{z}\ap{(1)}\bigr) & = \mathbf{\Sigma}_0\ap{(1)}\bm{z}\ap{(1)} + \mathbf{\Sigma}_1\ap{(1)}\dot{\bm{z}}\ap{(1)}, \\
\bm{f}\ap{(2)}\bigl(\dot{\bm{z}}\ap{(2)}, \bm{z}\ap{(2)}\bigr) & = \mathbf{\Sigma}_0\ap{(2)}\bm{z}\ap{(2)} + \mathbf{\Sigma}_1\ap{(2)}\dot{\bm{z}}\ap{(2)}.
\end{align}
\end{subequations}
Again, it may be assumed that $\mathbf{Sym}_2(\mathbb{R}) \ni \mathbf{\Sigma}_0\ap{(1)}, \mathbf{\Sigma}_0\ap{(2)} \succ \mathbf{0}$ and $\mathbf{Sym}_2(\mathbb{R}) \ni \mathbf{\Sigma}_1\ap{(1)}, \mathbf{\Sigma}_1\ap{(2)} \succeq \mathbf{0}$. Clearly, it must hold that
\begin{align}
\bm{f}\ap{(1)}\bigl(\dot{\bm{z}}\ap{(1)}\bm{z}\ap{(1)}\bigr)  = - \bm{f}\ap{(2)}\bigl(\dot{\bm{z}}\ap{(2)}, \bm{z}\ap{(2)}\bigr) .
\end{align}
Consequently, assuming $\mathbf{\Sigma}_0\ap{(1)} = \mathbf{\Sigma}_0\ap{(2)}$, $\mathbf{\Sigma}_1\ap{(1)} = \mathbf{\Sigma}_1\ap{(2)}$ yields
\begin{align}\label{eq:rheol2}
\begin{split}
\bm{f}(\dot{\bm{z}},\bm{z}) & \triangleq \bm{f}\ap{(1)}\bigl(\dot{\bm{z}}\ap{(1)}, \bm{z}\ap{(1)}\bigr)   \equiv -\bm{f}\ap{(2)}\bigl(\dot{\bm{z}}\ap{(2)}, \bm{z}\ap{(2)}\bigr)  \\
& \equiv \dfrac{1}{2}\Bigl[\bm{f}\ap{(1)}\bigl(\dot{\bm{z}}\ap{(1)}, \bm{z}\ap{(1)}\bigr)-\bm{f}\ap{(2)}\bigl(\dot{\bm{z}}\ap{(2)}, \bm{z}\ap{(2)}\bigr)\Bigr] = \mathbf{\Sigma}_0\bm{z} + \mathbf{\Sigma}_1\dot{\bm{z}},
\end{split}
\end{align}
where
\begin{align}
\mathbf{\Sigma}_0 \triangleq \dfrac{1}{2}\mathbf{\Sigma}_0\ap{(1)}\equiv \dfrac{1}{2}\mathbf{\Sigma}_0\ap{(2)}, \\
\mathbf{\Sigma}_1 \triangleq \dfrac{1}{2}\mathbf{\Sigma}_1\ap{(1)} \equiv \dfrac{1}{2}\mathbf{\Sigma}_1\ap{(2)}.
\end{align}
Compared to~\eqref{eq:fbSedndss}, Eq.~\eqref{eq:rheol2} has exactly the same structure as Eq.~\eqref{eq:rheol1}, and describes the bristle force solely in terms of the relative deflection and velocity $\bm{z}$ and $\dot{\bm{z}}$. Therefore, it may be concluded that, when the stiffness and damping matrices $\mathbf{\Sigma}_0$ and $\mathbf{\Sigma}_1$ are constant, Eq.~\eqref{eq:rheol1} suffices to describe both the situations depicted in Fig.~\ref{fig:LumpModel}. In rolling contact, this is generally true for bodies that are either incommensurably or comparably rigid, where, to a good approximation, the expressions for the relative sliding and rigid velocities are independent of the assumed constitutive model \cite{KinematicsMio}. In intermediate cases where the contacting bodies undergo different amounts of deformation, this no longer holds, and the kinematics becomes coupled with the material's rheological equations, even in the simplest case of pure elastic behaviours \cite{KinematicsMio}. Such complex situations are deliberately disregarded in this paper, but will be addressed in future works.

\subsubsection{Local friction model and friction dissipation}\label{sect:frictionModel}
In addition to a rheological representation of the bristle element, it is necessary to specify a local model for the friction force as a function of the sliding velocity, $\bm{f}\ped{r}(\bm{v}\ped{s})$. Whereas an excellent and exhaustive overview of analytical friction models may be found in \cite{Antali}, the approach followed in this paper is specifically inspired by those of \cite{Sorine,Tsiotras3}, and starts with important considerations about the \emph{passivity} of friction. The fundamental idea is that the normalised friction force generated during the reciprocal sliding motion between the contacting bodies should maximise the \emph{dissipation rate} $-\bm{f}\ped{r}^{\mathrm{T}}\bm{v}\ped{s}$. To this end, the following optimisation problem is formulated: 
\begin{align}\label{eq:Problem}
\begin{split}
\bm{f}\ped{r}^\star & = \argmax_{\bm{f}\ped{r} \in \mathcal{C}(\bm{v}\ped{s}(t))} -\bm{f}\ped{r}^{\mathrm{T}}\bm{v}\ped{s}(t), \quad \bm{v}\ped{s}(t) \in \mathbb{R}^2, \; t \in [0,T],
\end{split}
\end{align}
where $\mathcal{C}(\bm{v}\ped{s})$ denotes the set of admissible normalised friction forces, depending on the sliding velocity:
\begin{align}
\mathcal{C}(\bm{v}\ped{s}) & \triangleq \biggl\{\bm{f}\ped{r} \in \mathbb{R}^2 \mathrel{\bigg|} \norm{\mathbf{M}^{-1}(\bm{v}\ped{s})\bm{f}\ped{r}}_2 \leq 1\biggr\},
\end{align}
and
\begin{align}\label{eq;matrixM}
\mathbf{M}(\bm{v}\ped{s}) = \begin{bmatrix} \mu_{xx}(\bm{v}\ped{s}) &  \mu_{xy}(\bm{v}\ped{s}) \\ \mu_{xy}(\bm{v}\ped{s})  & \mu_{yy}(\bm{v}\ped{s})\end{bmatrix}
\end{align}
is a symmetric, positive definite matrix, i.e., $\mathbf{Sym}_2(\mathbb{R}) \ni \mathbf{M}(\bm{y}) \succ \mathbf{0}$ for all $\bm{y} \in \mathbb{R}^2$. In this manuscript, it is generally assumed that $\mathbf{M} \in C^0(\mathbb{R}^2;\mathbf{Sym}_2(\mathbb{R}))$. For many practical applications, and especially in isotropic conditions, $\mathbf{M}(\bm{v}\ped{s})$ is often postulated to be of the form $\mathbf{M}(\bm{v}\ped{s}) = \mu(\bm{v}\ped{s})\mathbf{I}_2$, where $\mu : \mathbb{R}^2 \to [\mu\ped{min},\infty)$, with $ \mu\ped{min} \in \mathbb{R}_{>0}$, provides an analytical expression for the friction coefficient. A very common expression for $\mu(\bm{v}\ped{s})$ is, for instance,
\begin{align}\label{eq:muExample}
\mu(\bm{v}\ped{s}) = \mu\ped{d} + (\mu\ped{s}-\mu\ped{d})\exp\Biggl(-\biggl(\dfrac{\norm{\bm{v}\ped{s}}_2}{v\ped{S}}\biggr)^{\delta\ped{S}}\Biggr)+ \mu\ped{v}(\bm{v}\ped{s}),
\end{align}
where $\mu\ped{s},\mu\ped{d} \in \mathbb{R}_{>0}$ denote the static and dynamic friction coefficients, $v\ped{S} \in \mathbb{R}_{>0}$ indicates the Stribeck velocity, $\delta\ped{S} \in \mathbb{R}_{\geq 0}$ the Stribeck exponent, and $\mu\ped{v} : \mathbb{R}^2 \to \mathbb{R}_{\geq 0}$ captures the viscous friction.

In any case, the solution to the problem described by Eqs.~\eqref{eq:Problem}-\eqref{eq;matrixM} was obtained in \cite{Sorine,Tsiotras3}, and is given as in Theorem~\ref{lemma:friction} below.
\begin{theorem}\label{lemma:friction}
The solution to Eqs.~\eqref{eq:Problem}-\eqref{eq;matrixM} is given by
\begin{align}\label{eq:fdbddsOrg}
\bm{f}\ped{r}^\star\bigl(\bm{v}\ped{s}(t)\bigr) = -\dfrac{\mathbf{M}^2\bigl(\bm{v}\ped{s}(t)\bigr)\bm{v}\ped{s}(t)}{\norm{\mathbf{M}\bigl(\bm{v}\ped{s}(t)\bigr)\bm{v}\ped{s}(t)}_2}, \quad \bm{v}\ped{s}(t) \in \mathbb{R}^2\setminus \{\bm{0}\}, \; t\in [0,T],
\end{align}
with
\begin{align}\label{eq:dissOpt}
-\bm{f}\ped{r}^\star\bigl(\bm{v}\ped{s}(t)\bigr)^{\mathrm{T}}\bm{v}\ped{s}(t) = \norm{\mathbf{M}\bigl(\bm{v}\ped{s}(t)\bigr)\bm{v}\ped{s}(t)}_2, \quad t\in [0,T].
\end{align}
\begin{proof}
See the proof of Theorem 2 in \cite{Tsiotras3}.
\end{proof}
\end{theorem}
Equation~\eqref{eq:fdbddsOrg} provides some preliminary but useful indications to generalise the FrBD model to two-dimensional contact in a physically consistent manner. Its interpretation is that the generated friction force, as a function of the sliding velocity, should maximise the dissipated power at the interface between the contacting bodies. However, it is worth observing that Eq.~\eqref{eq:fdbddsOrg} is not defined for $\bm{v}\ped{s} = \bm{0}$. For this reason, but also to facilitate the implementation of numerical algorithms \cite{Rill}, it may be convenient to replace Eq.~\eqref{eq:fdbddsOrg} with the following formula:
 \begin{align}\label{eq:frModified}
\bm{f}\ped{r}^\star(\bm{v}\ped{s}) =- \dfrac{\mathbf{M}^2(\bm{v}\ped{s})\bm{v}\ped{s}}{\norm{\mathbf{M}(\bm{v}\ped{s})\bm{v}\ped{s}}_{2,\varepsilon}},
\end{align}
where $\varepsilon \in \mathbb{R}_{\geq 0}$ represents a regularisation parameter, and $\norm{\cdot}_{2,\varepsilon} \in C^0(\mathbb{R}^2;\mathbb{R}_{\geq 0})$ is a regularisation of the Euclidean norm $\norm{\cdot}_2$ for $\varepsilon\in \mathbb{R}_{>0}$, often converging uniformly to $\norm{\cdot}_2$ in $C^0(\mathbb{R}^2;\mathbb{R}_{\geq 0})$ for $\varepsilon \to 0$ (e.g., $\norm{\bm{y}}_{2,\varepsilon }= \sqrt{\norm{\bm{y}}_2^2 +\varepsilon}$), and with $\norm{\cdot}_{2,\varepsilon} \in C^1(\mathbb{R}^2;\mathbb{R}_{\geq 0})$ for $\varepsilon \in \mathbb{R}_{>0}$. 

Before moving to the next Sect.~\ref{sect:DynamicDer}, an important consideration is formalised in Remark~\ref{remark:1} below.

\begin{remark}\label{remark:1}
If $\mathbf{M}(\bm{v}\ped{s}) = \mu(\bm{v}\ped{s})\mathbf{I}_2$, as it happens in the isotropic case, then Eq.~\eqref{eq:fdbddsOrg} reduces to
\begin{align}\label{eq:fdbddsOrg0}
\bm{f}\ped{r}^\star(\bm{v}\ped{s}) = -\mu(\bm{v}\ped{s})\dfrac{\bm{v}\ped{s}}{\norm{\bm{v}\ped{s}}_{2}}, \quad \bm{v}\ped{s} \in \mathbb{R}^2\setminus\{\bm{0}\}.
\end{align}
In this case, Eq.~\eqref{eq:frModified} may be modified accordingly as
\begin{align}\label{eq:frIsotrEps}
\bm{f}\ped{r}^\star(\bm{v}\ped{s}) = -\mu(\bm{v}\ped{s})\dfrac{\bm{v}\ped{s}}{\norm{\bm{v}\ped{s}}_{2,\varepsilon}},
\end{align}
which is the simplest two-dimensional counterpart to Eq. (4) in \cite{FrBD}. In this context, it is worth emphasising that, for $\varepsilon \in \mathbb{R}_{>0}$, Eq.~\eqref{eq:frIsotrEps} does not follow directly from~\eqref{eq:frModified}, and must be obtained by regularising~\eqref{eq:fdbddsOrg} \emph{a posteriori}. The norm $\norm{\cdot}_2$ in Eq.~\eqref{eq:muExample} may be similarly regularised.
\end{remark}

\subsection{Dynamic friction model derivation}\label{sect:DynamicDer}
Section~\ref{sect:rhelANdFr} introduced the two main modelling ingredients needed to extend the dynamic FrBD friction model to the two-dimensional contact case. Theorem~\ref{thm:Theorem1} below provides the final, analytical building block.

\begin{theorem}[Edwards \cite{Edwards}]\label{thm:Theorem1}
Suppose that the mapping $\bm{H} : \mathbb{R}^{m+n}\to \mathbb{R}^n$ is $C^1$ in a neighbourhood of a point $(\bm{x}^\star,\bm{y}^\star)$, where $\bm{H}(\bm{x}^\star,\bm{y}^\star) = \bm{0}$. If the Jacobian matrix $\nabla_{\bm{y}}\bm{H}(\bm{x}^\star,\bm{y}^\star)^{\mathrm{T}}$ is nonsingular, there exist a neighbourhood $\mathcal{X}$ of $\bm{x}^\star$ in $\mathbb{R}^m$, a neighbourhood $\mathcal{Y}$ of $(\bm{x}^\star,\bm{y}^\star)$ in $\mathbb{R}^{m+n}$, and a mapping $\bm{h} \in C^1(\mathcal{X};\mathbb{R}^n)$ such that $\bm{y} = \bm{h}(\bm{x})$ solves the equation $\bm{H}(\bm{y},\bm{x}) = \bm{0}$ in $\mathcal{Y}$. 
In particular, the implicitly defined mapping $\bm{h}(\cdot)$ is the limit of the sequence $\{\bm{h}_k\}_{ k\in \mathbb{N}_0}^\infty$ of the successive approximations inductively defined by
\begin{subequations}
\begin{align}
\bm{h}_{k+1}(\bm{x}) & = \bm{h}_k(\bm{x}) - \nabla_{\bm{y}}\bm{H}(\bm{x}^\star,\bm{y}^\star)^{-\mathrm{T}}\bm{H}\bigl(\bm{x},\bm{h}_k(\bm{x})\bigr), \\
 \bm{h}_0(\bm{x}) & = \bm{y}^\star,
\end{align}
\end{subequations}
for $\bm{x} \in \mathcal{X}$.
\end{theorem}
Starting with Eqs.~\eqref{eq:rheol1} and~\eqref{eq:frModified}, Theorem~\eqref{thm:Theorem1} may be sapiently invoked to generalise the FrBD model to the two-dimensional contact case. The fundamental observation is that, in the absence of inertial effects, the net bristle force $\bm{f}(\dot{\bm{z}},\bm{z})$ should be equal to the frictional one $\bm{f}\ped{r}(\bm{v}\ped{s}(\dot{\bm{z}},\bm{v}\ped{r}))$ generated at the bodies' interface, that is, $\bm{f}(\dot{\bm{z}},\bm{z}) = \bm{f}\ped{r}(\bm{v}\ped{s}(\dot{\bm{z}},\bm{v}\ped{r}))$. Therefore, adopting the notation of Theorem~\ref{thm:Theorem1} and equating Eqs.~\eqref{eq:rheol1} and~\eqref{eq:frModified} gives
\begin{align}\label{eq:H}
\bm{H}(\dot{\bm{z}}, \bm{z}, \bm{v}\ped{r}) = \bm{f}(\dot{\bm{z}}, \bm{z}) - \bm{f}\ped{r}\bigl(\bm{v}\ped{s}(\dot{\bm{z}},\bm{v}\ped{r})\bigr) = \bm{0}, \quad t \in (0,T),
\end{align}
which is an implicit nonlinear ODE for the bristle dynamics $\bm{z}(t)$. In the sliding regime, where $\norm{\dot{\bm{z}}}_2 \ll \norm{\bm{v}\ped{r}}_2$, and for sufficiently smooth $\bm{H}(\cdot,\cdot,\cdot)$, Eq.~\eqref{eq:H} may be approximated by invoking Theorem~\ref{thm:Theorem1} with $\bm{x} = (\bm{z},\bm{v}\ped{r})$ and $\bm{y} = \dot{\bm{z}}$, yielding
\begin{align}\label{eq:zk01}
\dot{\bm{z}}_{k+1} = \dot{\bm{z}}_k-\nabla_{\bm{\dot{z}}}\bm{H}(\dot{\bm{z}}^\star,\bm{z}^\star,\bm{v}\ped{r}^\star)^{-\mathrm{T}}\bm{H}(\dot{\bm{z}}_k, \bm{z}, \bm{v}\ped{r}), \quad k \in \mathbb{N}_0.
\end{align}
In turn, specifying $\bm{f}(\dot{\bm{z}},\bm{z})$ and $\bm{f}\ped{r}(\bm{v}\ped{s}) = \bm{f}\ped{r}^\star(\bm{v}\ped{s})$ as in Eqs.~\eqref{eq:rheol1} and~\eqref{eq:frModified} leads to
\begin{align}\label{eq:nablaH}
\begin{split}
\nabla_{\bm{\dot{z}}}\bm{H}(\dot{\bm{z}},\bm{z},\bm{v}\ped{r})^{\mathrm{T}} & = \nabla_{\dot{\bm{z}}}\bm{f}(\dot{\bm{z}},\bm{z})^{\mathrm{T}} - \nabla_{\dot{\bm{z}}}\bm{v}\ped{s}(\dot{\bm{z}},\bm{v}\ped{r})^{\mathrm{T}}\nabla_{\bm{v}\ped{s}}\bm{f}\ped{r}\bigl(\bm{v}\ped{s}(\dot{\bm{z}},\bm{v}\ped{r})\bigr)^{\mathrm{T}} \\
 & = \nabla_{\dot{\bm{z}}}\bm{f}(\dot{\bm{z}},\bm{z})^{\mathrm{T}} - \nabla_{\bm{v}\ped{s}}\bm{f}\ped{r}\bigl(\bm{v}\ped{s}(\dot{\bm{z}},\bm{v}\ped{r})\bigr)^{\mathrm{T}}  \approx \nabla_{\dot{\bm{z}}}\bm{f}(\dot{\bm{z}},\bm{z})^{\mathrm{T}} + \dfrac{\mathbf{M}^2\bigl(\bm{v}\ped{s}(\dot{\bm{z}},\bm{v}\ped{r})\bigr)}{\norm{\mathbf{M}\bigl(\bm{v}\ped{s}(\dot{\bm{z}},\bm{v}\ped{r})\bigr)\bm{v}\ped{s}(\dot{\bm{z}},\bm{v}\ped{r})}_{2,\varepsilon}},
\end{split}
\end{align}
where the approximation committed in the last line is informally justified by the fact that the friction force $\bm{f}\ped{r}(\bm{v}\ped{s})$ may not be defined for $\bm{v}\ped{s} = \bm{0}$ (for instance, when $\varepsilon = 0$ in~\eqref{eq:frModified}). Moreover, owing to the assumption $\mathbf{\Sigma}_0 \succ \mathbf{0}$, it may be easily inferred that, in steady-state conditions ($\dot{\bm{z}} = \bm{0}$), there exist a unique solution $\bm{z}^\star = \bm{z}^\star(\bm{v}\ped{r})$ to Eq.~\eqref{eq:H} for all $\bm{v}\ped{r}\in \mathbb{R}^2$. Consequently, combining Eq.~\eqref{eq:zk01} and~\eqref{eq:nablaH} with $(\dot{\bm{z}}^\star, \bm{z}^\star, \bm{v}\ped{r}^\star) = (\bm{0},\bm{z}^\star(\bm{v}\ped{r}),\bm{v}\ped{r})$, and truncating Eq.~\eqref{eq:zk01} at $k=1$ provides, for an initial guess $\dot{\bm{z}}_0 = \bm{0}$, 
\begin{subequations}\label{eq:ODEModel}
\begin{align}
& \dot{\bm{z}}(t) = -\mathbf{G}^{-1}\bigl(\bm{v}\ped{r}(t)\bigr)\biggl[\mathbf{\Sigma}_0\norm{\mathbf{M}\bigl(\bm{v}\ped{r}(t)\bigr)\bm{v}\ped{r}(t)}_{2,\varepsilon}\bm{z}(t)+\mathbf{M}^2\bigl(\bm{v}\ped{r}(t)\bigr)\bm{v}\ped{r}(t)\biggr], \quad t \in (0,T), \\
& \bm{z}(0) = \bm{z}_0.
\end{align}
\end{subequations}
with $\mathbf{G} : \mathbb{R}^2 \to \mathbf{Sym}_2(\mathbb{R})$ according to
\begin{align}\label{eq:MG}
\mathbf{G}(\bm{v}\ped{r}) = \mathbf{\Sigma}_1\norm{\mathbf{M}(\bm{v}\ped{r})\bm{v}\ped{r}}_{2,\varepsilon} + \mathbf{M}^2(\bm{v}\ped{r}).
\end{align}
Equations~\eqref{eq:ODEModel} and~\eqref{eq:MG} describe the \emph{two-dimensional FrBD model}. In the sequel, it will be particularised to describe various distributed contact phenomena, including rolling and spinning. Before moving to Sect.~\ref{sect:models}, it is beneficial to draw some preliminary considerations about the mathematical structure of the FrBD model, as well as on its dissipative nature. Some preliminary intuition may be gained by first considering the stationary solution to Eq.~\eqref{eq:ODEModel}, which reads
\begin{align}\label{eq:zStat}
\bm{z}(\bm{v}\ped{r}) = -\mathbf{\Sigma}_0^{-1}\dfrac{\mathbf{M}^2(\bm{v}\ped{r}\bigr)\bm{v}\ped{r}}{\norm{\mathbf{M}(\bm{v}\ped{r})\bm{v}\ped{r}}_{2,\varepsilon}}.
\end{align}
The corresponding stationary bristle force, obtained by setting $\dot{\bm{z}} = \bm{0}$ in Eq.~\eqref{eq:rheol1}, reads
\begin{align}\label{eq:fStat}
\bm{f}(\bm{v}\ped{r}) = \bm{f}\ped{r}(\bm{v}\ped{r})  = -\dfrac{\mathbf{M}^2(\bm{v}\ped{r}\bigr)\bm{v}\ped{r}}{\norm{\mathbf{M}(\bm{v}\ped{r})\bm{v}\ped{r}}_{2,\varepsilon}}.
\end{align}
Equation~\eqref{eq:fStat} shows that, under steady-state conditions ($\dot{\bm{z}}(t) = \bm{0}$), the FrBD model can reproduce arbitrary static friction characteristics. It is worth noting that, unlike the LuGre model, a minus sign appears in Eqs.~\eqref{eq:zStat} and~\eqref{eq:fStat}. This difference is consistent with the adopted sign convention and with the formulation proposed in \cite{Antali}.
More generally, the ODE~\eqref{eq:ODEModel} governs the evolution of the bristle dynamics in the vicinity of sliding, that is, for relatively small deformation velocities $\dot{\bm{z}} \approx \bm{0}$ and $\bm{v}\ped{s}(\dot{\bm{z}},\bm{v}\ped{r}) \approx \bm{v}\ped{r}$.

The second observation pertains the passivity properties of the model. More specifically, regarding $\bm{v}\ped{s}(\dot{\bm{z}},\bm{v}\ped{r})$ and $\bm{f}(\dot{\bm{z}},\bm{z})$ as the input and output, respectively, and combining Eqs.~\eqref{eq:rheol1} and~\eqref{eq:dissOpt} yields
\begin{align}
\begin{split}
-\bm{f}\ped{r}^\star\bigl(\bm{v}\ped{s}(\dot{\bm{z}},\bm{v}\ped{r})\bigr)^{\mathrm{T}}\bm{v}\ped{s}(\dot{\bm{z}},\bm{v}\ped{r}) =- \bm{f}(\dot{\bm{z}},\bm{z})^{\mathrm{T}}\bm{v}\ped{s}(\dot{\bm{z}},\bm{v}\ped{r}) = -\bm{z}^{\mathrm{T}}\mathbf{\Sigma}_0\bm{v}\ped{s}(\dot{\bm{z}},\bm{v}\ped{r}) - \dot{\bm{z}}^{\mathrm{T}}\mathbf{\Sigma}_1\bm{v}\ped{s}(\dot{\bm{z}},\bm{v}\ped{r})  \geq 0.
\end{split}
\end{align}
Recalling Eq.~\eqref{eq:slidingS} and rearranging the above inequality also gives
\begin{align}\label{eq:dissipationRateRigid}
-\bm{f}(\dot{\bm{z}},\bm{z})^{\mathrm{T}}\bm{v}\ped{r} = - \bm{z}^{\mathrm{T}}\mathbf{\Sigma}_0\bm{v}\ped{r} - \dot{\bm{z}}^{\mathrm{T}}\mathbf{\Sigma}_1\bm{v}\ped{r} \geq \bm{z}^{\mathrm{T}}\mathbf{\Sigma}_0\dot{\bm{z}} + \dot{\bm{z}}^{\mathrm{T}}\mathbf{\Sigma}_1\dot{\bm{z}}.
\end{align}
Since $\mathbf{\Sigma}_1 \succeq \mathbf{0}$ by assumption, the last term appearing in Eq.~\eqref{eq:dissipationRateRigid} is always nonnegative. However, the sign of $\bm{z}^{\mathrm{T}}\mathbf{\Sigma}_0\dot{\bm{z}}$ is generally undetermined. Therefore, Eq.~\eqref{eq:dissipationRateRigid} is eloquent in telling that, if $\bm{v}\ped{r}$ and $\bm{f}(\dot{\bm{z}},\bm{z})$ are regarded as the input and output of the model, the \emph{rigid dissipation rate} $-\bm{f}(\dot{\bm{z}},\bm{z})^{\mathrm{T}}\bm{v}\ped{r}$ need not to be nonnegative. In other words, when $\bm{v}\ped{r}$ is regarded as an external input, passivity does not automatically hold. However, the dynamic FrBD formulation described by Eqs.~\eqref{eq:ODEModel} and~\eqref{eq:MG} discards Eq.~\eqref{eq:fdbddsOrg}. Consequently, the friction force is fully determined by the rheological model~\eqref{eq:rheol1}. In this setting, the notion of sliding velocity is eliminated entirely, and passivity and dissipativity must instead be assessed with the rigid velocity $\bm{v}\ped{r}$ interpreted as an input. As shown in Sect.~\ref{sect:diss}, the proposed formulation maintains passivity under virtually all circumstances. This contrasts with the LuGre formulation, which requires introducing velocity-dependent damping coefficients to enforce passivity. Nonetheless, within the dynamical FrBD framework, these properties seem to be mathematical consequences of the sufficiently accurate approximation of Eq.~\eqref{eq:H} obtained via Theorem~\ref{thm:Theorem1}, rather than strict physical requirements demanding strenuous elucubration.

Finally, it may be noted that, by disregarding the damping contribution in Eq.~\eqref{eq:fxs}, the FrBD model reduces to the two-dimensional formulation of the LuGre model presented in \cite{Tsiotras3}. For sliding systems, the term $\mathbf{\Sigma}_1\norm{\mathbf{M}(\bm{v}\ped{r})\bm{v}\ped{r}}_{2,\varepsilon} $ may not be negligible in general, whereas, in the context of rolling contact phenomena, the matrix $\mathbf{\Sigma}_1$ is often observed to be small \cite{TsiotrasConf,Tsiotras1,Tsiotras3,Deur0,Deur1,Deur2}.

\section{Distributed rolling contact models}\label{sect:models}
The present section introduces three rolling contact models with different degrees of complexity that can be derived from Eq.~\eqref{eq:ODEModel}. Specifically, the next Sect.~\ref{sect:gener} discusses some generalities, whereas the three formulations are detailed in Sects.~\ref{sect:modelRollSimple},~\ref{sect:semilinear}, and~\ref{sect:LinearLargeSpin}.

%More specifically,~\ref{sect:modelRollSimple} is dedicated to rolling contact phenomena, whereas the case of pure spinning is discussed separately in Sect.~\ref{sect:spinnSimple}.

\subsection{Generalities}\label{sect:gener}
Before generalising the novel FrBD model to two-dimensional rolling contact, three important aspects that deserve a proper discussion concern the formulation of Eq.~\eqref{eq:ODEModel} as a PDE, following the Eulerian approach, the prescription of appropriate \emph{boundary conditions} (BCs), and the calculation of the tangential forces and vertical moment. These are addressed respectively in Sects.~\ref{sect:eulerian},~\ref{sect:BC}, and~\ref{ect:forces}.

\subsubsection{Eulerian approach and transport velocity}\label{sect:eulerian}
As illustrated in Fig.~\ref{fig:RollingBodies}, the rolling contact problem between two bodies is typically studied in a contact-fixed reference frame $(O;x,y,z)$, with the $x$-axis (longitudinal) oriented along the main rolling direction, the $z$-axis (vertical) pointing into one of the two bodies (often the lower one), and the $y$-axis (lateral) defined to complete a right-handed coordinate system. The origin $O$ coincides with the centroid of the (apparent), possibly time-varying, contact area $\mathscr{C}(t)$, which is often supposed to be independent of the tangential (longitudinal and lateral) interactions between the two bodies, and solely determined by the normal contact configuration \cite{KalkerBook}. Both the situations of reciprocal rolling, as depicted in Fig.~\ref{fig:RollingBodies}(a), and translational and rolling contact, as shown in Fig.~\ref{fig:RollingBodies}(b), may be considered: in the first case, the shape of the contact area may vary over time, but the origin $O$ is typically fixed; in the second scenario, the reference frame moves together with one of the two bodies (usually, the upper one).
\begin{figure}
\centering
\includegraphics[width=1\linewidth]{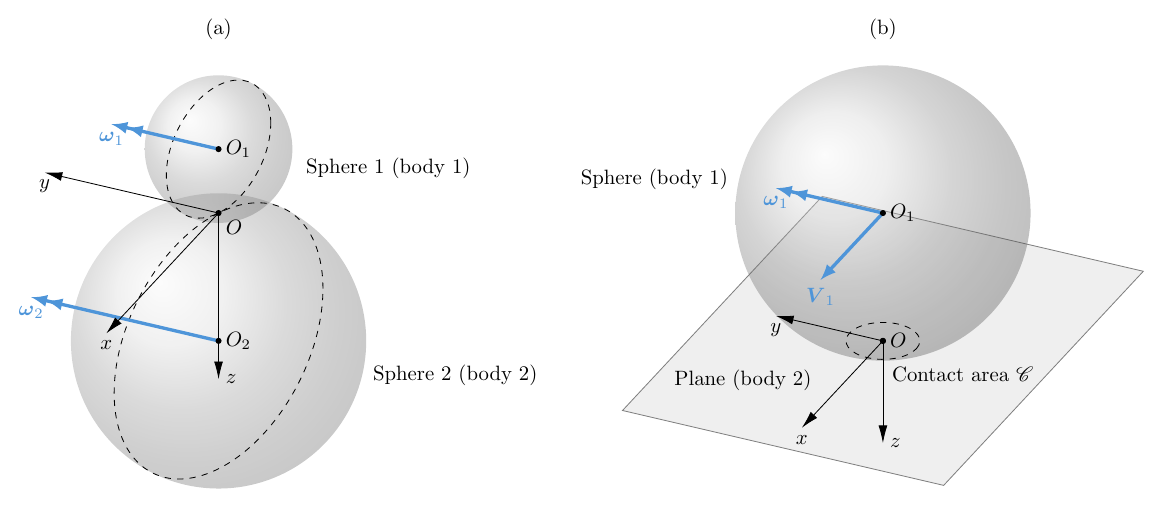} 
\caption{Rolling contact problem between: (a) two spheres with angular velocities $\bm{\omega}_1, \bm{\omega}_2 \in \mathbb{R}^3$; (b) a sphere translating and rolling over a stationary plane, where $\bm{V}_1\in \mathbb{R}^3$ denotes the translational velocity of its centre, and $\bm{\omega}_1\in \mathbb{R}^3$ its rolling velocity.}
\label{fig:RollingBodies}
\end{figure}
In this paper, the contact area is assumed to be a regular closed, compact subset of $\mathbb{R}^2$, that is, $\mathscr{C}(t)\subset \mathbb{R}^2$, with boundary $\partial \mathscr{C}(t)$ and interior $\mathring{\mathscr{C}}(t)$. Inside the contact area, the displacement of a bristle element will depend on its position, that is, $\bm{z}(t) = \bm{z}(\bm{x},t) = [z_x(\bm{x},t)\; z_y(\bm{x},t)]^{\mathrm{T}}$, with $\bm{x} \in \mathscr{C}(t)$. According to the Eulerian approach, the total time derivative appearing in Eq.~\eqref{eq:ODEModel} becomes
\begin{align}\label{eq:Eulerian}
\dot{\bm{z}}(\bm{x},t) = \dod{\bm{z}(\bm{x},t)}{t} = \dpd{\bm{z}(\bm{x},t)}{t} + \bigl(\bm{V}(\bm{x},t)\cdot\nabla_{\bm{x}}\bigr)\bm{z}(\bm{x},t),
\end{align}
where $\mathbb{R}^2 \ni \bm{V}(\bm{x},t) = [V_x(\bm{x},t) \; V_y(\bm{x},t)]^{\mathrm{T}}$ denotes the transport velocity, and $\mathbb{R}^2 \ni \nabla_{\bm{x}} \triangleq [\pd{}{x}\; \pd{}{y}]^{\mathrm{T}}$ is the tangential gradient. Moreover, considering the \emph{rolling speed} $[V\ped{min}, V\ped{max}] \ni V\ped{r}(t) \triangleq \norm{\bm{V}(\bm{0},t)}_2$, with $0 < V\ped{min} \leq V\ped{max}$, it is customary to define the \emph{travelled distance} $\mathbb{R}_{\geq 0} \ni s \triangleq \int_0^t V\ped{r}(t^\prime) \dif t^\prime$, so that Eq.~\eqref{eq:ODEModel} may be recast as
\begin{subequations}\label{eq:PDenoMOdel09}
\begin{align}\label{eq:PDenoMOdel}
& \dpd{\bm{z}(\bm{x},s)}{s} + \bigl(\bar{\bm{V}}(\bm{x},s)\cdot\nabla_{\bm{x}}\bigr)\bm{z}(\bm{x},s) = \mathbf{\Sigma}\bigl(\bar{\bm{v}}\ped{r}(\bm{x},s),s\bigr)\bm{z}(\bm{x},s) + \bm{h}\bigl(\bar{\bm{v}}\ped{r}(\bm{x},s),s\bigr), \quad \bm{x} \in \mathring{\mathscr{C}}(s), \; s \in (0,S), \\
& \bm{z}(\bm{x},0) = \bm{z}_0(\bm{x}), \quad \bm{x}\in \mathring{\mathscr{C}}_0,\label{eq:PDenoMOdelIC}
\end{align}
\end{subequations}
where $\mathring{\mathscr{C}}_0 \triangleq \mathring{\mathscr{C}}(0)$, $\mathbb{R}_{>0}\ni S \triangleq \int_0^T V\ped{r}(t)\dif t$, $\mathbb{R}^2 \ni \bar{\bm{V}}(\bm{x},s) = [\bar{V}_x(\bm{x},s)\; \bar{V}_y(\bm{x},s)]^{\mathrm{T}} \triangleq \bm{V}(\bm{x},s)/V\ped{r}(s)$, $\mathbb{R}^2 \ni \bar{\bm{v}}\ped{r}(\bm{x},s) = [\bar{v}_{\textnormal{r}x}(\bm{x},s)\; \bar{v}_{\textnormal{r}y}(\bm{x},s)]^{\mathrm{T}} \triangleq \bm{v}\ped{r}(\bm{x},s)/V\ped{r}(s)$, and the functions $\mathbf{\Sigma} : \mathbb{R}^2\times\mathbb{R}_{\geq 0} \to \mathbf{M}_2(\mathbb{R})$ and $\bm{h}: \mathbb{R}^2\times\mathbb{R}_{\geq 0} \to \mathbb{R}^2$ are given by
\begin{subequations}\label{eq:Hfunct}
\begin{align}
\mathbf{\Sigma}(\bar{\bm{v}}\ped{r},s) & \triangleq -\dfrac{1}{V\ped{r}(s)}\mathbf{G}^{-1}\bigl(V\ped{r}(s)\bar{\bm{v}}\ped{r}\bigr)\mathbf{\Sigma}_0\norm{\mathbf{M}\bigl(V\ped{r}(s)\bar{\bm{v}}\ped{r}\bigr)V\ped{r}(s)\bar{\bm{v}}\ped{r}}_{2,\varepsilon}, \\
\bm{h}(\bar{\bm{v}}\ped{r},s) & \triangleq \mathbf{H}(\bar{\bm{v}}\ped{r},s)\bar{\bm{v}}\ped{r},
\end{align}
\end{subequations}
with $\mathbf{H} : \mathbb{R}^2 \times \mathbb{R}_{\geq 0} \to \mathbf{M}_2(\mathbb{R})$ reading 
\begin{align}\label{eq:matH}
\mathbf{H}(\bar{\bm{v}}\ped{r},s) \triangleq - \mathbf{G}^{-1}\bigl(V\ped{r}(s)\bar{\bm{v}}\ped{r}\bigr)\mathbf{M}^2\bigl(V\ped{r}(s)\bar{\bm{v}}\ped{r}\bigr).
\end{align}
Equation~\eqref{eq:PDenoMOdel} is now a PDE, and more specifically a \emph{transport equation}. In addition to the \emph{initial condition} (IC)~\eqref{eq:PDenoMOdelIC}, it needs to be supplemented with a BC, as discussed next in Sect.~\ref{sect:BC}.

\subsubsection{Boundary condition (BC)}\label{sect:BC}
In order to formulate an appropriate BC for the PDE~\eqref{eq:PDenoMOdel}, the following sets are introduced:
\begin{subequations}\label{eq:Inflow0Outflow}
\begin{align}
\mathscr{L}(s) &\triangleq \partial \mathscr{C}_{-}(s)= \Bigl\{\bm{x}\in \partial \mathscr{C}(s) \mathrel{\Big|}\bigl[\bar{\bm{V}}(\bm{x},s)-\bar{\bm{V}}_{\partial \mathscr{C}}(\bm{x},s)\bigr] \cdot \hat{\bm{n}}_{\partial \mathscr{C}}(\bm{x},s) < 0\Bigr\}, \\
\mathscr{N}(s) &\triangleq \partial \mathscr{C}_{0}(s) = \Bigl\{\bm{x}\in \partial \mathscr{C}(s) \mathrel{\Big|} \bigl[\bar{\bm{V}}(\bm{x},s)-\bar{\bm{V}}_{\partial \mathscr{C}}(\bm{x},s)\bigr] \cdot \hat{\bm{n}}_{\partial \mathscr{C}}(\bm{x},s) = 0\Bigr\}, \\
\mathscr{T}(s) &\triangleq \partial \mathscr{C}_{+}(s)= \Bigl\{\bm{x}\in \partial \mathscr{C}(s) \mathrel{\Big|} \bigl[\bar{\bm{V}}(\bm{x},s)-\bar{\bm{V}}_{\partial \mathscr{C}}(\bm{x},s)\bigr] \cdot \hat{\bm{n}}_{\partial \mathscr{C}}(\bm{x},s) > 0\Bigr\}, 
\end{align}
\end{subequations}
being $\hat{\bm{n}}_{\partial \mathscr{C}}(\bm{x},s)\in \mathbb{R}^2$ the outward unit normal to $\partial \mathscr{C}(s)$, and $\bar{\bm{V}}_{\partial \mathscr{C}}(\bm{x},s) \in \mathbb{R}^2$ its nondimensional velocity. In the literature, the portions $\mathscr{L}(s)$, $\mathscr{N}(s)$, and $\mathscr{T}(s)$ of the boundary are traditionally referred to as the \emph{leading}, \emph{neutral}, and \emph{trailing edge}, respectively.

If the contact area is fixed, that is $\mathscr{C}(s) = \mathscr{C}\equiv \mathscr{C}_0 $ and $\bar{\bm{V}}_{\partial \mathscr{C}}(\bm{x},s) = \bm{0}$, the mathematical definitions in~\eqref{eq:Inflow0Outflow} simplify to
\begin{subequations}\label{eq:Inflow0Outflow0}
\begin{align}
\mathscr{L}(s) &\triangleq \partial \mathscr{C}_{-}(s)= \bigl\{\bm{x}\in \partial \mathscr{C} \mathrel{\big|} \bar{\bm{V}}(\bm{x},s) \cdot \hat{\bm{n}}_{\partial \mathscr{C}}(\bm{x}) < 0\bigr\}, \\
\mathscr{N}(s) &\triangleq \partial \mathscr{C}_{0}(s)= \bigl\{\bm{x}\in \partial \mathscr{C} \mathrel{\big|} \bar{\bm{V}}(\bm{x},s) \cdot \hat{\bm{n}}_{\partial \mathscr{C}}(\bm{x}) = 0\bigr\},  \\
\mathscr{T}(s) &\triangleq \partial \mathscr{C}_{+}(s)= \bigl\{\bm{x}\in \partial \mathscr{C} \mathrel{\big|} \bar{\bm{V}}(\bm{x},s) \cdot \hat{\bm{n}}_{\partial \mathscr{C}}(\bm{x}) > 0\bigr\}, 
\end{align}
\end{subequations}
For a transport equation like Eq.~\eqref{eq:PDenoMOdel}, the natural BC should be prescribed at the inflow boundary, that is, on the leading edge $\mathscr{L}(s)$, for all $s \in (0,S)$. In particular, enforcing $\bm{z}(\bm{x},s) = \bm{0}$, ensures that the bristles enter the contact area undeformed. For purely elastic materials ($\mathbf{\Sigma}_1 = \mathbf{0}$ in Eq.~\eqref{eq:rheol1}), this condition arises naturally from the continuity of the stresses when transitioning from the free portions of the bodies to the contact region. Under the assumption of purely elastic behaviour, continuity of stress is indeed equivalent to continuity of deformation, which justifies the BC $\bm{z}(\bm{x},s) = \bm{0}$. For viscoelastic materials, however, this equivalence no longer holds. Using, for example, Eq.~\eqref{eq:rheol1}, the appropriate boundary condition would be $\bm{f}(\bm{x},s) = \mathbf{\Sigma}_0\bm{z}(\bm{x},s) + V\ped{r}(s)\mathbf{\Sigma}_1\od{\bm{z}(\bm{x},s)}{s} = \bm{0}$ on $\mathscr{L}(s)$ for all $s \in (0,S)$.
Despite this, the LuGre literature has consistently adopted the simpler condition $\bm{z}(\bm{x},s) = \bm{0}$. The same approach is adopted here, both because it leads to a significantly simpler formulation and because the coefficients of $\mathbf{\Sigma}_1$ are typically small in rolling contact \cite{TsiotrasConf,Tsiotras1,Tsiotras3,Deur0,Deur1,Deur2}. 

\subsubsection{Tangential forces and vertical moment}\label{ect:forces}
To proceed with the determination of the tangential forces and vertical moment generated by the frictional rolling contact process, with a little abuse of notation, it is convenient to restate Eq.~\eqref{eq:rheol1} as
\begin{align}\label{eq:fxt}
\bm{f}(\bm{x},t) = \mathbf{\Sigma}_0\bm{z}(\bm{x},t) + \mathbf{\Sigma}_1\dod{\bm{z}(\bm{x},t)}{t}, 
\end{align}
or equivalently, using the travelled distance as independent time-like variable,
 \begin{align}\label{eq:fxs}
\bm{f}(\bm{x},s) = \mathbf{\Sigma}_0\bm{z}(\bm{x},s) + V\ped{r}(s)\mathbf{\Sigma}_1\dod{\bm{z}(\bm{x},s)}{s}. 
\end{align}
Equations~\eqref{eq:fxt} and~\eqref{eq:fxs} describe the bristle froce per unit of vertical load.
Accordingly, the tangential forces $\mathbb{R}\ni \bm{F}_{\bm{x}}(s) = [F_x(s)\; F_y(s)]^{\mathrm{T}}$ and the vertical moment $M_z(s)\in \mathbb{R}$ may be computed as
\begin{subequations}\label{eq:FandM}
\begin{align}
\bm{F}_{\bm{x}}(s) & = \iint_{\mathscr{C}(s)}p(\bm{x},s)\bm{f}(\bm{x},s) \dif \bm{x}, \label{eq:Fundef}\\
M_z(s) & = \iint_{\mathscr{C}(s)}p(\bm{x},s)\bigl[ xf_y(\bm{x},s)-yf_x(\bm{x},s)\bigr] \dif \bm{x}, \quad s \in [0,S],\label{eq:Mzunderfr}
\end{align}
\end{subequations}
where $p \in C^0(\mathscr{C}\times[0,S];\mathbb{R}_{\geq 0})$ indicates the pressure distribution inside the contact area. A realistic contact pressure may be expected to be compactly supported (that is, vanishing on the boundary $\partial \mathscr{C}$), but, in many practical cases, non compactly supported distributions are also adopted, including the constant or exponentially decreasing ones. These modelling approaches are mainly motivated by relatively satisfactory agreements with measured quantities \cite{TsiotrasConf,Tsiotras1,Tsiotras3,Deur0,Deur1,Deur2}.

In any case, it is important to emphasise that the computation of the vertical moment $M_z(s)$ in Eq.~\eqref{eq:FandM} is performed with respect to the undeformed configuration. This is appropriate when the deformation $\bm{z}(\bm{x},s)$ is sufficiently small or when determining the actual deformed configuration is impractical, as in the second scenario illustrated in Fig.~\ref{fig:LumpModel}. However, when at least one of the bodies in rolling contact is sufficiently rigid, as in Fig.~\ref{fig:LumpModel}(a), the vertical moment can also be evaluated using the deformed configuration. In that case, Eq.~\eqref{eq:Mzunderfr} should be modified as follows:
\begin{align}\label{eq:Malt}
M_z(s) & = \iint_{\mathscr{C}(s)}p(\bm{x},s)\Bigl[\bigl(x+z_x(\bm{x},s)\bigr)f_y(\bm{x},s)-\bigl(y+z_y(\bm{x},s)\bigr)f_x(\bm{x},s)\Bigr] \dif \bm{x}, \quad s \in [0,S],
\end{align}
whereas~\eqref{eq:Fundef} remains formally unchanged.

Section~\ref{sect:modelRollSimple} below is dedicated to the simplest rolling contact model that may be deduced from Eq.~\eqref{eq:PDenoMOdel09}.

\subsection{Standard linear rolling contact model}\label{sect:modelRollSimple}
Rolling contact models with different orders of complexity may be derived from Eq.~\eqref{eq:PDenoMOdel09} by appropriately specifying expressions for the (nondimensional) transport and rigid relative velocities $\bar{\bm{V}}(\bm{x},s)$ and $\bar{\bm{v}}\ped{r}(\bm{x},s)$. 
Amongst these, the simplest variant may be obtained by observing that most rolling contact processes evolve along a main rolling direction. For instance, in normal driving conditions, tyres experience limited camber angles and turning speeds. Similarly, in rail-wheel interactions, geometrical and effective spin effects are extremely low. In these cases, the transport velocity appearing in Eq.~\eqref{eq:Eulerian} may be fairly approximately as
\begin{align}
\bm{V}(\bm{x},t) \approx -\begin{bmatrix}V\ped{r}(t)\\ 0\end{bmatrix},
\end{align}
and the corresponding nondimensional transport velocity becomes
\begin{align}\label{eq:VsStandard}
\bar{\bm{V}}(\bm{x},s) = \bar{\bm{V}} = -\begin{bmatrix}1\\ 0\end{bmatrix}.
\end{align}
Additionally, the nondimensional rigid relative velocity $\bar{\bm{v}}\ped{r}(\bm{x},s)$ coincides with the rigid slip velocity $\mathbb{R}^2 \ni \bar{\bm{v}}(\bm{x},s) = [\bar{v}_x(\bm{x},s) \; \bar{v}_y(\bm{x},s)]^{\mathrm{T}}$ given by
\begin{align}\label{eq:v0}
\bar{\bm{v}}(\bm{x},s) = -\bm{\sigma}(s) - \mathbf{A}_\varphi(s)\bm{x},
\end{align}
where $\mathbb{R}^2 \ni \bm{\sigma}(s) = [\sigma_x(s) \; \sigma_y(s)]^{\mathrm{T}}$ denotes the \emph{translational slip}, and the matrix $\mathbf{A}_{\varphi}(s) \in \mathbf{M}_2(\mathbb{R})$ reads
\begin{align}\label{eq:Aphi}
\mathbf{A}_\varphi (s)&  \triangleq \begin{bmatrix} 0 & -\varphi (s)\\
\varphi(s) & 0 \end{bmatrix},
\end{align}
being $\varphi(s) \in \mathbb{R}$ the \emph{total spin slip}, which is the sum of the \emph{geometrical spin} $\varphi_\gamma (s) \in \mathbb{R}$, and the \emph{effective spin} $\varphi_\psi(s)$:
\begin{align}
\varphi(s) = \varphi_\gamma(s) + \varphi_\psi(s).
\end{align}

Combining the general PDE~\eqref{eq:PDenoMOdel09} with~\eqref{eq:VsStandard}-\eqref{eq:Aphi} yields the \emph{standard linear rolling contact model} formalised below (\modref{stdmodel}).
\begin{modelbox}[stdmodel]{Standard linear model}
\begin{subequations}\label{eq:standard0}
\begin{align}
\begin{split}
& \dpd{\bm{z}(\bm{x},s)}{s} -\dpd{\bm{z}(\bm{x},s)}{x} =\mathbf{\Sigma}\bigl(\bar{\bm{v}}(\bm{x},s),s\bigr)\bm{z}(\bm{x},s) + \bm{h}\bigl(\bar{\bm{v}}(\bm{x},s),s\bigr), \quad \bm{x} \in \mathring{\mathscr{C}}(s), \; s \in (0,S),
\end{split}\\
& \bm{z}(\bm{x},s) = \bm{0}, \quad \bm{x}\in \mathscr{L}(s), \;  s \in (0,S), \\
& \bm{z}(\bm{x},0) = \bm{z}_0(\bm{x}), \quad \bm{x}\in \mathring{\mathscr{C}}_0.
\end{align}
\end{subequations}
\end{modelbox}
\modref{stdmodel} represents the FrBD counterpart of the standard brush models known from the literature, which are also equivalent to Kalker's simplified theory of rolling contact. In particular, \modref{stdmodel} is linear in the variable $\bm{z}(\bm{x},s)$, and can be solved numerically with standard methods for virtually any combination of slip velocities.
For the sequel, it is profitable to introduce the following functions:
\begin{subequations}\label{eq:Sigmaf0}
\begin{align}
\tilde{\mathbf{H}}(\bm{x},s) & \triangleq \mathbf{H}\bigl(\bar{\bm{v}}(\bm{x},s),s\bigr), \\
\tilde{\bm{h}}(\bm{x},s) & \triangleq \tilde{\mathbf{H}}(\bm{x},s)\bar{\bm{v}}(\bm{x},s) = \bm{h}\bigl(\bar{\bm{v}}(\bm{x},s),s\bigr), \\
\tilde{\mathbf{\Sigma}}(\bm{x},s) & \triangleq \mathbf{\Sigma}\bigl(\bar{\bm{v}}(\bm{x},s),s\bigr),
\end{align}
\end{subequations}
so that Eq.~\eqref{eq:standard0} may be restated as
\begin{subequations}\label{eq:standard1}
\begin{align}
\begin{split}
& \dpd{\bm{z}(\bm{x},s)}{s} -\dpd{\bm{z}(\bm{x},s)}{x} =\tilde{\mathbf{\Sigma}}(\bm{x},s)\bm{z}(\bm{x},s) + \tilde{\bm{h}}(\bm{x},s), \quad \bm{x} \in \mathring{\mathscr{C}}(s), \; s \in (0,S),
\end{split}\\
& \bm{z}(\bm{x},s) = \bm{0}, \quad \bm{x}\in \mathscr{L}(s), \;  s \in (0,S), \\
& \bm{z}(\bm{x},0) = \bm{z}_0(\bm{x}), \quad \bm{x}\in \mathring{\mathscr{C}}_0.
\end{align}
\end{subequations}
The form~\eqref{eq:standard1} is more compact and amenable to analysis. In fact, due to the term $\tilde{\mathbf{\Sigma}}(\bm{x},s)$ being highly nonlinear in $\bm{x}$, it is not possible to derive explicit analytical solutions to the PDE~\eqref{eq:standard1}, which makes the adoption of numerical techniques necessary. With these premises, Sect.~\ref{sect:wellPosStandard} investigates well-posedness of the PDE~\eqref{eq:standard1} considering a fixed contact area.

\subsubsection{Well-posedness}\label{sect:wellPosStandard}
The present section delivers well-posedness results for a simplified version of Eq.~\eqref{eq:standard1}, obtained by assuming a time-invariant contact area. The more general case of a time-varying contact area may be addressed as explained in \cite{Tribology}, after mapping the variable domain into a fixed one. In this paper, the effect of time-varying vertical forces and contact areas is studied qualitatively in Sect.~\ref{sect:normVertVar}, limited to the line contact case. 
Proceeding instead with the well-posedness analysis, for a contact area fixed over time, or equivalently, travelled distance, the PDE~\eqref{eq:standard1} reduces to
\begin{subequations}\label{eq:standard2}
\begin{align}
\begin{split}
& \dpd{\bm{z}(\bm{x},s)}{s} -\dpd{\bm{z}(\bm{x},s)}{x} =\tilde{\mathbf{\Sigma}}(\bm{x},s)\bm{z}(\bm{x},s) + \tilde{\bm{h}}(\bm{x},s), \quad (\bm{x},s) \in \mathring{\mathscr{C}}\times (0,S),\label{eq:standard2PDE}
\end{split}\\
& \bm{z}(\bm{x},s) = \bm{0}, \quad (\bm{x},s) \in \mathscr{L}\times (0,S), \label{eq:BCquasiStandard}\\
& \bm{z}(\bm{x},0) = \bm{z}_0(\bm{x}), \quad \bm{x}\in \mathring{\mathscr{C}}.
\end{align}
\end{subequations}
Well-posedness results for the simplified Eq.~\eqref{eq:standard2} are asserted below by Theorem~\ref{thm:ex1}. 
\begin{theorem}[Existence and uniqueness of solutions]\label{thm:ex1}
Suppose that $\mathring{\mathscr{C}}\subset \mathbb{R}^2$ is bounded, with boundary $\partial \mathscr{C}$ piecewise $C^1$. Then, for all $\tilde{\mathbf{\Sigma}} \in C^0(\mathscr{C}\times[0,S];\mathbf{M}_2(\mathbb{R}))$ and $\tilde{\bm{h}} \in C^0([0,S];L^2(\mathring{\mathscr{C}};\mathbb{R}^2))$ as in Eq.~\eqref{eq:Sigmaf0}, and ICs $\bm{z}_0 \in L^2(\mathring{\mathscr{C}};\mathbb{R}^2)$, the PDE~\eqref{eq:standard2} admits a unique \emph{mild solution} $\bm{z} \in C^0([0,S];L^2(\mathring{\mathscr{C}};\mathbb{R}^2))$. Additionally, if $\tilde{\mathbf{\Sigma}} \in C^1(\mathscr{C}\times[0,S];\mathbf{M}_2(\mathbb{R}))$, $\tilde{\bm{h}} \in C^1([0,S];L^2(\mathring{\mathscr{C}};\mathbb{R}^2))$, and the IC $\bm{z}_0 \in \mathscr{D}(\mathscr{A})$, with $\mathscr{D}(\mathscr{A}) \triangleq \{\bm{\zeta} \in L^2(\mathring{\mathscr{C}};\mathbb{R}^2) \mathrel{|} \pd{\bm{\zeta}}{x} \in  L^2(\mathring{\mathscr{C}};\mathbb{R}^2), \; \eval[0]{\bm{\zeta}}_{\mathscr{L}} = \bm{0}\} $, the solution is \emph{classical}, that is, $\bm{z} \in C^1([0,S];L^2(\mathring{\mathscr{C}};\mathbb{R}^2)) \cap C^0([0,S];\mathscr{D}(\mathscr{A}))$.
\begin{proof}[Proof]
The result follows along the same lines as the proof of Theorem~\ref{thm:ex2}, which is given in Appendix~\ref{app:Proof1}.
\end{proof}
\end{theorem}
Some comments about the regularity of the original matrices and inputs are as follows: $\mathbf{M}\in C^0(\mathbb{R}^2; \mathbf{Sym}_2(\mathbb{R}))$ implies $\mathbf{G}\in C^0(\mathbb{R}^2; \mathbf{Sym}_2(\mathbb{R}))$ for all $\varepsilon \in \mathbb{R}_{\geq 0}$, ensuring the existence of mild solutions for $\bar{\bm{v}} \in C^0(\mathscr{C}\times[0,S];\mathbb{R}^2)$ and $V\ped{r} \in C^0([0,S]; [V\ped{min}, V\ped{max}])$. Additionally, $\mathbf{M}\in C^1(\mathbb{R}^2; \mathbf{Sym}_2(\mathbb{R}))$ implies $\mathbf{G}\in C^1(\mathbb{R}^2; \mathbf{Sym}_2(\mathbb{R}))$ for all $\varepsilon \in \mathbb{R}_{>0}$. Therefore, by composition of continuously differentiable functions, it follows that $\tilde{\mathbf{\Sigma}} \in C^1(\mathscr{C}\times[0,S];\mathbf{M}_2(\mathbb{R}))$ and $\tilde{\bm{h}} \in C^1([0,S];L^2(\mathring{\mathscr{C}};\mathbb{R}^2))$ for all $\bar{\bm{v}} \in C^1(\mathscr{C}\times[0,S];\mathbb{R}^2)$ and $V\ped{r} \in C^1([0,S]; [V\ped{min}, V\ped{max}])$. Essentially, the regularised rolling contact~\modref{stdmodel} (with $\varepsilon \in \mathbb{R}_{>0}$) admits classical solutions for sufficiently smooth slip and rolling velocities. 
%It is worth noting that Theorem~\ref{thm:ex1} is stated under conditions that are not the most general possible. In fact, the existence of mild solutions may be established under less restrictive assumptions on the matrix $\tilde{\mathbf{\Sigma}}(\bm{x},s)$, although this requires more involved (and rather tedious) mathematical machinery; readers interested in the broader mathematical context may consult \cite{Pazy,Tanabe1,Tanabe} for further details. Nevertheless, the focus of this paper is not on proving the most general theoretical results, but rather on offering a rigorous foundation for the numerical analyses presented in Sects.~\ref{sect:math} and~\ref{sect:numer}, for which the assumptions of Theorem~\ref{thm:ex1} are more than adequate. In any case, 

%The existence and uniqueness results asserted by Theorem~\ref{thm:ex1} are of a general nature, but only qualitative. In some cases, Eq.~\eqref{eq:standard2} admits explicit closed-form solutions, at least in steady-state conditions. Some formulae are derived in Sect.~\ref{sect:closedForm}.

%\subsubsection{Closed-form solution}\label{sect:closedForm}

\subsection{Semilinear rolling contact model for large spin slips}\label{sect:semilinear}
Large spin slips give rise to geometric effects that are not accurately captured by~\modref{stdmodel}.
The classes of models introduced below describe the dynamics of deformable bodies rolling on a rigid, flat substrate in the presence of substantial spin slip. Their primary applications include tyre-road interaction and spherical robots.
The same formulations also apply to other elastic rolling contact pairs, such as wheel-rail systems, but these typically experience only modest spin slip, making the full theory developed here less relevant in practice.

In particular, a first spin component distorts the trajectories of the bristles within the contact area. In the most general case of sideways rolling, this gives
\begin{align}\label{eq:generalV}
\bar{\bm{V}}(\bm{x},s) = -\begin{bmatrix} \varepsilon_y(s) \\ -\varepsilon_x(s) \end{bmatrix} + \mathbf{A}_{\varphi_1}(s)\bm{x},
\end{align}
where $\mathbb{R}^2 \ni \bm{\varepsilon}(s) = [\varepsilon_x(s)\; \varepsilon_y(s)]^{\mathrm{T}}$, with $\norm{\bm{\varepsilon}(s)}_2 = 1$, and $\mathbf{A}_{\varphi_1}(s) \in \mathbf{M}_2(\mathbb{R})$ reads
\begin{align}\label{eq:geometricSPinTensor}
\mathbf{A}_{\varphi_1}(s) \triangleq  \begin{bmatrix}0 & \varphi_1(s) \\ -\varphi_1(s) & 0 \end{bmatrix}.
\end{align}
In particular, for tyres and railway wheels, there is no sideways rolling, and therefore $\varepsilon_y(s) =1$ and $\varepsilon_x(s) = 0$ in Eq.~\eqref{eq:generalV}. The nature of the spin term $\varphi_1(s) \in \mathbb{R}$ appearing in Eq.~\eqref{eq:generalV} depends heavily on the reciprocal deformative behaviour between the contacting bodies. In the case where the second body is infinitely rigid compared to the first one, as it happens in tyre-road interactions, $\varphi_1(s) \equiv \varphi_\gamma(s)$, and the matrix $\mathbf{A}_{\varphi_1}(s)$ may be fairly referred to as the \emph{geometric spin tensor}. In contrast, when the two contacting bodies undergo comparable deformations, as is the case for rail and wheel, $\varphi_1(s) \equiv \frac{\varphi_\gamma(s)-\varphi_\psi(s)}{2}$. It is clear that, in the latter situation, $\varphi_1(s)$ cannot be interpreted as a pure geometrical spin.

The presence of large spin slips also affects the expression for the nondimensional rigid relative velocity, which becomes a function of the bristle deformation according to
\begin{align}\label{eq:RigidVelSpin}
\bar{\bm{v}}\ped{r}\bigl(\bm{z}(\bm{x},s),\bm{x},s\bigr) =\bar{\bm{v}}(\bm{x},s)- \mathbf{A}_{\varphi_2}(s)\bm{z}(\bm{x},s),
\end{align}
where $\bar{\bm{v}}(\bm{x},s)$ reads again as in~\eqref{eq:v0}, and $\mathbf{A}_{\varphi_2}(s) \in \mathbf{M}_2(\mathbb{R})$ writes
\begin{align}\label{eq:turnSpinMatrix}
\mathbf{A}_{\varphi_2}(s) & \triangleq \begin{bmatrix} 0 & -\varphi_2(s) \\ \varphi_2(s) & 0\end{bmatrix}.
\end{align}
Again, the specific definition of the spin term $\varphi_2(s) \in \mathbb{R}$ in Eq.~\eqref{eq:turnSpinMatrix} is contingent on the reciprocal deformative behaviour of the contacting bodies. When one of the two bodies may be regarded as rigid, then $\varphi_2(s) \equiv \varphi_\psi(s)$, and the matrix $\mathbf{A}_{\varphi_2}(s)$ appearing in~\eqref{eq:RigidVelSpin} may be referred to as the \emph{effective spin tensor}\footnote{In Vehicle Dynamics, it is also called \emph{turning tensor} or \emph{turn spin tensor}.}; \emph{vice versa}, in a situation of similarity, $\varphi_2(s) \equiv -\varphi_1(s)$.
For notational convenience, the spins $\varphi_1(s)$ and $\varphi_2(s)$ may be collected into the vector\footnote{This definition of $\bm{\varphi}(s)$ is simply convenient for the mathematical analysis of Sect.~\ref{sect:math}. For a deformable body rolling over a rigid substrate, $\bm{\varphi}(s) = [\varphi_1(s) \; \varphi_2(s)]^{\mathrm{T}} \equiv [\varphi_\gamma(s) \; \varphi_\psi(s)]^{\mathrm{T}}$, so that $\bm{\varphi}(s)$ is indeed sufficient to completely characterise spin conditions. Conversely, when the two bodies in rolling contact exhibit a similar behaviour, the definition $\bm{\varphi}(s) \triangleq [\varphi_\gamma(s) \; \varphi_\psi(s)]^{\mathrm{T}}$ should be used (since also in this case $\varphi(s) = \varphi_\gamma(s) + \varphi_\psi (s)$).} $\mathbb{R}^2 \ni \bm{\varphi}(s) = [\varphi_1(s) \; \varphi_2(s)]^{\mathrm{T}}$.

According to Eqs.~\eqref{eq:generalV}-\eqref{eq:turnSpinMatrix}, the \emph{semilinear model for large spin slips} (\modref{semilinmodel}) is formulated below.
\begin{modelbox}[semilinmodel]{Semilinear model for large spin slips}
\begin{subequations}
\begin{align}\label{eq:semilinear}
\begin{split}
& \dpd{\bm{z}(\bm{x},s)}{s} + \bigl(\bar{\bm{V}}(\bm{x},s)\cdot \nabla_{\bm{x}}\bigr)\bm{z}(\bm{x},s) =\mathbf{\Sigma}\Bigl(\bar{\bm{v}}\ped{r}\bigl(\bm{z}(\bm{x},s),\bm{x},s\bigr),s\Bigr)\bm{z}(\bm{x},s)\\
& \qquad \qquad \qquad \qquad \qquad \qquad \qquad\qquad + \bm{h}\Bigl(\bar{\bm{v}}\ped{r}\bigl(\bm{z}(\bm{x},s),\bm{x},s\bigr),s\Bigr), \quad \bm{x} \in \mathring{\mathscr{C}}(s), \; s \in (0,S),
\end{split}\\
& \bm{z}(\bm{x},s) = \bm{0}, \quad \bm{x}\in \mathscr{L}(s), \;  s \in (0,S), \\
& \bm{z}(\bm{x},0) = \bm{z}_0(\bm{x}), \quad \bm{x}\in \mathring{\mathscr{C}}_0.
\end{align}
\end{subequations}
\end{modelbox}
The above Eq.~\eqref{eq:semilinear} represents a full semilinear PDE. Its complex structure does not allow for invoking simple well-posedness results, and a rigorous analysis is beyond the scope of the present paper. A rich variety of numerical results is, however, reported in Sect.~\ref{sect:numer}, where the predictions of Eq.~\eqref{eq:semilinear} are compared to those obtained using the linear model introduced next in Sect.~\ref{sect:LinearLargeSpin}.

Before moving to Sect.~\ref{sect:LinearLargeSpin}, however, it is first beneficial to draw some considerations about Eqs.~\eqref{eq:generalV} and~\eqref{eq:RigidVelSpin}. These may be interpreted by introducing the \textit{instantaneous tilting centre} $C_{\varphi_1}(s)$ with coordinates
\begin{equation}\label{eq:xC1}
\bm{x}_{C_{\varphi_1}}(s) = \begin{bmatrix} x_{C_{\varphi_1}}(s) & y_{\varphi_1}(s)\end{bmatrix}^\mathrm{T}\triangleq  \dfrac{\bm{\varepsilon}(s)}{\varphi_1(s)}.
\end{equation}
Utilising~\eqref{eq:xC1}, Eq.~\eqref{eq:generalV} may be rewritten as
\begin{align}
\bar{\bm{V}}(\bm{x},s) = \mathbf{A}_{\varphi_1}(s)\bigl(\bm{x}-\bm{x}_{C_{\varphi_1}}(s) \bigr),
\end{align}
which states that the trajectories of the bristles inside the contact patch are instantaneously centred at $C_{\varphi_1}(s)$. For a tyre rolling over a rigid road, the situation is schematically illustrated in Fig.~\ref{fig:tireSchematic}, where $C_{\varphi_1}(s) \equiv C_\gamma(s)$ (since $\varphi_1(s) = \varphi_\gamma (s)$). From geometrical considerations, it is clear that, for a deformable wheel rolling over a rigid substrate, the point $C_{\varphi_1}(s)$ should always lie outside of the contact area. In this case, the second spin component represents an additional contribution to the rigid relative velocity due to differences in the elastic behaviour of the contact pairs, and should be regarded as an effective spin component, since deformable bristles only protrude from the wheel ($\varphi_2(s) = \varphi_\psi(s)$). For omnidirectional bodies, such as spherical robots, this geometrical constraint may be violated, but then it becomes unclear how to distinguish between the tilting and spinning spin components. It seems, therefore, natural to assume that $C_{\varphi_1}(s)$ should lie outside $\mathscr{C}$.
These considerations were instrumental in deriving explicit solutions within the simpler theoretical setting offered by the linear brush models in \cite{Meccanica2,SphericalWheel,LuGreSpin}, and will also prove crucial in Sect.~\ref{sect:InSD0D}.

In contrast, when the contacting bodies have comparable elastic properties, as in the case of wheel and rail, the instantaneous tilting centre $C_{\varphi_1}(s)$ becomes a difference between two signed curvatures, and, in general, there is no guarantee that it would lie outside $\mathscr{C}$. However, as already mentioned, railway wheels tend to experience extremely small spin slips, with the consequence that often $C_{\varphi_1}(s)$ lies indeed outside the contact area. Moreover, the fact that $\varphi_2(s) = -\varphi_1(s)$ may be explained by recalling that, conceptually, deformable bristles should be attached to both bodies (wheel and rail). In such a situation, there exist no real geometric and effective spins.

\begin{figure}
\centering
\includegraphics[width=0.6\linewidth]{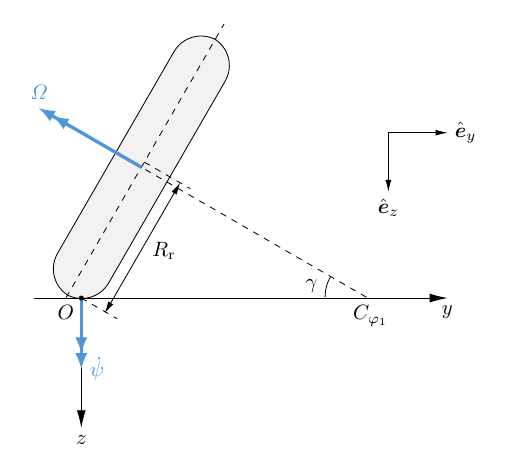} 
\caption{Schematic of a tyre rolling over a rigid road. For a rigid wheel, the spin components would be defined as $\varphi_1 \triangleq \frac{\sin \gamma}{R_\delta}$ and $\varphi_2 \triangleq -\frac{\dot{\psi}}{\Omega R_\delta}$, where $\gamma \in \mathbb{R}$ denotes the camber angle, $\dot{\psi} \in \mathbb{R}$ the turning speed, $\Omega \in \mathbb{R}_{>0}$ the angular velocity, and $R_\delta \in \mathbb{R}_{>0}$ the deformed radius. From simple trigonometrical considerations, it follows that $C_{\varphi_1}$ always lies outside $\mathscr{C}$. In reality, for a deformable tyre, the spin slips are often defined as $\varphi_1 = \varphi_\gamma \triangleq \frac{(1-\varepsilon_\gamma)\sin\gamma}{R\ped{r}}$ and $\varphi_2 = \varphi_\psi \triangleq -\frac{\dot{\psi}}{V\ped{r}}$, where $\varepsilon_\gamma \in [0,1)$ denotes the \emph{camber reduction factor}, $R\ped{r} \geq R_\delta$ the \emph{effective rolling radius}, and $V\ped{r} \in \mathbb{R}_{>0}$ is the rolling speed. Also in this case, $C_{\varphi_1}$ must necessarily lie outside $\mathscr{C}$.}
\label{fig:tireSchematic}
\end{figure}

%\linestyle{matblue}{solid}{1pt} Linear model,
%\linestyle{matred}{dash dot}{0.8pt} Semilinear model,
%and \linestyle{matgreen}{dashed}{0.5pt} Nonlinear model.

\subsection{Linear rolling contact model for large spin slips}\label{sect:LinearLargeSpin}
An approximation to the semilinear PDE~\eqref{eq:semilinear} may be recovered by conveniently replacing $\bar{\bm{v}}\ped{r}(\bm{z}(\bm{x},s),\bm{x},s)$ by $\bar{\bm{v}}(\bm{x},s)$ in the nonlinear functions $\mathbf{\Sigma}(\cdot,\cdot)$ and $\mathbf{H}(\cdot,\cdot)$, yielding the following \emph{linear model for large spin slips}.
\begin{modelbox}[linmodel2]{Linear model for large spin slips}
\begin{subequations}\label{eq:linwasdir}
\begin{align}\label{eq:linwasdir2}
\begin{split}
& \dpd{\bm{z}(\bm{x},s)}{s} + \bigl(\bar{\bm{V}}(\bm{x},s)\cdot \nabla_{\bm{x}}\bigr)\bm{z}(\bm{x},s) =\mathbf{\Sigma}\bigl(\bar{\bm{v}}(\bm{x},s),s\bigr)\bm{z}(\bm{x},s) \\
&\qquad \qquad \qquad \qquad \qquad \qquad \qquad\qquad  + \mathbf{H}\bigl(\bar{\bm{v}}(\bm{x},s),s\bigr)\bar{\bm{v}}\ped{r}\bigl(\bm{z}(\bm{x},s),\bm{x},s\bigr), \quad \bm{x} \in \mathring{\mathscr{C}}(s), \; s \in (0,S),
\end{split}\\
& \bm{z}(\bm{x},s) = \bm{0}, \quad \bm{x}\in \mathscr{L}(s), \;  s \in (0,S), \\
& \bm{z}(\bm{x},0) = \bm{z}_0(\bm{x}), \quad \bm{x}\in \mathring{\mathscr{C}}_0.
\end{align}
\end{subequations}
\end{modelbox}
As opposed to~\eqref{eq:semilinear}, Eq.~\eqref{eq:linwasdir2} is linear in the variable $\bm{z}(\bm{x},s)$, which permits establishing key results regarding the well-posedness and stability of~\modref{linmodel2}. In fact, for $\varepsilon = 0$, linearising Eq.~\eqref{eq:linwasdir2} around the equilibrium $(\bm{z}^\star(\bm{x}), \bar{\bm{v}}^\star(\bm{x})) = (\bm{0},\bm{0}) \in \mathbb{R}^4$ yields the same models analysed in \cite{Meccanica2,SphericalWheel}, which admit closed-form solutions for certain simple contact shapes both in steady-state and transient conditions.
Moreover, the simplified PDE~\eqref{eq:linwasdir2} may serve as a foundation for fixed-point algorithms designed to iteratively approximate the solution of~\eqref{eq:semilinear}.

For what follows, and in the same spirit of Sect.~\ref{sect:modelRollSimple}, the following functions are introduced: 
\begin{subequations}\label{eq:functionsH}
\begin{align}
\tilde{\mathbf{H}}(\bm{x},s) & \triangleq \mathbf{H}\bigl(\bar{\bm{v}}(\bm{x},s),s\bigr), \\
\tilde{\bm{h}}(\bm{x},s) & \triangleq \tilde{\mathbf{H}}(\bm{x},s)\bar{\bm{v}}(\bm{x},s) =  \bm{h}\bigl(\bar{\bm{v}}(\bm{x},s),s\bigr), \\
\tilde{\mathbf{\Sigma}}_{\varphi}(\bm{x},s) & \triangleq \mathbf{\Sigma}\bigl(\bar{\bm{v}}(\bm{x},s),s\bigr)- \tilde{\mathbf{H}}(\bm{x},s)\mathbf{A}_{\varphi_2}(s).
\end{align}
\end{subequations}
Utilising~\eqref{eq:functionsH}, Eq.~\eqref{eq:linwasdir} may be recast more conveniently as
\begin{align}\label{eq:linearLargeSpinform2}
\begin{split}
& \dpd{\bm{z}(\bm{x},s)}{s} + \bigl(\bar{\bm{V}}(\bm{x},s)\cdot \nabla_{\bm{x}}\bigr)\bm{z}(\bm{x},s) =\tilde{\mathbf{\Sigma}}_{\varphi}(\bm{x},s)\bm{z}(\bm{x},s) +\tilde{\bm{h}}(\bm{x},s),\quad \bm{x} \in \mathring{\mathscr{C}}(s),  \; s \in (0,S),
\end{split}\\
& \bm{z}(\bm{x},s) = \bm{0}, \quad \bm{x}\in \mathscr{L}(s), \;  s \in (0,S).
\end{align}
The form~\eqref{eq:linearLargeSpinform2} is more amenable to mathematical analysis. The next Sect.~\ref{sect:wllPops2} is dedicated to inferring well-posedness and stability properties for~\eqref{eq:linearLargeSpinform2} under additional assumptions on the domain $\mathring{\mathscr{C}}(s)$ the transport velocity $\bar{\bm{V}}(\bm{x},s)$.

Before moving to Sect.~\ref{sect:wllPops2}, it is worth commenting on the predictive capabilities of~\modref{linmodel2} in relation to those of~\ref{stdmodel} and~\ref{semilinmodel}. In this context, the first observation concerns the equivalence between the PDEs~\eqref{eq:semilinear} and~\eqref{eq:linwasdir2} for $\varphi_2(s) \approx 0$, even in the presence of large components $\varphi_1(s)$. Essentially, in the context of tyre-road interactions,~\modref{linmodel2} can be employed to study the effect of large camber angles and moderate turn spins. Of course, since the bristle deformation $\bm{z}(\bm{x},s)$ may be expected to be small, Eq.~\eqref{eq:linwasdir2} should constitute a sufficiently good approximation to~\eqref{eq:semilinear} even when the second spin component is non-negligible ($\varphi_2(s) = \varphi_\psi(s)$ for tyres and spherical robots). For wheel-rail systems, small values of $\varphi_2(s)$ automatically imply small $\varphi_1(s)$, and hence the predictive capabilities of~\modref{linmodel2} coincide with those of~\ref{stdmodel} for $\varphi_2(s) \approx 0$. In fact, \modref{stdmodel} is typically sufficient when it comes to describing wheel-rail rolling contact phenomena.

\subsubsection{Well-posedness}\label{sect:wllPops2}
Existence and uniqueness results for~\modref{linmodel2} are stated under the assumption that the contact area and the transport velocity are time-independent. Accordingly, the PDE~\eqref{eq:linearLargeSpinform2} is reformualted as
\begin{subequations}\label{eq:quasiStationary0}
\begin{align}
\begin{split}
& \dpd{\bm{z}(\bm{x},s)}{s} + \bigl(\bar{\bm{V}}(\bm{x})\cdot \nabla_{\bm{x}}\bigr)\bm{z}(\bm{x},s) =\tilde{\mathbf{\Sigma}}_{\varphi}(\bm{x},s)\bm{z}(\bm{x},s) +\tilde{\bm{h}}(\bm{x},s),\quad (\bm{x},s) \in\mathring{\mathscr{C}} \times (0,S),\label{eq:quasiStationary0dy}
\end{split}\\
& \bm{z}(\bm{x},s) = \bm{0}, \quad (\bm{x},s) \in \mathscr{L}\times (0,S),\label{eq:BCquasistesd} \\
& \bm{z}(\bm{x},0) = \bm{z}_0(\bm{x}), \quad \bm{x} \in \mathring{\mathscr{C}},
\end{align}
\end{subequations}
with the nondimensional transport velocity $\bar{\bm{V}}(\bm{x})$ given, in the most general case, by
\begin{align}\label{eq:Vss}
\bar{\bm{V}}(\bm{x}) = -\begin{bmatrix} \varepsilon_y \\ - \varepsilon_x\end{bmatrix} + \mathbf{A}_{\varphi_1}\bm{x}.
\end{align}
%\begin{align}\label{eq:Vss}
%\bar{\bm{V}}(\bm{x}) = -\begin{bmatrix} \varepsilon_y \\ - \varepsilon_x\end{bmatrix} + \mathbf{A}_{\varphi_1}\bm{x}, \quad \textnormal{or} \quad \bar{\bm{V}}(\bm{x})=  -\begin{bmatrix} 1 \\0\end{bmatrix} + \mathbf{A}_{\varphi_1}\bm{x}, \quad \textnormal{or} \quad \bar{\bm{V}}(\bm{x}) = -\begin{bmatrix} 1 \\0\end{bmatrix}.
%\end{align}
Accordingly, well-posedness results are enounced in Theorem~\ref{thm:ex2} below.
\begin{theorem}[Existence and uniqueness of solutions]\label{thm:ex2}
Suppose that $\bar{\bm{V}} \in C^1(\mathscr{C};\mathbb{R}^2)$ reads as in Eq.~\eqref{eq:Vss}, and that $\mathring{\mathscr{C}}\subset \mathbb{R}^2$ is bounded, with boundary $\partial \mathscr{C}$ piecewise $C^1$. Then, for all $\tilde{\mathbf{\Sigma}}_{\varphi} \in C^0(\mathscr{C}\times[0,S];\mathbf{M}_2(\mathbb{R}))$ and $\tilde{\bm{h}} \in C^0([0,S];L^2(\mathring{\mathscr{C}};\mathbb{R}^2))$ as in Eq.~\eqref{eq:functionsH}, and ICs $\bm{z}_0 \in L^2(\mathring{\mathscr{C}};\mathbb{R}^2)$, the PDE~\eqref{eq:quasiStationary0} admits a unique mild solution $\bm{z} \in C^0([0,S];L^2(\mathring{\mathscr{C}};\mathbb{R}^2))$. Additionally, if $\tilde{\mathbf{\Sigma}}_{\varphi} \in C^1(\mathscr{C}\times[0,S];\mathbf{M}_2(\mathbb{R}))$, $\tilde{\bm{h}} \in C^1([0,S];L^2(\mathring{\mathscr{C}};\mathbb{R}^2))$, and the IC $\bm{z}_0 \in \mathscr{D}(\mathscr{A})$, with $\mathscr{D}(\mathscr{A}) \triangleq \{\bm{\zeta} \in L^2(\mathring{\mathscr{C}};\mathbb{R}^2) \mathrel{|} (\bar{\bm{V}}\cdot\nabla_{\bm{x}})\bm{\zeta} \in L^2(\mathring{\mathscr{C}};\mathbb{R}^2), \;  \eval[0]{\bm{\zeta}}_{\mathscr{L}} = \bm{0}\}$, the solution is classical, that is, $\bm{z} \in C^1([0,S];L^2(\mathring{\mathscr{C}};\mathbb{R}^2)) \cap C^0([0,S];\mathscr{D}(\mathscr{A}))$.
\begin{proof}
See Appendix~\ref{app:Proof1}.
\end{proof}
\end{theorem}
Compared to the smoothness requirements for the well-posedness of~\modref{stdmodel}, it may be realised that Theorem~\ref{thm:ex2} demands additionally $\bm{\varphi} \in \mathbb{R}\times C^0([0,S];\mathbb{R})$ for mild solutions, and $\bm{\varphi} \in \mathbb{R}\times C^1([0,S];\mathbb{R})$ (with regularisation parameter $\varepsilon \in \mathbb{R}_{>0}$) for classical solutions. 
This concludes the presentation of the three rolling contact models with FrBD friction. The next Sect.~\ref{sect:math} is devoted to investigating the stability and passivity of~\modref{stdmodel} and~\ref{linmodel2} (or, more precisely, Eqs.~\eqref{eq:standard2} and~\eqref{eq:quasiStationary0}).

\section{Mathematical properties}\label{sect:math}
The present section studies some important mathematical features of Eqs.~\eqref{eq:standard2} and~\eqref{eq:quasiStationary0}. More specifically, stability is discussed in Sect.~\ref{sect:stability}, whereas Sect.~\ref{sect:diss} focuses on dissipativity and passivity. The analysis is conducted with particular attention to the physical interpretation of these mathematical properties.

\subsection{Stability}\label{sect:stability}
Stability is discussed concerning both the input-to-state and input-to-output behaviours. For the systems described by Eqs.~\eqref{eq:standard2} and~\eqref{eq:quasiStationary0} with output~\eqref{eq:FandM} (resp.~\eqref{eq:Fundef} and~\eqref{eq:Malt}), the notions of input-to-state and {input-to-output stability are formalised according to Definitions~\ref{def:ISS} and~\ref{def:IOS} below, which are formulated for classical solutions of~\eqref{eq:standard2} and~\eqref{eq:quasiStationary0}, and regarding the rigid slip velocity $\bar{\bm{v}}(\bm{x},s)$ as an input.

\begin{definition}[Input-to-state-stability (ISS)]\label{def:ISS}
The PDE~\eqref{eq:standard2} (resp.~\eqref{eq:quasiStationary0}) is called (uniformly) \emph{input-to-state stable} (ISS) in the spatial $L^2$-norm if, for all inputs $\bar{\bm{v}} \in C^1(\mathscr{C}\times\mathbb{R}_{\geq 0};\mathbb{R}^2) \cap L^\infty(\mathscr{C}\times\mathbb{R}_{\geq 0};\mathbb{R}^2)$, rolling speeds $V\ped{r} \in C^1(\mathbb{R}_{\geq 0};[V\ped{min},V\ped{max}])$, spins $\bm{\varphi} \in \mathbb{R}\times C^1(\mathbb{R}_{\geq 0};\mathbb{R})\cap \mathbb{R}\times L^\infty(\mathbb{R}_{\geq 0};\mathbb{R})$, and ICs $\bm{z}_0 \in \mathscr{D}(\mathscr{A})$, there exist functions $\beta \in \mathcal{KL}$ and $\gamma \in \mathcal{K}_\infty$ such that
\begin{align}\label{eq:ISSbetaGamma}
\begin{split}
\norm{\bm{z}(\cdot,s)}_{L^2(\mathring{\mathscr{C}};\mathbb{R}^2)} & \leq \beta\Bigl(\norm{\bm{z}_0(\cdot)}_{L^2(\mathring{\mathscr{C}};\mathbb{R}^2)}, s\Bigr) + \gamma\Bigl( \norm{\bar{\bm{v}}(\cdot,\cdot)}_\infty\Bigr), \quad s \in \mathbb{R}_{\geq 0}.
\end{split}
\end{align}
\end{definition}

\begin{definition}[Input-to-output stability (IOS)]\label{def:IOS}
The PDE~\eqref{eq:standard2} (resp.~\eqref{eq:quasiStationary0}) with output~\eqref{eq:FandM} (resp.~\eqref{eq:Fundef} and~\eqref{eq:Malt}) is called \emph{input-to-output stable} (IOS) if, for all inputs $\bar{\bm{v}} \in C^1(\mathscr{C}\times\mathbb{R}_{\geq 0};\mathbb{R}^2) \cap L^\infty(\mathscr{C}\times\mathbb{R}_{\geq 0};\mathbb{R}^2)$, rolling speeds $V\ped{r} \in C^1(\mathbb{R}_{\geq 0};[V\ped{min},V\ped{max}])$, spins $\bm{\varphi} \in \mathbb{R}\times C^1(\mathbb{R}_{\geq 0};\mathbb{R})\cap \mathbb{R}\times L^\infty(\mathbb{R}_{\geq 0};\mathbb{R})$, and ICs $\bm{z}_0 \in \mathscr{D}(\mathscr{A})$, there exist functions $\beta \in \mathcal{KL}$ and $\alpha \in \mathcal{K}$ such that 
\begin{align}\label{eq:FBBBBBBSS}
\begin{split}
\norm{\bigl(\bm{F}_{\bm{x}}(s),M_z(s)\bigr)}_2 & \leq \beta\Bigl(\norm{\bm{z}_0(\cdot)}_{L^2(\mathring{\mathscr{C}};\mathbb{R}^2)}, s\Bigr)  + \alpha\Bigl(\norm{\bar{\bm{v}}(\cdot,\cdot)}_\infty\Bigr), \quad s\in \mathbb{R}_{\geq 0}.
\end{split}
\end{align}
\end{definition}
The above definitions, inspired from \cite{Khalil,Prieur}, are contingent on the specific systems considered in the paper; a modern, exhaustive introduction to the notions of ISS and \emph{integral input-to-state stability} (iISS) for infinite-dimensional systems may be instead found in \cite{Prieur}. It is also worth stressing that Definitions~\ref{def:ISS} and~\ref{def:IOS} have been enounced concerning the simplified PDEs~\eqref{eq:standard2} and~\eqref{eq:quasiStationary0}, which, in contrast to~\modref{stdmodel} and~\ref{linmodel2}, are postulated on fixed spatial domains and consider a constant transport velocity. Indeed, for these equations, well-posedness was established by Theorems~\ref{thm:ex1} and~\ref{thm:ex2}, which cover both classical and mild solutions. Some of the results advocated in the sequel may, however, be extended to the case of time-varying domains and transport velocities, albeit not without additional difficulties. In this context, it should be clarified that effects associated with variable normal forces and contact areas are not primarily related to the notions of stability and dissipativity for the (rolling) friction model; this is another reason why such dynamics are disregarded in the following analyses. 

Finally, it is worth commenting on the choice of identifying $\bar{\bm{v}}(\bm{x},s)$ as the main input to the PDEs~\eqref{eq:standard2} and~\eqref{eq:quasiStationary0}. Concerning specifically Eq.~\eqref{eq:standard2}, it is clear that, for a constant rolling velocity, the rigid relative velocity constitutes indeed the only input, and incorporates all the information about the translational slips $\bm{\sigma}(s)$ and the total spin $\varphi(s)$. For Eq.~\eqref{eq:quasiStationary0}, this is not true anymore, since the spin input $\bm{\varphi}(s)$ enters the PDE nonlinearly, causing additional terms to appear in the transport velocity. Additionally, $\bar{\bm{v}}\ped{r}(\bm{z}(\bm{x},s),\bm{x},s) \not = \bar{\bm{v}}(\bm{x},s)$, and, more generally, the tuple $(\bm{\sigma}(s),\bm{\varphi}(s), V\ped{r}(s))\in \mathbb{R}^5$ should be regarded as the real input to Eq.~\eqref{eq:quasiStationary0}. However, since there exists a constant $\chi \in \mathbb{R}_{>0}$ such that $\norm{\bar{\bm{v}}(\cdot,\cdot)}_\infty \leq \chi \norm{(\bm{\sigma}(\cdot),\bm{\varphi}(\cdot), V\ped{r}(\cdot))}_\infty$, the ISS and IOS inequalities of the type~\eqref{eq:ISSbetaGamma} and~\eqref{eq:FBBBBBBSS} may always be reformualted in terms of $(\bm{\sigma}(s),\bm{\varphi}(s), V\ped{r}(s))$.

With these premises, ISS and IOS are enounced respectively in Sects.~\ref{sect:InSD0D} and~\ref{sect:IOS}.
%The first part of the analysis concerns the stability properties of the solution to the PDE~\eqref{eq:LuGreDistr}. Lemma~\ref{lemma:boundedness} states that bounded sliding velocities produce bounded deformations. The result is established in the spatial $L^2$-norm for the distributed state.
%\begin{definition}[Integral input-to-state-stability (iISS)]
%The PDE~\eqref{eq:PDEfric} is called \emph{integral input-to-state stable} (iISS) in the spatial $L^2$-norm if, for all inputs $v \in C^0(\mathbb{R}_{\geq 0};\mathbb{R}) \cap L^2(\mathbb{R}_{\geq 0};\mathbb{R})$ and ICs $z_0 \in L^2((0,1);\mathbb{R})$, there exist functions $\beta \in \mathcal{KL}$ and $\gamma \in \mathcal{K}_\infty$ such that
%\begin{align}\label{eq:ISSbetaGammaI}
%\begin{split}
%\norm{z(\cdot,t)}_{L^2((0,1);\mathbb{R})} & \leq \beta\Bigl(\norm{z_0(\cdot)}_{L^2((0,1);\mathbb{R})}, t\Bigr) \\
%& \quad + \gamma\Biggl( \int_0^t \norm{v(t^\prime)}_2^2 \dif t^\prime \Biggr), \quad t \in \mathbb{R}_{\geq 0}.
%\end{split}
%\end{align}
%\end{definition}
\subsubsection{Input-to-state stability}\label{sect:InSD0D}

Concerning both the PDEs~\eqref{eq:standard2} and~\eqref{eq:quasiStationary0}, ISS results are asserted by Lemma~\ref{lemma:boundedness} below.

%In particular, when $\bm{\varepsilon}(s) = \bm{\varepsilon}$ and $\varphi_1(s) = \varphi_1$ are constant over travelled distance, it is always possible to reorient the reference frame so that $\varepsilon_y = 1$ and $\varepsilon_x = 0$ in Eq.~\eqref{eq:Vss}.

\begin{lemma}[Input-to-state stability (ISS)]\label{lemma:boundedness}
Suppose that $\mathscr{C} \subset \mathbb{R}^2$ is nonempty and compact. Then, the PDE~\eqref{eq:standard2} is (uniformly) ISS in the spatial $L^2$-norm. For the PDE~\eqref{eq:quasiStationary0}, assume additionally that the point $C_{\varphi_1}$ with coordinate $\bm{x}_{C_{\varphi_1}} = \bm{\varepsilon}/\varphi_1$ lies outside $\mathscr{C}$, and that $\mathscr{C}$ is strictly separable from $C_{\varphi_1}$ by a line, i.e., there exists a vector $\mathbb{R}^2 \ni \bar{\bm{\nu}} = [\bar{\nu}_x\; -\bar{\nu}_y]^{\mathrm{T}}\not = \bm{0}$ and a constant $\bar{\eta} \in \mathbb{R}_{>0}$ such that
\begin{align}\label{eq:N}
\bar{N}(\bm{x}) \triangleq \bar{\bm{\nu}}\cdot \bigl(\bm{x}-\bm{x}_{C_{\varphi_1}}\bigr) < -\bar{\eta}, \quad \bm{x} \in \mathscr{C}.
\end{align}
Then, the PDE~\eqref{eq:quasiStationary0} is also (uniformly) ISS in the spatial $L^2$-norm. 
\begin{proof}
See Appendix~\ref{app:Proof2}.
\end{proof}
\end{lemma}
The separation condition expressed by Eq.~\eqref{eq:N} in Lemma~\ref{lemma:boundedness} is reminiscent of the discussion already initiated in Sect.~\ref{sect:semilinear}. As noted earlier, the instantaneous tilting centre $C_{\varphi_1}$ (which in this case is fixed) must lie outside the contact area. If this holds, and if one can draw a line through $C_{\varphi_1}$ that does not intersect $\mathscr{C}$, then the PDE~\eqref{eq:quasiStationary0} can be shown to be ISS. In essence, this condition ensures that spinning effects do not dominate the rolling dynamics. Of course, the PDE~\eqref{eq:standard2} implicitly assume small spin slips, and therefore the separation condition of Eq.~\eqref{eq:N} is not formally needed.

Lemma~\ref{lemma:boundedness} is also instrumental in showing the IOS of the PDEs~\eqref{eq:standard2} and~\eqref{eq:quasiStationary0}\footnote{In this context, it should be observed that, whilst this paper is mainly concerned with ISS and IOS, suitable iISS estimates \cite{Prieur} may also be derived .}, which is addressed next in Sect.~\ref{sect:IOS}.

\subsubsection{Input-to-output stability}\label{sect:IOS}
Intuitively, bounded bristle deformations should also produce bounded tangential forces and vertical moments. This is the interpretation of the IOS bound~\eqref{eq:FBBBBBBSS}. For coherence, in what follows, it is assumed that $p(\bm{x},s) = p(\bm{x})$, with $p \in C^0(\mathscr{C};\mathbb{R}_{\geq 0})$ (the contact area is automatically supposed to be fixed over time, or equivalently, travelled distance, as required by the PDEs~\eqref{eq:standard2} and~\eqref{eq:quasiStationary0}).
Accordingly, from Eq.~\eqref{eq:Fundef} and the conditions $\bar{\bm{v}} \in C^0(\mathscr{C}\times\mathbb{R}_{\geq 0};\mathbb{R}^2) \cap L^\infty(\mathscr{C}\times\mathbb{R}_{\geq 0};\mathbb{R}^2)$, $V\ped{r} \in C^0(\mathbb{R}_{\geq 0};[V\ped{min},V\ped{max}])$, and $\bm{\varphi} \in \mathbb{R}\times C^0(\mathbb{R}_{\geq 0};\mathbb{R})\cap \mathbb{R}\times L^\infty(\mathbb{R}_{\geq 0};\mathbb{R})$, it is possible to infer the existence of constants $\beta_1,\alpha_1 \in \mathbb{R}_{>0}$ such that
\begin{align}\label{eq:Fbound}
\norm{\bm{F}_{\bm{x}}(s)}_2 \leq \beta_1\norm{\bm{z}(\cdot,s)}_{L^2(\mathring{\mathscr{C}};\mathbb{R}^2)} + \alpha_1\norm{\bar{\bm{v}}(\cdot,\cdot)}_\infty, \quad s\in \mathbb{R}_{\geq 0}.
\end{align}
Similarly, from Eqs.~\eqref{eq:Mzunderfr} and~\eqref{eq:Malt}, it follows the existence of $\beta_2,\alpha_2 \in \mathbb{R}_{>0}$ and $\beta_3,\alpha_3 \in \mathbb{R}_{\geq0}$ satisfying
\begin{align}\label{eq:Mbound}
\abs{M_z(s)} & \leq \beta_2\norm{\bm{z}(\cdot,s)}_{L^2(\mathring{\mathscr{C}};\mathbb{R}^2)} + \beta_3\norm{\bm{z}(\cdot,s)}_{L^2(\mathring{\mathscr{C}};\mathbb{R}^2)}^2 +\alpha_2\norm{\bar{\bm{v}}(\cdot)}_\infty+\alpha_3\norm{\bar{\bm{v}}(\cdot,\cdot)}_\infty^2, \quad s \in \mathbb{R}_{\geq 0},
\end{align}
with $\beta_3 = \alpha_3 = 0$ for~\eqref{eq:Mzunderfr}.

Combining Eqs.~\eqref{eq:Fbound} and~\eqref{eq:Mbound} with~\eqref{eq:ISSbetaGamma}, and exploiting the fact that all the norms are equivalent in finite-dimensional spaces, gives IOS for all $\bm{z}(\cdot,s) \in L^2(\mathring{\mathscr{C}};\mathbb{R}^2)$. This simple result is formalised in Corollary~\ref{corollary:IOS} below.
\begin{corollary}[Input-to-output stability (IOS)]\label{corollary:IOS}
Suppose that $p(\bm{x},s) = p(\bm{x})$, with $p \in C^0(\mathscr{C};\mathbb{R}_{\geq 0})$. Then, under the corresponding assumptions as Lemma~\ref{lemma:boundedness}, the PDE~\eqref{eq:standard2} (resp.~\eqref{eq:quasiStationary0}) with output~\eqref{eq:FandM} (resp.~\eqref{eq:Fundef} and~\eqref{eq:Malt}) is IOS.
\end{corollary}
\begin{remark}\label{Remark:oododo}
As explained in detail in \cite{DistrLuGre}, replacing the total derivatives appearing in Eqs.~\eqref{eq:fxt} and~\eqref{eq:fxs} with the partial one would preclude a similar result to that asserted by Corollary~\ref{corollary:IOS}. 
\end{remark}
Remark~\ref{Remark:oododo} concludes the stability analysis of the systems described by Eqs.~\eqref{eq:standard2} and~\eqref{eq:quasiStationary0}. The following Sect.~\ref{sect:diss} examines their dissipativity and passivity properties.

\subsection{Dissipativity and passivity}\label{sect:diss}

The notion of passivity relates to important physical properties of a friction model. Before stating its mathematical definition in the context of this paper, it is beneficial to clarify certain aspects regarding the power dissipation at the interface between the contacting bodies. This is done with respect to~\modref{stdmodel} and~\ref{linmodel2} separately. For simplicity, in the sequel, the contact area is supposed to be fixed, although the conclusions of the following analysis are formally independent of this assumption. 
First, \modref{stdmodel} with output~\eqref{eq:FandM} is considered. Accordingly, the dissipated power per unit of speed may be calculated as
\begin{align}\label{eq:Ppass}
\begin{split}
-\bigl\langle p(\bm{x})\bm{f}(\cdot,s), \bar{\bm{v}}\ped{r}(\cdot,s)\bigr\rangle_{L^2(\mathring{\mathscr{C}};\mathbb{R}^2)} & = -\bigl\langle p(\bm{x})\bm{f}(\cdot,s), \bar{\bm{v}}(\cdot,s)\bigr\rangle_{L^2(\mathring{\mathscr{C}};\mathbb{R}^2)} =- \iint_{\mathscr{C}} p(\bm{x})\bm{f}^{\mathrm{T}}(\bm{x},s)\bar{\bm{v}}(\bm{x},s) \dif \bm{x} \\
& = \iint_{\mathscr{C}}p(\bm{x})\bm{f}^{\mathrm{T}}(\bm{x},s)\bm{\sigma}(s) \dif \bm{x} + \iint_{\mathscr{C}}p(\bm{x})\bigl[xf_y(\bm{x},s)-yf_x(\bm{x},s)\bigr]\varphi(s) \dif \bm{x} \\
& = \bm{F}_{\bm{x}}^{\mathrm{T}}(s)\bm{\sigma}(s) + M_z(s)\varphi(s), \quad s \in [0,S].
\end{split}
\end{align}
Equation~\eqref{eq:Ppass} yields the expected result, consistent with previous literature: the total power dissipated is given by the product between the tangential forces and vertical moment, $(\bm{F}_{\bm{x}}(s), M_z(s)) \in \mathbb{R}^3$, and the translational and total spin slips, $(\bm{\sigma}(s), \varphi(s)) \in \mathbb{R}^3$. This outcome is unsurprising, as under the assumption of sufficiently small spins, the slip and spin variables indeed act as the dual quantities of the tangential forces and vertical moment. It is worth emphasising that the final equivalence in Eq.~\eqref{eq:Ppass} was obtained by evaluating the vertical equilibrium in the undeformed configuration. Conversely, when the bristle deformation $\bm{z}(\bm{x},s)$ is relatively large, and the deformed configuration can be identified uniquely, the vertical equilibrium may be evaluated using~\eqref{eq:Malt}. Considering again~\modref{stdmodel}, this leads to
\begin{align}\label{eq:Ppass0}
\begin{split}
-\bigl\langle p(\bm{x})\bm{f}(\cdot,s), \bar{\bm{v}}\ped{r}(\cdot,s)\bigr\rangle_{L^2(\mathring{\mathscr{C}};\mathbb{R}^2)} & =  \bm{F}_{\bm{x}}^{\mathrm{T}}(s)\bm{\sigma}(s) + M_z(s)\varphi(s) \\
& \quad - \varphi(s)\iint_{\mathscr{C}}p(\bm{x})\bigl[z_x(\bm{x},s)f_y(\bm{x},s)-z_y(\bm{x},s)f_x(\bm{x},s)\bigr] \dif \bm{x}, \quad s \in [0,S],
\end{split}
\end{align}
which removes a dissipation term proportional to the total spin $\varphi(s)$ (additional spin contributions felt by the bristle's tip have been neglected).

Employing instead~\modref{linmodel2}, and considering situations in which one of the two bodies is rigid, such as in the case of tyre-road interaction ($\varphi_1(s) = \varphi_\gamma(s)$ and $\varphi_2(s) = \varphi_\psi(s)$), the vertical equilibrium may be evaluated using either Eq.~\eqref{eq:FandM} or~\eqref{eq:Fundef} and~\eqref{eq:Malt}. Specifically, considering~\modref{linmodel2} with output~\eqref{eq:FandM}, similar calculations as previously yield
\begin{align}\label{eq:Ppass2}
\begin{split}
& -\bigl\langle p(\bm{x})\bm{f}(\cdot,s), \bar{\bm{v}}\ped{r}(\bm{z}(\cdot,s),\cdot,s)\bigr\rangle_{L^2(\mathring{\mathscr{C}};\mathbb{R}^2)} 
= \bm{F}_{\bm{x}}^{\mathrm{T}}(s)\bm{\sigma}(s) + M_z(s)\varphi(s) \\
& \qquad \qquad \quad + \varphi_\psi(s)\iint_{\mathscr{C}}p(\bm{x})\bigl[z_x(\bm{x},s)f_y(\bm{x},s)-z_y(\bm{x},s)f_x(\bm{x},s)\bigr] \dif \bm{x}, \quad s \in [0,S].
\end{split}
\end{align}
On the other hand, when the vertical equilibrium is computed according to Eqs.~\eqref{eq:Fundef} and~\eqref{eq:Malt}, the following expression is obtained:
\begin{align}\label{eq:Ppass3}
\begin{split}
&- \bigl\langle p(\bm{x})\bm{f}(\cdot,s), \bar{\bm{v}}\ped{r}(\bm{z}(\cdot,s),\cdot,s)\bigr\rangle_{L^2(\mathring{\mathscr{C}};\mathbb{R}^2)} 
= \bm{F}_{\bm{x}}^{\mathrm{T}}(s)\bm{\sigma}(s) + M_z(s)\varphi(s) \\
& \qquad \qquad \quad - \varphi_\gamma(s)\iint_{\mathscr{C}}p(\bm{x})\bigl[z_x(\bm{x},s)f_y(\bm{x},s)-z_y(\bm{x},s)f_x(\bm{x},s)\bigr] \dif \bm{x}, \quad s \in [0,S].
\end{split}
\end{align}
Equations~\eqref{eq:Ppass2} and~\eqref{eq:Ppass3}, which apply when one of the bodies in rolling contact is rigid, differ in how they treat the geometric and effective spin terms. In essence, Eq.~\eqref{eq:Ppass2} includes an additional contribution from the turn spin $\varphi_\psi(s)$, arising from the fact that the bristle tips experience a relative velocity component associated with an effective rotation about the vertical axis.
In contrast, Eq.~\eqref{eq:Ppass3} states that the geometric spin $\varphi_\gamma(s)$ does not dissipate power in the undeformed configuration, since it affects only the trajectory of the bristles within the contact patch and is therefore related to the velocity of their roots rather than their tips (which only feel the effect of $\varphi_\psi(s)$).
Both Eq.~\eqref{eq:Ppass2} and~\eqref{eq:Ppass3} are simultaneously valid; the difference lies solely in the manner in which the vertical moment is evaluated. It is also worth remarking that, for~\modref{stdmodel}, the equivalence $\bar{\bm{v}}\ped{r}(\bm{z}(\bm{x},s),\bm{x},s) = \bar{\bm{v}}(\bm{x},s)$ only holds when $\varphi_2(s) = 0$. Since, from a physical viewpoint, passivity should be regarded as a property of the friction model in isolation, it appears reasonable to consider the case $\varphi_2(s) = 0$ when analysing~\modref{stdmodel}. Indeed, when $\varphi_2(s) \not = 0$, \modref{stdmodel} may be interpreted as a feedback system where the rigid relative velocity depends upon the state through a feedback term $\mathbf{A}_{\varphi_2}(s)\bm{z}(\bm{x},s)$. In this context, a question that naturally arises is whether all the nondimensional velocities $\bar{\bm{v}}(\bm{x},s)$, rolling speeds $V\ped{r}(s)$, and spins $\bm{\varphi}(s)$ (with $\varphi_2(s) = 0$ in Eq.~\eqref{eq:quasiStationary0}) ensuring the well-posedness of the PDEs~\eqref{eq:standard2} and~\eqref{eq:quasiStationary0}, as asserted by Theorems~\ref{thm:ex1} and~\ref{thm:ex2}, imply automatically $-\langle p(\bm{x})\bm{f}(\cdot,s), \bar{\bm{v}}\ped{r}(\cdot,s)\rangle_{L^2(\mathring{\mathscr{C}};\mathbb{R}^2)} \geq 0$ for all $s \in [0,S]$ in~\eqref{eq:Ppass}.  In particular, restricted to classical solutions of the equations under consideration, the concepts of dissipativity and passivity are enounced according to Definition~\ref{def:idss2} below. A more general introduction to these notions, limited to linear systems, is offered in \cite{Zwart}.

%\subsubsection{Dissipativity and passivity}
\begin{definition}[Dissipativity and passivity]\label{def:idss2}
The system described by the PDE~\eqref{eq:standard2} (resp.~\eqref{eq:quasiStationary0}) with output~\eqref{eq:Fundef} (resp.~\eqref{eq:Fundef} and~\eqref{eq:Malt}) is called \emph{dissipative} if, for all inputs $\bar{\bm{v}} \in C^1(\mathscr{C}\times\mathbb{R}_{\geq 0};\mathbb{R}^2) \cap L^\infty(\mathscr{C}\times\mathbb{R}_{\geq 0};\mathbb{R}^2)$, rolling speeds $V\ped{r} \in [V\ped{min},V\ped{max}]$, spins $\bm{\varphi} \in \mathbb{R}\times \{0\}$, and ICs $\bm{z}_0 \in \mathscr{D}(\mathscr{A})$, there exist a supply rate $w : L^2(\mathring{\mathscr{C}};\mathbb{R}^2)\times L^2(\mathring{\mathscr{C}};\mathbb{R}^2) \to \mathbb{R}$ and a storage function $W : L^2(\mathring{\mathscr{C}};\mathbb{R}^2) \to \mathbb{R}_{\geq 0}$ such that 
\begin{align}\label{eq:Fvres}
\begin{split}
&\int_0^s w\bigl(\bm{f}(\cdot,s^\prime),-\bar{\bm{v}}\ped{r}(\cdot,s^\prime)\bigr) \dif s^\prime = \int_0^s w\bigl(\bm{f}(\cdot,s^\prime),-\bar{\bm{v}}(\cdot,s^\prime)\bigr) \dif s^\prime \geq W\bigl(\bm{z}(\cdot,s)\bigr)-W\bigl(\bm{z}_0(\cdot)\bigr), \quad s \in [0,S].
\end{split}
\end{align}
It is called \emph{passive} if $w(\bm{f}(\cdot,s),-\bar{\bm{v}}(\cdot,s)) =-\langle p(\cdot)\bm{f}(\cdot,s), \bar{\bm{v}}\ped{r}(\cdot,s)\rangle_{L^2(\mathring{\mathscr{C}};\mathbb{R}^2)}= -\langle p(\cdot)\bm{f}(\cdot,s), \bar{\bm{v}}(\cdot,s)\rangle_{L^2(\mathring{\mathscr{C}};\mathbb{R}^2)}$.
\end{definition}
It is important to clarify that the notion of passivity, as formalised in Definition~\ref{def:idss2}, pertains to the rolling contact process rather than the friction model in isolation. Nonetheless, since rolling contact dynamics are directly governed by frictional interactions, the system should be expected to dissipate energy. This observation provides a physical rationale for the passivity property: the friction-induced dissipation ensures that the rate of energy supplied to the system is bounded below by the time derivative of an appropriate storage function.
The main result in this sense is formalised in Lemma~\ref{lemma:DissF20} below.
\begin{lemma}[Passivity]\label{lemma:DissF20}
Suppose that $p \in C^1(\mathscr{C};\mathbb{R}_{\geq 0})$ satisfies
\begin{align}\label{eq:condP}
\nabla_{\bm{x}}\cdot p(\bm{x})\bar{\bm{V}}(\bm{x}) \leq 0, \quad \bm{x} \in \mathscr{C},
\end{align}
and, for all $\bm{y}\in \mathbb{R}^2$, $\mathbf{M}^2(\bm{y}) \mathbf{G}^{-1}(\bm{y})\mathbf{\Sigma}_1 \succeq \mathbf{0}$.
Then, the system described by the PDE~\eqref{eq:standard2} (resp.~\eqref{eq:quasiStationary0}) with output~\eqref{eq:Fundef} (resp.~\eqref{eq:Fundef} and~\eqref{eq:Malt}) is passive with storage function
\begin{align}\label{eq:VdissF2}
W\bigl(\bm{z}(\cdot,s)\bigr) \triangleq \dfrac{1}{2}\iint_{\mathscr{C}} p(\bm{x}) \bm{z}^{\mathrm{T}}(\bm{x},s)\mathbf{\Sigma}_0\bm{z}(\bm{x},s) \dif \bm{x}. 
\end{align}
\begin{proof}
See Appendix~\ref{app:Proof3}.
\end{proof}
\end{lemma}
The inequality in Eq.~\eqref{eq:condP} is similar to those derived for the one-dimensional variants of the distributed LuGre and FrBD models in \cite{DistrLuGre,FrBD}, confirming the importance of the vertical pressure distribution in ensuring passivity properties. For instance, restricting the attention to the line contact case, as well as a constant nondimensional transport velocity as in~\modref{stdmodel}, it may be easily concluded from~\eqref{eq:condP} that non-compactly-supported pressure distributions that are either constant or exponentially decreasing in the longitudinal direction qualify as valid choices, with the latter type even rendering the model \emph{strictly dissipative} (the inequality~\eqref{eq:Fvres} is satisfied in a strict sense). On the other hand, for a parabolic pressure distribution vanishing on the boundary of the contact area, it is possible to show that Eq.~\eqref{eq:Fvres} cannot be verified. Once again, it is essential to observe that the result advocated in Lemma~\ref{lemma:DissF20} requires adopting the total time-like derivative of the distributed state in Eqs.~\eqref{eq:fxt} and~\eqref{eq:fxs}. Before moving to the numerical experiments reported in Sect.~\ref{sect:numer}, some concluding remarks are formalised below. 

\begin{remark}
The condition $\mathbf{M}^2(\bm{y}) \mathbf{G}^{-1}(\bm{y})\mathbf{\Sigma}_1 \succeq \mathbf{0}$ for all $\bm{y}\in \mathbb{R}^2$ in Lemma~\ref{lemma:DissF20} is trivially verified when the matrices $\mathbf{\Sigma}_1$ and $\mathbf{M}(\bm{y})$ are diagonal or commute, as common in engineering practice. When the condition $\mathbf{M}^2(\bm{y}) \mathbf{G}^{-1}(\bm{y})\mathbf{\Sigma}_1 \succeq \mathbf{0}$ is violated, passivity might not hold, due to the spurious coupling effects introduced between the longitudinal and lateral deformations of the bristles and the produced frictional forces, similar to what happens in the presence of large spin slips $\varphi_2$. However, dissipativity may be more generally shown to hold for any combinations of $\mathbf{\Sigma}_1$ and $\mathbf{M}(\bm{y})$, albeit being less interesting from a physical perspective. The inequality~\eqref{eq:condP} is obviously satisfied when the pressure distribution is constant (for ~\modref{stdmodel}, distributions decreasing along the rolling direction also verify Eq.~\eqref{eq:condP} strictly). Whilst admittedly artificial, the use of spatially varying pressure profiles fulfilling~\eqref{eq:condP} seems necessary to reproduce the sign reversal of the steady-state vertical moment at large lateral slips\footnote{Asymmetries in the vertical pressure distribution can arise in viscoelastic rolling contact systems due to internal dissipation effects. This behavior is typically observed, for example, in rubber tyres.}. Conversely, the condition~\eqref{eq:condP} is typically violated by compactly supported pressure distributions such as parabolic and trapezoidal ones. In this context, it is however important to clarify that the friction model remains always passive: it is the rolling contact process, as a macroscopic phenomenon, that may become non-passive. A more detailed discussion is provided in Sect.~\ref{sect:miscellaneous} and in \cite{Tsiotras3,Deur1}. In the end, the requirements imposed by Lemma~\ref{lemma:DissF20} appear very mild. 
\end{remark}

\begin{remark}
For simplicity, inequalities of the type~\eqref{eq:ISSbetaGamma},~\eqref{eq:FBBBBBBSS}, and~\eqref{eq:Fvres} were proven only for classical solutions ($\varepsilon \in \mathbb{R}_{>0}$ in~\eqref{eq:standard2} and~\eqref{eq:quasiStationary0}); the extension to mild solutions that may be more generally obtained for $\varepsilon \in \mathbb{R}_{\geq0}$ may be established with the aid of more involved (and rather tedious) mathematical machinery, and is not pursued here. Besides, in the context of this paper, these are inessential technicalities that would obfuscate the main message of the manuscript rather than enhancing its readability.
\end{remark}

\section{Numerical simulations}\label{sect:numer}
In the following,~\modref{stdmodel},~\ref{semilinmodel}, and~\modref{linmodel2} are numerically simulated and compared concerning both their steady-state and transient behaviours. More specifically, Sect.~\ref{sect:SS} is dedicated to steady rolling contact, whereas unstationary phenomena are investigated in Sect.~\ref{sect:transient}. Implementational and computational details are reported in Appendix~\ref{sect:details}.

\subsection{Steady-state behaviour}\label{sect:SS}
A comprehensive understanding of stationary rolling contact is crucial for explaining how slip and spin affect steady forces and moments, whilst also providing the conceptual tools needed to untangle the subtle and often more convoluted effects observed during transients. The next Sects.~\ref{sect:slipSurf} and~\ref{sect:Action} discuss in detail the steady-state behaviours of the proposed rolling contact models relatively to the notions of \emph{force-slip surfaces} and \emph{action surfaces}, respectively. Finally, Sect.~\ref{sect:miscellaneous} concludes by presenting a miscellaneous assortment of numerical results.

For simplicity, in the sequel, the attention is restricted to the configuration depicted in Fig.~\ref{fig:LumpModel}(a), where the substrate is supposed to be infinitely rigid. In this case, $\bm{\varphi}(s) = [\varphi_1(s) \; \varphi_2(s)]^{\mathrm{T}}\equiv [\varphi_\gamma(s) \; \varphi_\psi(s)]^{\mathrm{T}}$, $\varphi(s) = \varphi_\gamma(s) + \varphi_\psi(s) \equiv\varphi_1(s) + \varphi_2(s)$, and therefore the vector $\bm{\varphi}(s)$ fully characterises the spin state. 

\subsubsection{Force-slip surfaces}\label{sect:slipSurf}
In steady-state conditions, the main interest typically consists in deducing analytical relationships between the tangential forces and vertical moment, $(\bm{F}_{\bm{x}},M_z) \in \mathbb{R}^3$, and the slip and spin inputs, $(\bm{\sigma},\bm{\varphi}) \in \mathbb{R}^4$, in the form
\begin{subequations}
\begin{align}
\bm{F}_{\bm{x}} & = \hat{\bm{F}}_{\bm{x}}(\bm{\sigma},\bm{\varphi}), \\
M_z & = \hat{M}_z(\bm{\sigma},\bm{\varphi}), 
\end{align}
\end{subequations}
where $(\hat{\bm{F}}_{\bm{x}},\hat{M}_z) : \mathbb{R}^4 \to \mathbb{R}^3$ may be generally referred to as the \emph{force-slip surfaces}. Indeed, for a given value of the rolling speed $V\ped{r}$, regarding the slip and spin inputs $(\bm{\sigma},\bm{\varphi})$ as independent variables, each steady-state characteristic, namely the longitudinal force $F_x$, the lateral force $F_y$, and the vertical moment $M_z$, may be interpreted as a four-dimensional surface living in a five-dimensional space\footnote{The mapping $(\bm{\sigma},\bm{\varphi}) \mapsto (\bm{F}_{\bm{x}},M_z)$ is from $\mathbb{R}^4$ into $\mathbb{R}^3$. Thus, considering, for instance, $F_x \in \mathbb{R}$, $\mathrm{graph}(\bm{\sigma},\bm{\varphi},F_x)\subset \mathbb{R}^5$ lives in a five-dimensional space. Regarding also $V\ped{r} \in \mathbb{R}$ as a variable, $\mathrm{graph}(\bm{\sigma},\bm{\varphi},V\ped{r},F_x)\subset \mathbb{R}^6$ lives instead in a six-dimensional space. The situation is analogous for the other characteristics.}. Typically, the force-slip surfaces are usually continuous, uniformly bounded, and at least locally $C^1$. In this context, it is worth observing that the rolling speed itself could be regarded as an input; in which case, the force-slip surfaces should be more properly interpreted as five-dimensional surfaces embedded in a six-dimensional space. However, most friction and rolling contact models assume a minor dependence upon $V\ped{r}$ compared to the slips and spins. 

In the absence of spin ($\bm{\varphi} = \bm{0}$), the force-slip surfaces generated for a rectangular and elliptical contact area with parabolic pressure distribution, modelled as explained in Appendix~\ref{app:Patch}, are illustrated in Fig.~\ref{fig:SlipSurfaces}. More specifically, the surfaces plotted in Fig.~\ref{fig:SlipSurfaces} were obtained by sampling slip values uniformly in $[-1,1]^2$, and computing numerically the corresponding tangential forces and vertical moment using~\modref{stdmodel} (all the models presented in the paper are equivalent for $\bm{\varphi} = \bm{0}$), with the parameter values listed in Table~\ref{tab:parameters}. As already known from a substantial body of literature on tyre dynamics, the tangential forces exhibit similar trends versus the respective spin, whereas the vertical moment is mainly affected by the lateral component $\sigma_y$. The surfaces plotted in Fig.~\ref{fig:SlipSurfaces} are similar to those obtained by employing other descriptions of steady rolling contact, such as the brush models with Coulomb-Amontons friction. In this context, it is worth observing that the force-slip surfaces in Fig.~\ref{fig:SlipSurfaces} were produced considering an isotropic rolling contact system, whereas the adoption of anisotropic formulations would introduce asymmetries in the trends of the tangential forces (for instance, tyres are much stiffer in the longitudinal direction). 

More detailed plots are reported in Sect.~\ref{sect:miscellaneous}.
\begin{figure}
\centering
\includegraphics[width=1\linewidth]{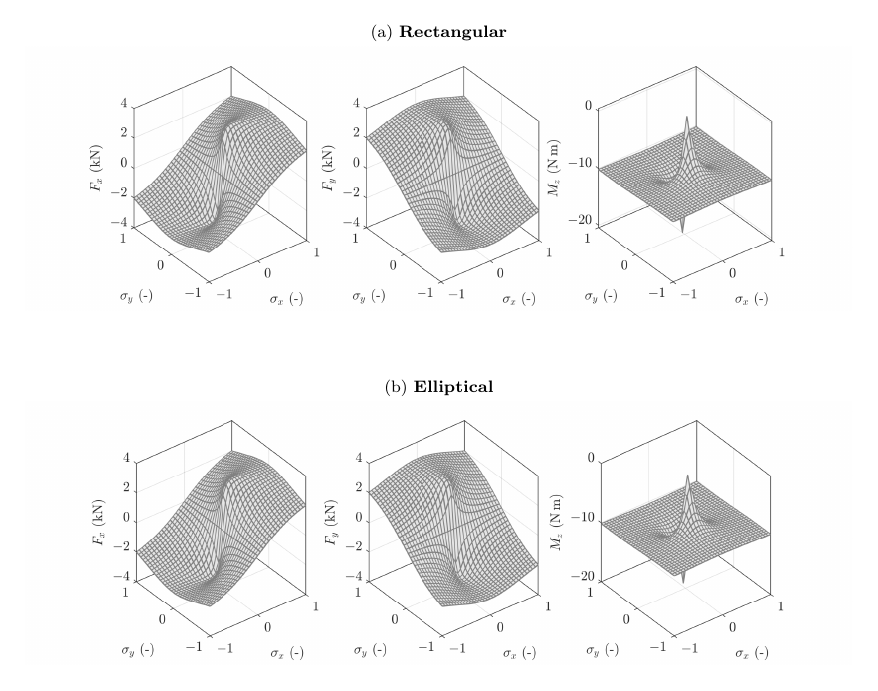} 
\caption{Force-slip surfaces in the absence of spin slips (parabolic pressure distribution): (a) rectangular contact area; (b) elliptical contact area. Model parameters as in Table~\ref{tab:parameters}.}
\label{fig:SlipSurfaces}
\end{figure}

\begin{table}[h!]\centering 
\caption{Model parameters. Matrices $\mathbf{\Sigma}_0$ and $\mathbf{\Sigma}_1$ as in Eq.~\eqref{eq:Sigmas2}; $\mathbf{M}(\bm{y}) = \mu(\bm{y})\mathbf{I}_2$ with $\mu(\bm{y})$ as in Eq.~\eqref{eq:muExample}}
{\begin{tabular}{|c|l|c|c|}
\hline
Parameter & Description & Unit & Value \\
\hline 
$a$ & Contact area semilength & m & 0.075\\
$b$ & Contact area semiwidth & m & 0.05\\
$\sigma_{0x}$ &Longitudinal micro-stiffness & $\textnormal{m}^{-1}$ & 240 \\
$\sigma_{0y}$ & Lateral micro-stiffness & $\textnormal{m}^{-1}$ & 240 \\
$\sigma_{1x}$ &Longitudinal micro-damping & $\textnormal{s}\,\textnormal{m}^{-1}$ & 0 \\
$\sigma_{1y}$ & Lateral micro-damping & $\textnormal{s}\,\textnormal{m}^{-1}$ & 0 \\
$\mu\ped{s}$ & Static friction coefficient & -& 1\\
$\mu\ped{d}$ & Dynamic friction coefficient & -& 0.7\\
$\mu\ped{v}$ & Viscous friction coefficient & -& 0\\
$v\ped{S}$ & Stribeck velocity & $\textnormal{m}\,\textnormal{s}^{-1}$ & 3.49 \\
$\delta\ped{S}$ & Stribeck exponent & - & 0.6 \\
$F_{z}$ & Vertical force & N & 4000 \\
$\varepsilon$ & Regularisation parameter & $\textnormal{m}^2\,\textnormal{s}^{-2}$ & $10^{-12}$ \\ 
\hline
\end{tabular} }
\label{tab:parameters}
\end{table}

\subsubsection{Action surfaces}\label{sect:Action}
%For example, studies show that for a planar surface undergoing both sliding and spinning, the interaction between tangential and angular friction forces reduces the effort required to push it tangentially compared to when it has no angular velocity.

The interaction between the steady-state characteristics may also be studied by considering the implicit surface
\begin{align}\label{eq:ActionSurf}
\hat{\bm{F}}(\bm{F}_{\bm{x}}, M_z) = 0,
\end{align} 
which describes the vertical moment as a function of the tangential forces. Inspired by the terminology adopted in \cite{Guiggiani} (Chap. 11), in this paper, $\hat{\bm{F}}(\bm{F}_{\bm{x}}, M_z) = 0$ is referred to as the \emph{action surface}. Mathematically, the notion of action surface appears strictly related to that of \emph{limit surface} theorised in \cite{Goyal1,Goyal2}; however, it is essential to emphasise that, within the FrBD and LuGre frameworks, there is no clear distinction between stick and slip regimes, as both formulations only attempt to approximate the average (sliding) dynamics of a bristle element. Nevertheless, Eq.~\eqref{eq:ActionSurf} yields an equivalent principle to investigate the coupling between the macroscopic forces and torque generated during steady rolling. 

The action surfaces obtained for a rectangular and elliptical contact area by simulating numerically~\modref{semilinmodel} are illustrated in Fig.~\ref{fig:ActionSurfaces} for different combinations of spin slips $\bm{\varphi} = \bm{0}$, $(3,0)$, and $(0,3)$, respectively. Inspection of Fig.~\ref{fig:ActionSurfaces} reveals that large spin slips may profoundly influence the shape of the action surface. In particular, the parameters $\varphi_1 = \varphi_\gamma$ and $\varphi_2 = \varphi_\psi$ seem to produce similar effects, with the latter causing a slightly more pronounced asymmetrisation. Generally speaking, the plots reported in Fig.~\ref{fig:ActionSurfaces} confirm the important coupling between the tangential forces and the vertical moment in rolling contact.

\begin{figure}
\centering
\includegraphics[width=1\linewidth]{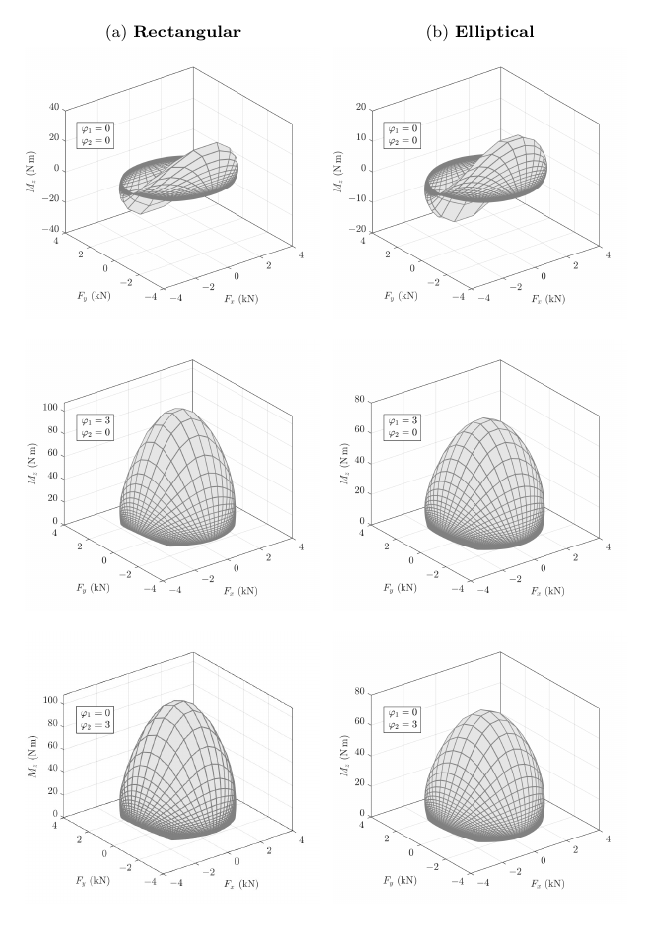} 
\caption{Action surfaces for different combinations of spin slips $\bm{\varphi} = (\varphi_1,\varphi_2)$ (parabolic pressure distribution): (a) rectangular contact area; (b) elliptical contact area. Model parameters as in Table~\ref{tab:parameters}.}
\label{fig:ActionSurfaces}
\end{figure}

In this context, of particular importance are the projections of the action surface on the plane $F_x-Fy$, which are depicted in Fig.~\ref{fig:Circles} for $\bm{\varphi} = \bm{0}$ (solid thick lines), $\bm{\varphi} = (3,0)$ (solid lines), and $\bm{\varphi} = (0,3)$ (dashed lines). In tyre dynamics, these projections are widely known as \emph{friction circles} \cite{Guiggiani}. The friction circles are typically obtained under isotropic assumptions, whereas in the presence of anisotropy (e.g., $\sigma_{0x} \not = \sigma_{0y}$ in Eq.~\eqref{eq:Sigmas02}) a \emph{friction ellipse} is more generally obtained \cite{Guiggiani}.

Representing steady-state characteristics through action surfaces is particularly valuable for rolling contact systems operating close to the limit of their capabilities. In vehicle dynamics, for instance, control algorithms are often designed to keep vehicles at the limits of handling, where tyres exploit the full friction potential available at the contact patch. Neglecting the contribution of the vertical moment, this corresponds to operating along the boundary of the friction circles, as illustrated in Fig.~\ref{fig:Circles},  by prescribing a desired acceleration direction. Such strategies are especially effective when the tyre-road friction coefficient is unknown, but maximum performance is required, as typically happens in emergency manoeuvres \cite{Fors1,Fors2,Fors3,Fors4}. These algorithms build on concepts closely related to those used in robotics, particularly the notion of limit surfaces. Recent work has extended this framework to non-Coulomb friction models, including LuGre and elastoplastic formulations \cite{2D}.
\begin{figure}
\centering
\includegraphics[width=1\linewidth]{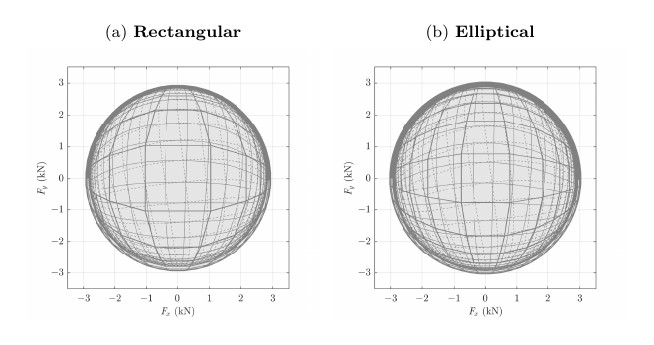} 
\caption{Friction circles (parabolic pressure distribution): (a) rectangular contact patch; (b) elliptical contact patch. Line styles: $\bm{\varphi} = \bm{0}$ (solid thick lines), $\bm{\varphi} = (3,0)$ (solid lines), $\bm{\varphi} = (0,3)$ (dashed lines). Model parameters as in Table~\ref{tab:parameters}.}
\label{fig:Circles}
\end{figure} 

\subsubsection{Miscellaneous results}\label{sect:miscellaneous}
A variety of composite plots is finally collected into Figs.~\ref{fig:GoughRect}-\ref{fig:EllSpins}, where the trend of the main steady-state characteristics is illustrated for different operating conditions, and considering again a rectangular and elliptical contact area. Starting with Figs.~\ref{fig:GoughRect} and~\ref{fig:GoughEll}, produced in the absence of spin ($\bm{\varphi} = \bm{0}$), it may be observed that the lateral force $F_y$ and vertical moment $M_z$ as a function of the lateral slip $\sigma_y$ are both symmetric with respect to the origin. It is worth noting that the vertical moment $M_z$ in Figs.~\ref{fig:GoughRect} and~\ref{fig:GoughEll} does not change sign as the lateral slip increases. This behaviour should be imputed to the adopted parabolic pressure distribution, which is symmetric with respect to the centre of the contact patch. A sign reversal would require an asymmetric pressure profile, as explained in \cite{Tsiotras3,Deur1}. Under combined translational slips, both $F_y$ and $M_z$ exhibit a reduction in magnitude. Overall, the trends remain qualitatively coherent with those typically reported in railway and tyre dynamics.

The friction circles are also depicted for different values of the lateral slip, along with the vertical moment as a function of $\sigma_x$, which is symmetric with respect to the vertical axis. For $\bm{\varphi} = \bm{0}$, \modref{stdmodel},~\ref{semilinmodel}, and~\ref{linmodel2} are all equivalent and predict the same behaviours.
\begin{figure}
\centering
\includegraphics[width=1\linewidth]{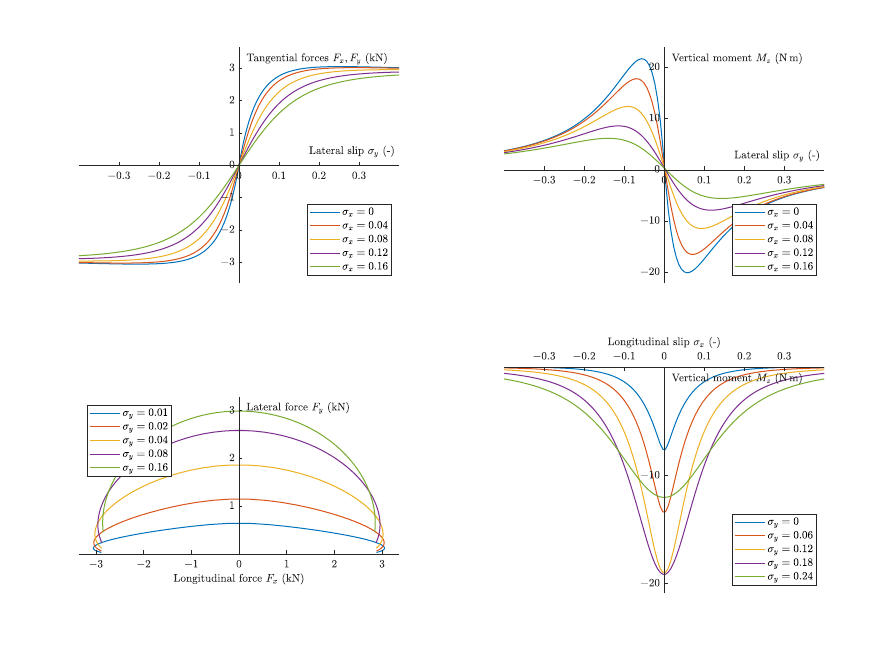} 
\caption{Steady-state characteristics in the absence of spin slips (rectangular contact area with parabolic pressure distribution). Line styles: \modref{semilinmodel} (solid thick lines), \modref{linmodel2} (solid lines), \modref{stdmodel} (dashed lines). Model parameters as in Table~\ref{tab:parameters}.}
\label{fig:GoughRect}
\end{figure} 

\begin{figure}
\centering
\includegraphics[width=1\linewidth]{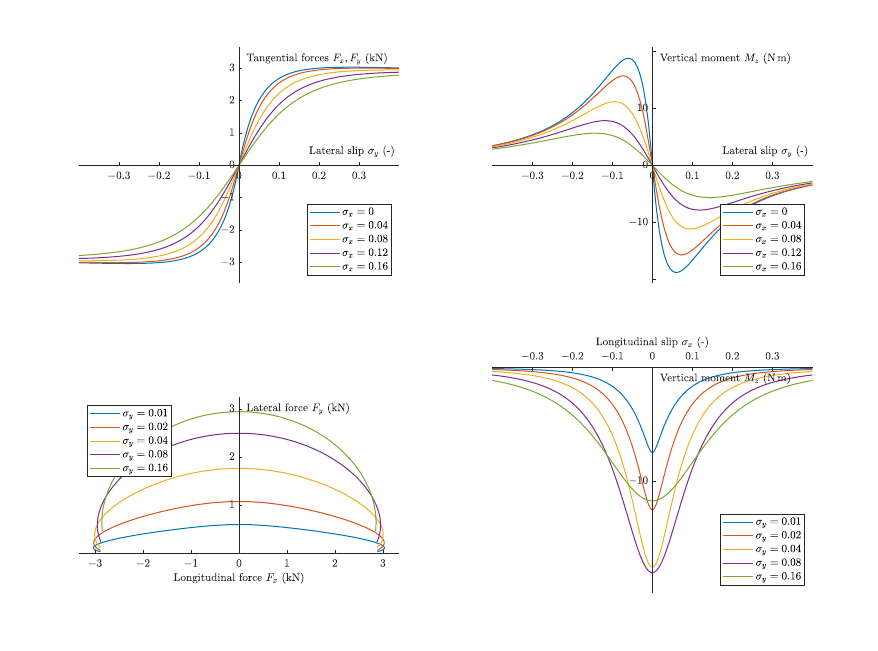} 
\caption{Steady-state characteristics in the absence of spin slips (elliptical contact area with parabolic pressure distribution). Line styles: \modref{semilinmodel} (solid thick lines), \modref{linmodel2} (solid lines), \modref{stdmodel} (dashed lines). Model parameters as in Table~\ref{tab:parameters}.}
\label{fig:GoughEll}
\end{figure} 
Some discrepancies emerge when comparing the three formulations under large spin slips, as illustrated in Figs.~\ref{fig:RectSpins} and~\ref{fig:EllSpins} for rectangular and elliptical contact areas, respectively. The disagreement is particularly evident in the friction circles associated with small lateral slips and $\bm{\varphi} = (0,3)$. In this case, although the linear model~\modref{linmodel2} partially captures the asymmetry induced by high spin values, consistent with the predictions of~\modref{semilinmodel}, the standard formulation~\modref{stdmodel} fails to reproduce this behaviour accurately. For $\bm{\varphi} = (3,0)$, the discrepancy is less pronounced. It is worth noting, however, that in tyre dynamics large values of $\varphi_1 = \varphi_\gamma$ are often associated with secondary effects -- such as distortions of the contact shape -- that are not explicitly modelled here.
\begin{figure}
\centering
\includegraphics[width=1\linewidth]{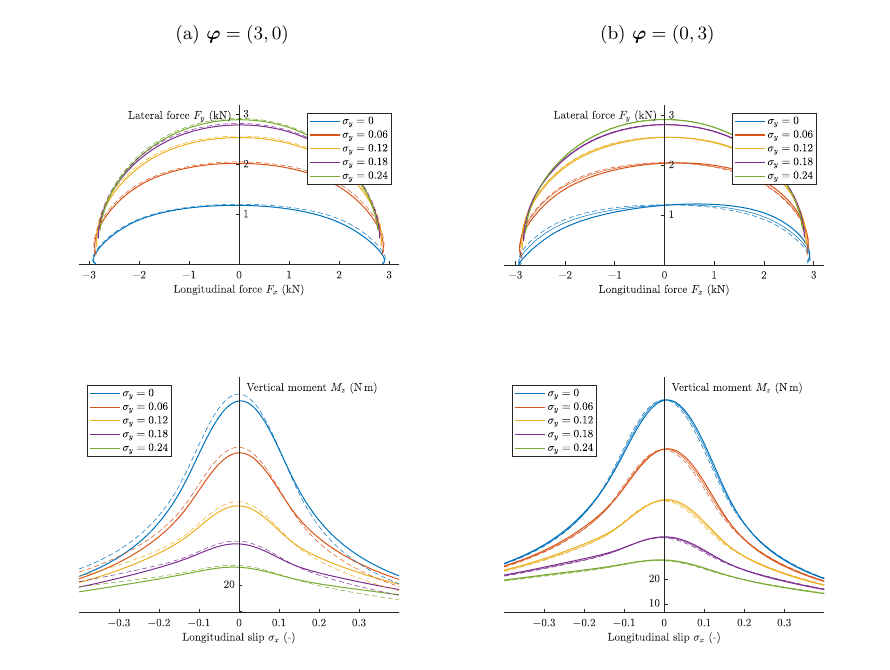} 
\caption{Steady-state characteristics in the presence of large spin slips (rectangular contact area with parabolic pressure distribution): (a) $\bm{\varphi} = (3,0)$; (b) $\bm{\varphi} = (0,3)$. Line styles: \modref{semilinmodel} (solid thick lines), \modref{linmodel2} (solid lines), \modref{stdmodel} (dashed lines). Model parameters as in Table~\ref{tab:parameters}.}
\label{fig:RectSpins}
\end{figure} 

\begin{figure}
\centering
\includegraphics[width=1\linewidth]{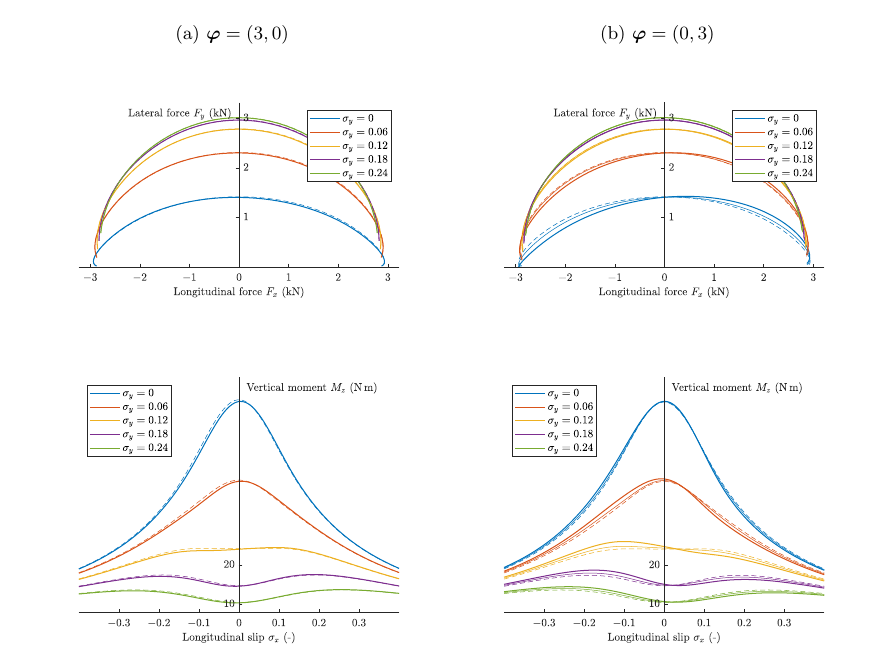} 
\caption{Steady-state characteristics in the presence of large spin slips (elliptical contact area with parabolic pressure distribution): (a) $\bm{\varphi} = (3,0)$; (b) $\bm{\varphi} = (0,3)$. Line styles: \modref{semilinmodel} (solid thick lines), \modref{linmodel2} (solid lines), \modref{stdmodel} (dashed lines). Model parameters as in Table~\ref{tab:parameters}.}
\label{fig:EllSpins}
\end{figure}

\subsection{Transient behaviour}\label{sect:transient}
Transient rolling contact is studied considering two main effects: \emph{relaxation dynamics}, and secondary phenomena connected with time-varying normal forces and contact areas. These are investigated respectively in Sect.~\ref{sect:relax} and~\ref{sect:normVertVar}.

\subsubsection{Relaxation phenomena in rolling and spinning contact}\label{sect:relax}
According to~\modref{stdmodel},~\ref{semilinmodel}, and~\ref{linmodel2}, rolling contact processes are governed by (possibly nonlinear) transport equations. Consequently, when a rolling contact system is subjected to a constant slip or spin input, it does not produce an immediate steady-state response, but instead evolves through a transient phase.

In first-order dynamical systems described by linear ODEs, a key notion is that of \emph{relaxation time}: the time required for the system to reach approximately 63\% of its steady-state response following a step input. This time is typically an integer multiple of a characteristic time constant. An analogous concept is widely used in tyre dynamics, where the \emph{relaxation length} denotes the distance a tyre must travel to develop about 63\% of its steady-state characteristics (in terms of tangential forces and/or vertical moment) \cite{Higuchi2,Rill1,PAC}. The relaxation length is typically influenced by factors other than friction, including the flexibility of structural components such as the tyre carcass and sidewall. However, relaxation phenomena are also well documented in relatively rigid contact pairs, like wheel and rail, where their duration is on the order of the contact patch length. In particular, under vanishing sliding conditions (that is, when the slip inputs are sufficiently small), standard brush models with Coulomb-Amontons friction predict that transients disappear exactly after travelling one full contact patch length. In contrast, when limited friction or local sliding occurs within the contact area, the transient duration becomes shorter and typically corresponds to the distance required to reach the so-called \emph{breakaway point}. An important insight that emerges from this discussion is that, due to the hyperbolic nature of~\modref{stdmodel},~\ref{semilinmodel}, and~\ref{linmodel2}, the response of a rolling contact system converges in finite time, and not exponentially, to a steady state (at least, in the absence of additional dynamics related to the deformation of compliant elements).

Similar conclusions, which may be reached by directly solving Eq.~\eqref{eq:standard2}\footnote{A formal solution to Eq.~\eqref{eq:standard2} may be recovered using the method of the characteristic lines; however, since the coefficients are highly nonlinear in the variables $x$ and $s$, such a representation formula would be not very neat.}, are supported by the simulation results plotted in Figs.~\ref{fig:trans1} and~\ref{fig:trans2}, where the transient characteristics predicted using~\modref{linmodel2} are plotted for different combinations of slip and spin inputs. Specifically, Figs.~\ref{fig:trans1}(a) and (b) illustrate the unsteady behaviours of the tangential forces and vertical moment for $\bm{\sigma} = (0.16,0.08)$ and $(0.06,0.18)$ in the absence of spin ($\bm{\varphi} = \bm{0}$), and starting from zero initial conditions. Interestingly, the system exhibits a dynamical response that is similar to that of a first-order linear ODE, although the concepts from ODE theory are not directly applicable to the case under consideration. In particular, it may be observed that convergence is reached approximately for $s = 0.15$ m, which is in accordance with the values reported in Table~\ref{tab:parameters}.
 
\begin{figure}
\centering
\includegraphics[width=1\linewidth]{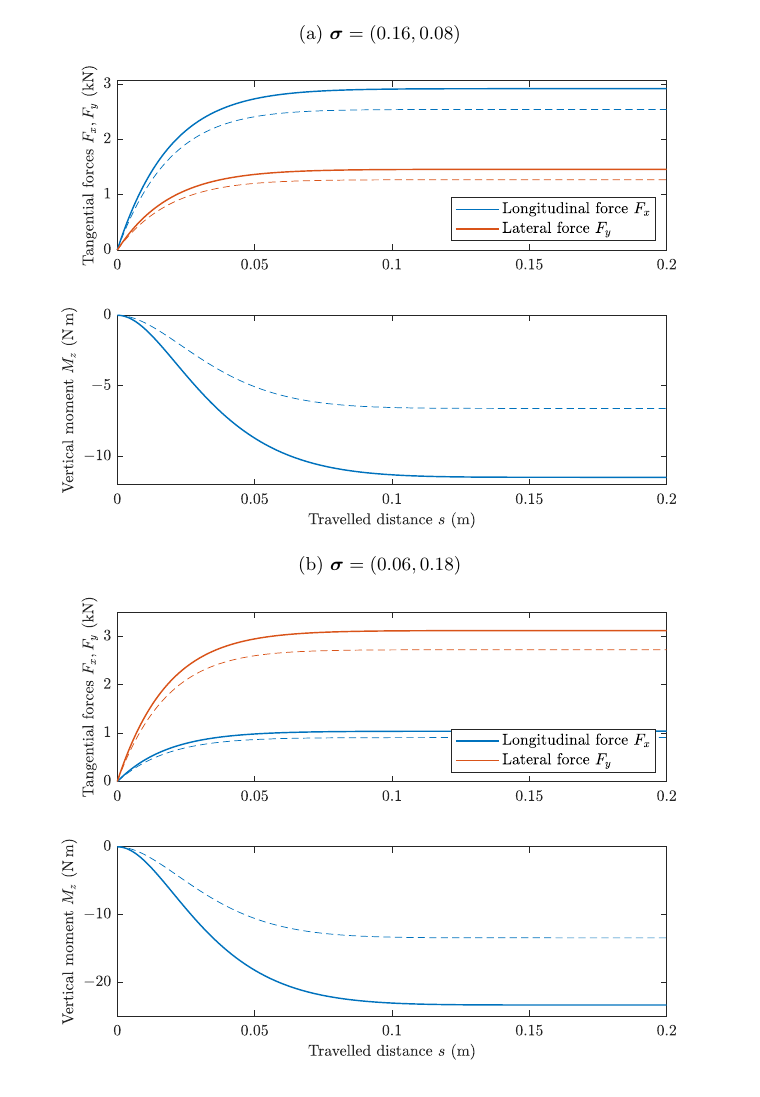} 
\caption{Transient forces predicted by~\modref{linmodel2} for step slip inputs in the absence of spin (parabolic pressure distribution). Line styles: rectangular contact area (solid thick line), elliptical contact area (dashed lines). Model parameters as in Table~\ref{tab:parameters}.}
\label{fig:trans1}
\end{figure}

According to both~\modref{semilinmodel} and~\ref{linmodel2}, the presence of large spins $\varphi_1(s) = \varphi_\gamma(s)$ modifies the transport velocity, distorting the trajectories of the bristles travelling within the contact area. For constant $\varphi_1(s) = \varphi_1$, these trajectories are circles centred at $C_{\varphi_1}$. This phenomenon has a minor influence on the duration of the transient phase, as revealed by inspection of Fig.~\ref{fig:trans2}(a), where the transient trends of $\bm{F}_{\bm{x}}(s)$ and $M_z(s)$ are illustrated for $\bm{\sigma} = (0.08,0.08)$. More interesting is instead the effect that large spin slips seem to exert on the dynamics of $M_z$, whose response is characterised by an initial overshoot in Fig.~\ref{fig:trans2}, as typically occurs for nonlinear systems.

In any case, it is important to emphasise that, unlike standard brush models with Coulomb-Amontons friction, the FrBD and LuGre formulations are inherently unable to predict transients shorter than the contact-patch length. This limitation arises because their governing PDEs are defined over the entire contact area. Although this constraint could, in principle, be alleviated by postulating \emph{ad-hoc} slip-dependent transport velocities, such a modification would lack a sound theoretical justification within the present framework and is therefore not pursued in this work.
\begin{figure}
\centering
\includegraphics[width=1\linewidth]{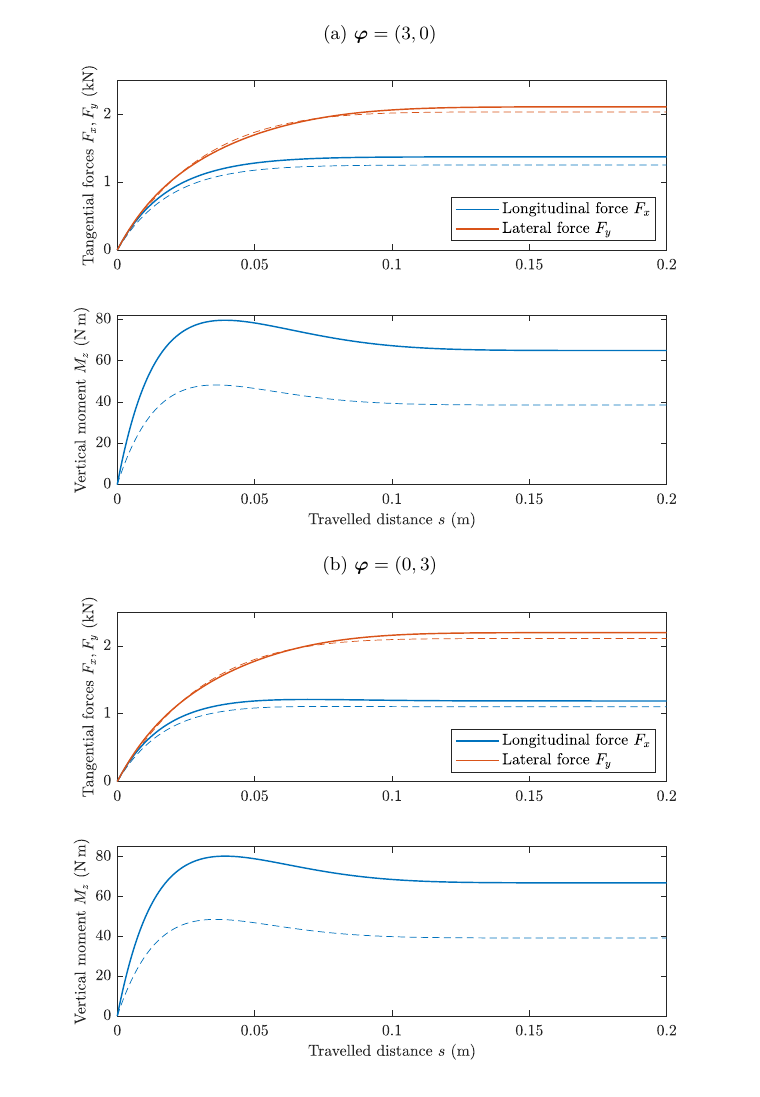} 
\caption{Transient forces predicted by~\modref{linmodel2} for step slip inputs in the presence of large spins (parabolic pressure distribution). Line styles: rectangular contact area (solid thick line), elliptical contact area (dashed lines). Model parameters as in Table~\ref{tab:parameters}.}
\label{fig:trans2}
\end{figure}

\subsubsection{Effect of time-varying normal forces}\label{sect:normVertVar}
Relaxation effects are amongst the most significant transient phenomena in tyre-road and wheel-rail interactions, as they may directly affect the stability of vehicular systems \cite{Takacs1,Takacs2,Takacs3,Takacs4,Takacs5,BicyclePDE,SemilinearV}. However, dynamical behaviours associated with oscillating normal forces and time-varying contact patches also play a fundamental role in shaping key phenomena such as slip losses, wheel polygonalisation, corrugation, and wear \cite{Ciavarella1,Ciavarella2,Ciavarella3}, whilst also influencing riding comfort in road vehicles \cite{Guiggiani}. In Sect.~\ref{sect:models}, \modref{stdmodel},~\ref{semilinmodel}, and~\ref{linmodel2} were postulated on a time-varying domain, providing an adequate framework to investigate the effect of oscillating vertical forces \cite{Colantonio} on the response of a rolling contact system. 

For simplicity, the following analysis is restricted to the case of line contact, with $\mathscr{C}(s) = \{ x \in \mathbb{R} \mathrel{|} -a(s) \leq x \leq a(s)\}$ for $a(s) \in [a\ped{min}, a\ped{max}]$, with $a(0) \triangleq a_0 \in [a\ped{min}, a\ped{max}]$. Then, \modref{stdmodel} may be recast explicitly as
\begin{subequations}\label{eq:standard0_tv}
\begin{align}
\begin{split}
& \dpd{\bm{z}(x,s)}{s} -\dpd{\bm{z}(x,s)}{x} =\tilde{\mathbf{\Sigma}}(x,s)\bm{z}(x,s) + \tilde{\bm{h}}(x,s), \quad \bm{x} \in (-a(s),a(s)), \; s \in (0,S),
\end{split}\\
& \bm{z}(a(s),s) = \bm{0}, \quad s \in (0,S), \\
& \bm{z}(x,0) = \bm{z}_0(x), \quad x \in (-a_0,a_0).
\end{align}
\end{subequations}
It is profitable to reformulate the PDE~\eqref{eq:standard0_tv} on a unit domain. To this end, the following transformation is introduced:
\begin{align}
\xi \triangleq \dfrac{a(s)-x}{2a(s)}, \quad x \in [-a(s),a(s)],
\end{align}
so that, with some abuse of notation, Eq.~\eqref{eq:standard0_tv} is converted into
\begin{subequations}\label{eq:standard0_tvFix}
\begin{align}
\begin{split}
& \dpd{\bm{z}(\xi,s)}{s} +\tilde{V}(\xi,s)\dpd{\bm{z}(\xi,s)}{\xi} =\tilde{\mathbf{\Theta}}(\xi,s)\bm{z}(\xi,s) + \tilde{\bm{g}}(x,s), \quad (\xi,s) \in (0,1) \times (0,S),
\end{split}\\
& \bm{z}(0,s) = \bm{0}, \quad s \in (0,S), \label{eq:Bcfixed}\\
& \bm{z}(\xi,0) = \bm{z}_0(\xi), \quad \xi \in (0,1),
\end{align}
\end{subequations}
being
\begin{align}\label{eq:tildeThetag}
\tilde{V}(\xi,s) & \triangleq \dfrac{1}{2a(s)}\biggl(1+(1-2\xi)\dod{a(s)}{s}\biggr), \\
\tilde{\mathbf{\Theta}}(\xi,s) & \triangleq \tilde{\mathbf{\Sigma}}\bigl(a(s)-2a(s)\xi,s\bigr), \\
 \tilde{\bm{g}}(\xi,s) & \triangleq \tilde{\bm{h}}\bigl(a(s)-2a(s)\xi,s\bigr).
\end{align}
Equation~\eqref{eq:standard0_tvFix} is now reformulated in a form that is more amenable to mathematical analysis. In this context, albeit not explicitly done in this paper, it is worth clarifying that~\modref{semilinmodel} and~\ref{linmodel2} may also be recast as PDEs evolving on a cylindrical, time-invariant domain, albeit relying on slightly more sophisticated arguments. For a general treatment of such problems in the context of rolling contact studies, the reader is redirected to \cite{Tribology}.

Theorem~\ref{thm:ex4} asserts the well-posedness of Eq.~\eqref{eq:standard0_tvFix}.
\begin{theorem}[Existence and uniqueness of solutions]\label{thm:ex4}
For all $\tilde{V}\in C^1([0,1]\times[0,S];[\tilde{V}\ped{min};\tilde{V}\ped{max}])$ with $\tilde{V}\ped{min} \in \mathbb{R}_{>0}$, $\tilde{\mathbf{\Theta}} \in C^1([0,1]\times[0,S];\mathbf{M}_2(\mathbb{R}))$, and $\tilde{\bm{g}} \in L^p((0,S);L^2((0,1);\mathbb{R}^2))$, $p \geq 1$, as in Eq.~\eqref{eq:tildeThetag}, and ICs $\bm{z}_0 \in L^2((0,1);\mathbb{R}^2)$, the PDE~\eqref{eq:standard0_tvFix} admits a unique mild solution $\bm{z} \in C^0([0,S];L^2((0,1);\mathbb{R}^2))$. Additionally, if $\tilde{\bm{g}} \in C^1([0,S];L^2((0,1);\mathbb{R}^2))$, and the IC $\bm{z}_0 \in H^1((0,1);\mathbb{R}^2)$ satisfies the BC~\eqref{eq:Bcfixed}, the solution is classical, that is, $\bm{z} \in C^1([0,S];L^2((0,1);\mathbb{R}^2)) \cap C^0([0,S];H^1((0,1);\mathbb{R}^2))$ and satisfies the BC~\eqref{eq:Bcfixed}.
\begin{proof}[Proof]
See Theorem 2.1 and Corollary 2.1 in \cite{MScthesis}.
\end{proof}
\end{theorem}

Starting with the representation in Eq.~\eqref{eq:standard0_tvFix}, and again with some abuse of notation, the tangential forces and vertical moment may be determined as
\begin{subequations}\label{eq:FandM_var}
\begin{align}
\bm{F}_{\bm{x}}(s) & =2a(s)\int_0^1 p(\xi,s)\bm{f}(\xi,s) \dif \xi, \label{eq:Fundef_varq}\\
M_z(s) & = 2a^2(s)\int_{0}^1 p(\xi,s) (1-2\xi)f_y(\xi,s) \dif \xi, \quad s \in [0,S],\label{eq:Mzunderfr_varq}
\end{align}
\end{subequations}
or, alternatively, considering the deformed configuration,
\begin{align}\label{eq:Malt_var2}
M_z(s) & = 2a^2(s)\int_{0}^1 p(\xi,s) \biggl(1-2\xi+\dfrac{z_x(\xi,s)}{a(s)}\biggr)f_y(\xi,s) \dif \xi, \quad s \in [0,S].
\end{align}

In particular, Fig.~\ref{fig:trans3} was produced by modelling the vertical force, contact patch semilength and longitudinal slip as
\begin{subequations}\label{eq:varyignSLipsa}
\begin{align}
F_z(s) & = F_{z,0} + F_{z,\mathrm{I}}\sin(\omega_s s), \\
a(s) & = a_0 + a\ped{I}\sin(\omega_s s), \\
\sigma_x(s) & = \sigma_{x,0} + \sigma_{x,\mathrm{I}}\sin(\omega_s s), 
\end{align}
\end{subequations}
where $\omega_s \in \mathbb{R}_{\geq 0}$ has the same dimension of a curvature, and represents a spatial frequency accounting for road or rail irregularities (such as roughness). The conditions $ a\ped{I} \in [0, a_0)$ and $\omega_s a\ped{I} \in [0,1)$ imply $\tilde{V}\in C^1([0,1]\times[0,S];[\tilde{V}\ped{min};\tilde{V}\ped{max}])$, ensuring the regularised PDE~\eqref{eq:standard0_tvFix} to be well-posed. In particular, the parameters $F_{z,\mathrm{I}}, a\ped{I}, \sigma_{x,\mathrm{I}} \in \mathbb{R}$ in Eq.~\eqref{eq:varyignSLipsa} were specified respectively as $F_{z,\mathrm{I}} = 0.1F_{z,0}$, $a\ped{I} = 0.1a_0$, and $\sigma_{x,\mathrm{I}} = 0.1\sigma_{x,0}$, and a constant contact distribution was employed.

Figure~\ref{fig:trans3}(a) illustrates the trend of the longitudinal force for $\omega_s = 10$ $\mathrm{m}^{-1}$ and two different values of $\sigma_{x,0} = 0.02$ and 0.06. By inspection of Fig.~\ref{fig:trans3}(a), it may be promptly concluded that, after a rapid transient phase related to the initial conditions, the force follows a sinusoidal trend that is dictated by those of the slip and contact patch length. Analogous considerations hold for Fig.~\ref{fig:trans3}(b), where the effect of higher excitation frequencies is shown ($\omega_s = 20$ $\mathrm{m}^{-1}$).
\begin{figure}
\centering
\includegraphics[width=1\linewidth]{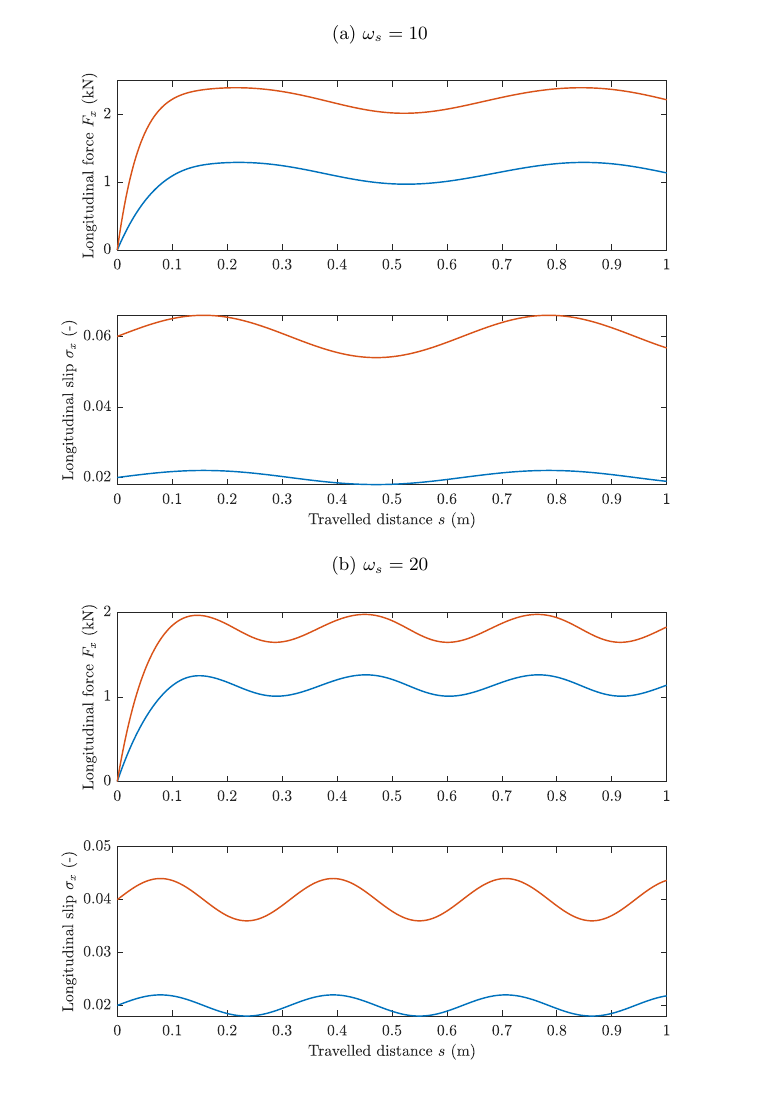} 
\caption{Transient longitudinal force in the presence of an oscillating normal load and a time-varying contact area. Line colours: (a) $\sigma_{x,0} = 0.02$ (blue lines), $\sigma_{x,0} = 0.06$ (orange lines); (b) $\sigma_{x,0} = 0.02$ (blue lines), $\sigma_{x,0} = 0.04$ (orange lines). Other parameters as in Table~\ref{tab:parameters}.}
\label{fig:trans3}
\end{figure}

\section{Conclusions}\label{sect:conclusion}
This work introduced a comprehensive two-dimensional extension of the FrBD friction framework, bridging the gap between physically grounded bristle-based friction modelling and the spatially distributed nature of rolling contact. Starting from a general rheological representation of a bristle-like element, combined with an analytical law for sliding friction, the proposed approach eliminates the notion of sliding velocity through a constructive use of the Implicit Function Theorem, yielding a dynamic friction model expressed solely in terms of the rigid relative velocity and bristle deformation. Compared to classic formulations equipped with Coulomb friction, this expedient permits avoiding the difficulties associated with separately modelling sliding and rolling regimes.

The distributed formulation developed in the paper enables frictional stresses to vary within the contact region according to local kinematics, transport effects, and spin. Specifically, three model variants of increasing complexity were derived, covering standard linear rolling contact, large-spin regimes, and semilinear dynamics. Concerning the linear descriptions, well-posedness was rigorously established, ensuring that the proposed PDE formulations are mathematically sound and suitable for the adoption of numerical schemes (at least in the context of the studies conducted in the manuscript). Stability analyses demonstrated \emph{input-to-state stability} (ISS) and \emph{input-to-output stability} (IOS), and clarified the roles of micro-stiffness and micro-damping in shaping the transient dynamics. In particular, coherently with previous findings, it emerged that passivity is almost always preserved if the damping term is modelled in a physically consistent manner.

Numerical results illustrated steady-state characteristics (including action and slip-forces surfaces) and transient behaviours connected with relaxation phenomena and oscillating normal forces. These simulations confirm that the two-dimensional FrBD framework can capture essential physical effects known from experiments and higher-fidelity models, whilst retaining a tractable and analytically transparent structure that is amenable to mathematical analysis and real-time implementations.

Beyond the specific formulations presented in the paper, the developed methodology, combining physical-oriented reasoning, implicit-function-based inversion, and distributed PDE representations, establishes a general template for future extensions. Potential directions include incorporating more detailed bristle mechanics, fully nonlinear damping laws, or coupling with flexible bodies. Indeed, the results of this work open the way for new physically-grounded, control-oriented rolling contact representations and provide a rigorous foundation for next-generation friction modelling in applications ranging from vehicle dynamics to tribology. In particular, a limitation of the specific model implementation presented in this paper relates to the simplicity of the assumed rheological representation of the bristle element, which may be appropriate for elastic or viscoelastic solids where relaxation effects are minor, but cannot accurately capture complex dissipative behaviours observed, for instance, in polymers and composite materials. Finally, starting with the formulations introduced in this manuscript, model-order-reduction strategies could be explored for potential integration with larger, system-level dynamic simulations of complex mechanical systems.

\section*{Funding declaration}
This research was financially supported by the project FASTEST (Reg. no. 2023-06511), funded by the Swedish Research Council. 

\section*{Acknowledgements}
The author thanks both reviewers for their valuable suggestions and comments, with particular appreciation to Reviewer 2 for especially helpful feedback that substantially improved the manuscript.

\section*{Compliance with Ethical Standards}

The authors declare that they have no conflict of interest.

\section*{Author Contribution declaration}
L.R. is the sole author and contributor to the manuscript.

\appendix

\section{Technical results}\label{app:1}
This appendix collects some technical results. Specifically, Appendix~\ref{app:Proof1} below sketches the proof of Theorem~\ref{thm:ex2}, whereas Appendices~\ref{app:Proof2} and~\ref{app:Proof3} contain the proofs of Lemmata~\ref{lemma:boundedness} and~\ref{lemma:DissF20}, respectively.

\subsection{Proof of Theorem~\ref{thm:ex2}}\label{app:Proof1}
The proof of Theorem~\ref{app:Proof1} is sketched below.
\begin{proof}[Sketch of the proof of Theorem~\ref{thm:ex2}]
Under the stated assumptions on $\partial \mathscr{C}$, the unbounded operator $(\mathscr{A},\mathscr{D}(\mathscr{A}))$, defined by
\begin{subequations}
\begin{align}
(\mathscr{A}\bm{\zeta})(\bm{x}) & \triangleq - \bigl(\bar{\bm{V}}(\bm{x})\cdot\nabla_{\bm{x}}\bigr)\bm{\zeta}(\bm{x}), \\
\mathscr{D}(\mathscr{A}) & \triangleq \Bigl\{\bm{\zeta} \in L^2(\mathring{\mathscr{C}};\mathbb{R}^2)\mathrel{\Big|} (\bar{\bm{V}}\cdot\nabla_{\bm{x}})\bm{\zeta} \in L^2(\mathring{\mathscr{C}};\mathbb{R}^2), \;  \eval[0]{\bm{\zeta}}_{\mathscr{L}} = \bm{0}\Bigr\},
\end{align}
\end{subequations}
generates a $C_0$-semigroup on $L^2(\mathring{\mathscr{C}};\mathbb{R}^2)$ (see \cite{Bardos}). The following perturbative term $\bm{f} : L^2(\mathring{\mathscr{C}};\mathbb{R}^2)\times [0,S] \to L^2(\mathring{\mathscr{C}};\mathbb{R}^2)$ is now considered:
\begin{align}
\bigl(\bm{f}(\bm{z},s)\bigr)(\bm{x}) \triangleq \tilde{\mathbf{\Sigma}}(\bm{x},s)\bm{z}(\bm{x}) + \tilde{\bm{h}}(\bm{x},s).
\end{align}
The conditions $\tilde{\mathbf{\Sigma}}_\varphi \in C^0(\mathscr{C}\times[0,S];\mathbf{M}_2(\mathbb{R}))$ and $\tilde{\bm{h}} \in C^0([0,S];L^2(\mathring{\mathscr{C}};\mathbb{R}^2))$ imply that $\bm{f} : L^2(\mathring{\mathscr{C}};\mathbb{R}^2)\times [0,S] \to L^2(\mathring{\mathscr{C}};\mathbb{R}^2)$ is continuous in $s$ on $[0,S]$, and uniformly Lipschitz continuous on $L^2(\mathring{\mathscr{C}};\mathbb{R}^2)$. Hence, all the hypotheses of Theorem 6.1.2 in \cite{Pazy} are verified, ensuring the existence and uniqueness of a mild solution $\bm{z} \in C^0([0,S];L^2(\mathring{\mathscr{C}};\mathbb{R}^2))$ for all ICs $\bm{z}_0 \in L^2(\mathring{\mathscr{C}};\mathbb{R}^2)$. Similarly, for $\tilde{\mathbf{\Sigma}}_\varphi \in C^1(\mathscr{C}\times[0,S];\mathbf{M}_2(\mathbb{R}))$ and $\tilde{\bm{h}} \in C^1([0,S];L^2(\mathring{\mathscr{C}};\mathbb{R}^2))$, all the hypotheses of Theorems 6.1.5 in \cite{Pazy} or 2.2 in \cite{MScthesis} are verified (see also \cite{Tanabe1,Tanabe}), ensuring the existence and uniqueness of a classical solution $\bm{z} \in C^1([0,S];L^2(\mathring{\mathscr{C}};\mathbb{R}^2)) \cap C^0([0,S];\mathscr{D}(\mathscr{A}))$ for all ICs $\bm{z}_0 \in \mathscr{D}(\mathscr{A})$.
\end{proof}

\subsection{Proof of Lemma~\ref{lemma:boundedness}}\label{app:Proof2}
The proof of Lemma~\ref{lemma:boundedness} is given below.
\begin{proof}[Proof of Lemma~\ref{lemma:boundedness}]
The proof is worked out for the PDE~\eqref{eq:quasiStationary0}; the result for~\eqref{eq:standard2} follows as a special case.
To start, it must be noted that Eq.~\eqref{eq:N} implies the existence of a vector $\mathbb{R}^2 \ni \bm{\nu} = [\nu_x\; -\nu_y]^{\mathrm{T}}\not = \bm{0}$ and a constant $\eta \in \mathbb{R}_{>0}$ such that
\begin{align}\label{eq:N2}
N(\bm{x}) \triangleq \varphi_1\bm{\nu}\cdot \bigl(\bm{x}-\bm{x}_{C_{\varphi_1}}\bigr) < -\eta, \quad \bm{x} \in \mathscr{C}.
\end{align}
Accordingly, the following Lyapunov function candidate is considered:
\begin{align}\label{eq:Lyapunov1}
W\bigl(\bm{z}(\cdot,s)\bigr) \triangleq \dfrac{1}{2}\iint_{\mathscr{C}} P(\bm{x}) \bm{z}^{\mathrm{T}}(\bm{x},s)\bm{z}(\bm{x},s) \dif \bm{x},
\end{align}
where $P(\bm{x}) = \exp(\theta \bm{\nu}\cdot\bm{x})$ with $\theta \in \mathbb{R}_{>0}$ to be appropriately selected. Taking the derivative of the above~\eqref{eq:Lyapunov1} along the dynamics~\eqref{eq:quasiStationary0dy} yields
\begin{align}
\begin{split}
\dod{W(s)}{s} & = -\dfrac{1}{2}\iint_{\mathscr{C}} P(\bm{x})\bar{\bm{V}}(\bm{x}) \cdot \nabla_{\bm{x}} \norm{\bm{z}(\bm{x},s)}_2^2 \dif \bm{x} + \iint_{\mathscr{C}}  P(\bm{x})\bm{z}^{\mathrm{T}}(\bm{x},s)\tilde{\mathbf{\Sigma}}_\varphi(\bm{x},s)\bm{z}(\bm{x},s) \dif \bm{x}  \\
& \quad + \iint_{\mathscr{C}} P(\bm{x})\bm{z}^{\mathrm{T}}(\bm{x},s)\tilde{\bm{h}}(\bm{x},s) \dif \bm{x}, \quad s \in \mathbb{R}_{>0}.
\end{split}
\end{align}
Integrating by parts the first term and recalling that $\bar{\bm{V}}(\bm{x})$ is solenoidal, namely $\nabla_{\bm{x}}\cdot \bar{\bm{V}}(\bm{x}) = 0$, gives
\begin{align}\label{eq:Lyap11}
\begin{split}
\dod{W(s)}{s} & \leq - \dfrac{1}{2}\oint_{\partial \mathscr{C}} P(\bm{x})\norm{\bm{z}(\bm{x},s)}_2^2 \bar{\bm{V}}(\bm{x})\cdot \hat{\bm{n}}_{\partial \mathscr{C}}(\bm{x}) \dif l +\dfrac{1}{2}\iint_{\mathscr{C}} \theta P(\bm{x})N(\bm{x}) \norm{\bm{z}(\bm{x},s)}_2^2\dif \bm{x} \\
& \quad + \iint_{\mathscr{C}}  P(\bm{x})\bm{z}^{\mathrm{T}}(\bm{x},s)\tilde{\mathbf{\Sigma}}_\varphi(\bm{x},s)\bm{z}(\bm{x},s) \dif \bm{x} + \iint_{\mathscr{C}} P(\bm{x})\bm{z}^{\mathrm{T}}(\bm{x},s)\tilde{\bm{h}}(\bm{x},s) \dif \bm{x},  \quad s\in \mathbb{R}_{>0},
\end{split}
\end{align}
with $N(\bm{x})$ as in Eq.~\eqref{eq:N}.
The BC~\eqref{eq:BCquasistesd} ensures the first term on the right-hand side of Eq.~\eqref{eq:Lyap11} to be nonpositive. Hence, utilising the inequality~\eqref{eq:N} gives
\begin{align}\label{eq:Lyap12}
\begin{split}
\dod{W(s)}{s} & \leq  -\dfrac{1}{2}\iint_{\mathscr{C}} \theta \eta P(\bm{x}) \norm{\bm{z}(\bm{x},s)}_2^2\dif \bm{x} + \iint_{\mathscr{C}}  P(\bm{x})\bm{z}^{\mathrm{T}}(\bm{x},s)\tilde{\mathbf{\Sigma}}_\varphi(\bm{x},s)\bm{z}(\bm{x},s) \dif \bm{x}  \\
& \quad + \iint_{\mathscr{C}} P(\bm{x})\bm{z}^{\mathrm{T}}(\bm{x},s)\tilde{\bm{h}}(\bm{x},s) \dif \bm{x},  \quad s\in \mathbb{R}_{>0}.
\end{split}
\end{align}
Moreover, $\bar{\bm{v}} \in C^0(\mathscr{C}\times\mathbb{R}_{\geq 0};\mathbb{R}^2) \cap L^\infty(\mathscr{C}\times\mathbb{R}_{\geq 0};\mathbb{R}^2)$, $V\ped{r} \in C^0(\mathbb{R}_{\geq 0};[V\ped{min},V\ped{max}])$, and $\bm{\varphi} \in \mathbb{R}\times C^0(\mathbb{R}_{\geq 0};\mathbb{R})\cap \mathbb{R}\times L^\infty(\mathbb{R}_{\geq 0};\mathbb{R})$ imply the existence of constants $M_{\tilde{\mathbf{\Sigma}}_\varphi}, M_{\mathbf{H}} \in \mathbb{R}_{\geq 0}$ such that
\begin{subequations}
\begin{align}
\norm{\tilde{\mathbf{\Sigma}}_\varphi(\bm{x},s)} & \leq M_{\tilde{\mathbf{\Sigma}}_\varphi}, \\
\norm{\tilde{\bm{h}}(\bm{x},s)}_2 & \leq M_{\mathbf{H}}\norm{\bar{\bm{v}}(\bm{x},s)}_2, 
\end{align}
\end{subequations}
for all $(\bm{x},s) \in \mathscr{C}\times \mathbb{R}_{\geq 0}$. Then, Applying Cauchy-Schwarz and subsequently Young's inequality for products to the last integral in Eq.~\eqref{eq:Lyap12} produces
\begin{align}\label{eq:Lyap13}
\begin{split}
\dod{W(s)}{s} & \leq  -\dfrac{1}{2}\iint_{\mathscr{C}} (\theta\eta-2M_{\tilde{\mathbf{\Sigma}}_\varphi}-1) P(\bm{x}) \norm{\bm{z}(\bm{x},s)}_2^2\dif \bm{x}   +\dfrac{1}{2}M_{\mathbf{H}}^2 \iint_{\mathscr{C}} P(\bm{x})\norm{\bar{\bm{v}}(\bm{x},s)}_2^2 \dif \bm{x},  \quad s\in \mathbb{R}_{>0}.
\end{split}
\end{align}
Thus, choosing $\theta > (2M_{\tilde{\mathbf{\Sigma}}_\varphi}+1)/\eta$ yields the existence of $\rho, \mu \in \mathbb{R}_{>0}$ such that
\begin{align}\label{eq:Lyap14}
\begin{split}
\dod{W(s)}{s} & \leq  -\rho W(s) +\mu\norm{\bar{\bm{v}}(\cdot,\cdot)}_\infty^2, \quad s \in \mathbb{R}_{>0}.
\end{split}
\end{align}
Finally, invoking the Grönwall-Bellman inequality provides
\begin{align}\label{eq:LyapGron}
\begin{split}
W\bigl(\bm{z}(\cdot,s)\bigr) & \leq  \exp(-\rho s)W\bigl(\bm{z}_0(\cdot)\bigr) + \dfrac{\mu}{\rho}\norm{\bar{\bm{v}}(\cdot,\cdot)}_\infty^2, \quad s \in \mathbb{R}_{\geq 0}.
\end{split}
\end{align}
Since the Lyapunov function~\eqref{eq:Lyapunov1} is equivalent to a squared norm on $L^2(\mathring{\mathscr{C}};\mathbb{R}^2)$, Eq.~\eqref{eq:LyapGron} implies~\eqref{eq:ISSbetaGamma}. 
Concerning more specifically the PDE~\eqref{eq:standard2}, the vector $\bm{\nu}$ in $P(\bm{x})$ may be specified as $\bm{\nu} = [1\; 0]^{\mathrm{T}}$ (the same is also true if $\bm{\varepsilon} = [0\; 1]^{\mathrm{T}}$ in Eq.~\eqref{eq:Vss}).
\end{proof}

\subsection{Proof of Lemma~\ref{lemma:DissF20}}\label{app:Proof3}
The proof of Lemma~\ref{lemma:DissF20} is given below.
\begin{proof}[Proof of Lemma~\ref{lemma:DissF20}]
Computing $-\langle p(\cdot)\bm{f}(\cdot,s),\bar{\bm{v}}(\cdot,s)\rangle_{L^2(\mathring{\mathscr{C}};\mathbb{R}^2)} $ provides
\begin{align}\label{eq:w0}
\begin{split}
 -\bigl\langle p(\cdot)\bm{f}(\cdot,s),\bar{\bm{v}}(\cdot,s)\bigr\rangle_{L^2(\mathring{\mathscr{C}};\mathbb{R}^2)}  & = -\iint_{\mathscr{C}} p(\bm{x})\bm{z}^{\mathrm{T}}(\bm{x},s)\mathbf{\Sigma}_0\bar{\bm{v}}(\bm{x},s)\dif \bm{x}  - \iint_{\mathscr{C}} p(\bm{x})\dod{\bm{z}^{\mathrm{T}}(\bm{x},s)}{s}V\ped{r}\mathbf{\Sigma}_1\bar{\bm{v}}(\bm{x},s)\dif \bm{x}\\
&= -\iint_{\mathscr{C}} p(\bm{x})\bm{z}^{\mathrm{T}}(\bm{x},s)\Bigl[\mathbf{\Sigma}_0 + V\ped{r}\tilde{\mathbf{\Sigma}}(\bm{x},s)^{\mathrm{T}}\mathbf{\Sigma}_1\Bigr]\bar{\bm{v}}(\bm{x},s) \dif \bm{x} \\
& \quad -  \iint_{\mathscr{C}} V\ped{r}p(\bm{x}) \bar{\bm{v}}^{\mathrm{T}}(\bm{x},s)\mathbf{H}\bigl(\bar{\bm{v}}(\bm{x},s),s\bigr)^{\mathrm{T}}\mathbf{\Sigma}_1\bar{\bm{v}}(\bm{x},s) \dif \bm{x}, \quad s \in (0,S).
\end{split}
\end{align}
Moreover, differentiating Eq.~\eqref{eq:VdissF2} along the dynamics~\eqref{eq:standard2PDE} or~\eqref{eq:quasiStationary0dy} gives
\begin{align}\label{eq:WdotPass}
\begin{split}
\dod{W\bigl(\bm{z}(\cdot,s)\bigr)}{s} & = \iint_{\mathscr{C}}p(\bm{x})\bm{z}^{\mathrm{T}}(\bm{x},s)\mathbf{\Sigma}_0\dpd{\bm{z}(\bm{x},s)}{s}\dif \bm{x} \\
& =  \iint_{\mathscr{C}}p(\bm{x})\bm{z}^{\mathrm{T}}(\bm{x},s)\mathbf{\Sigma}_0\Bigl( \tilde{\mathbf{\Sigma}}(\bm{x},s)\bm{z}(\bm{x},s) + \mathbf{H}\bigl(\bar{\bm{v}}(\bm{x},s),s\bigr)\bar{\bm{v}}(\bm{x},s)\Bigr) \dif \bm{x} \\
& \quad  -\dfrac{1}{2}\iint_{\mathscr{C}} p(\bm{x})\bar{\bm{V}}(\bm{x})\cdot\nabla_{\bm{x}}\Bigl(\bm{z}^{\mathrm{T}}(\bm{x},s)\mathbf{\Sigma}_0\bm{z}(\bm{x},s)\Bigr) \dif \bm{x}, \quad s \in (0,S).
\end{split}
\end{align}
Subtracting Eq.~\eqref{eq:WdotPass} from both sides of~\eqref{eq:w0} yields
\begin{align}\label{eq:w02}
\begin{split}
&  -\bigl\langle p(\cdot)\bm{f}(\cdot,s),\bar{\bm{v}}(\cdot,s)\bigr\rangle_{L^2(\mathring{\mathscr{C}};\mathbb{R}^2)} -\dod{W\bigl(\bm{z}(\cdot,s)\bigr)}{s} =  \iint_{\mathscr{C}} p(\bm{x})\bm{z}^{\mathrm{T}}(\bm{x},s)\mathbf{Q}_1(\bm{x},s)\bm{z}(\bm{x},s)\dif \bm{x} \\
& \qquad  +  \iint_{\mathscr{C}} V\ped{r}p(\bm{x}) \bar{\bm{v}}^{\mathrm{T}}(\bm{x},s)\mathbf{Q}_2(\bm{x},s)\bar{\bm{v}}(\bm{x},s) \dif \bm{x} \\
& \qquad  +\dfrac{1}{2}\iint_{\mathscr{C}} p(\bm{x})\bar{\bm{V}}(\bm{x})\cdot\nabla_{\bm{x}}\Bigl(\bm{z}^{\mathrm{T}}(\bm{x},s)\mathbf{\Sigma}_0\bm{z}(\bm{x},s)\Bigr) \dif \bm{x}, \quad s \in (0,S),
\end{split}
\end{align}
where
\begin{subequations}
\begin{align}
\mathbf{Q}_1(\bm{x},s) & \triangleq -\mathbf{\Sigma}_0\tilde{\mathbf{\Sigma}}(\bm{x},s) = \dfrac{1}{V\ped{r}}\norm{\mathbf{M}\bigl(V\ped{r}\bar{\bm{v}}(\bm{x},s)\bigr)V\ped{r}\bar{\bm{v}}(\bm{x},s)}_{2,\varepsilon}\mathbf{\Sigma}_0\mathbf{G}^{-1}\bigl(V\ped{r}\bar{\bm{v}}(\bm{x},s)\bigr)\mathbf{\Sigma}_0, \\
\mathbf{Q}_2(\bm{x},s) & \triangleq -\mathbf{H}\bigl(\bar{\bm{v}}(\bm{x},s),s\bigr)^{\mathrm{T}}\mathbf{\Sigma}_1 =\mathbf{M}^2\bigl(V\ped{r}\bar{\bm{v}}(\bm{x},s)\bigr) \mathbf{G}^{-1}\bigl(V\ped{r}\bar{\bm{v}}(\bm{x},s)\bigr)\mathbf{\Sigma}_1.
\end{align}
\end{subequations}
Clearly, $\mathbf{Sym}_2(\mathbb{R}) \ni \mathbf{Q}_1(\bm{x},s) \succeq \mathbf{0}$ for all $(\bm{x},s) \in \mathscr{C}\times\mathbb{R}_{\geq 0}$. Moreover, the condition $\mathbf{M}^2(\bm{y}) \mathbf{G}^{-1}(\bm{y})\mathbf{\Sigma}_1 \succeq \mathbf{0}$ for all $\bm{y}\in \mathbb{R}^2$ ensures that $\mathbf{Sym}_2(\mathbb{R}) \ni \mathbf{Q}_2(\bm{x},s) \succeq \mathbf{0}$ for all $(\bm{x},s) \in \mathscr{C}\times\mathbb{R}_{\geq 0}$. Hence, 
\begin{align}\label{eq:w03}
\begin{split}
- \bigl\langle p(\cdot)\bm{f}(\cdot,s),\bar{\bm{v}}(\cdot,s)\bigr\rangle_{L^2(\mathring{\mathscr{C}};\mathbb{R}^2)} & \geq \dod{W\bigl(\bm{z}(\cdot,s)\bigr)}{s} +\dfrac{1}{2}\iint_{\mathscr{C}} p(\bm{x})\bar{\bm{V}}(\bm{x})\cdot\nabla_{\bm{x}}\Bigl(\bm{z}^{\mathrm{T}}(\bm{x},s)\mathbf{\Sigma}_0\bm{z}(\bm{x},s)\Bigr) \dif \bm{x}, \quad s \in (0,S).
\end{split}
\end{align}
Integrating by parts and enforcing the BC~\eqref{eq:BCquasiStandard} or~\eqref{eq:BCquasistesd} gives
\begin{align}\label{eq:w02}
\begin{split}
-\bigl\langle p(\cdot)\bm{f}(\cdot,s),\bar{\bm{v}}(\cdot,s)\bigr\rangle_{L^2(\mathring{\mathscr{C}};\mathbb{R}^2)} & \geq \dod{W\bigl(\bm{z}(\cdot,s)\bigr)}{s}  -\dfrac{1}{2}\iint_{\mathscr{C}} \bigl( \nabla_{\bm{x}} \cdot p(\bm{x})\bar{\bm{V}}(\bm{x})\bigr) \bm{z}^{\mathrm{T}}(\bm{x},s)\mathbf{\Sigma}_0\bm{z}(\bm{x},s)\dif \bm{x}, \quad s \in (0,S).
\end{split}
\end{align}
Therefore, if inequality~\eqref{eq:condP} holds, it may be concluded that 
\begin{align}\label{eq:w03}
\begin{split}
-\bigl\langle p(\cdot)\bm{f}(\cdot,s),\bar{\bm{v}}(\cdot,s)\bigr\rangle_{L^2(\mathring{\mathscr{C}};\mathbb{R}^2)} & \geq  \dod{W\bigl(\bm{z}(\cdot,s)\bigr)}{s}, \quad s \in (0,S).
\end{split}
\end{align}
Integrating the above Eq.~\eqref{eq:w03} proves~\eqref{eq:Fvres} with $w(\bm{f}(\cdot,s),-\bar{\bm{v}}(\cdot,s)) =-\langle p(\cdot)\bm{f}(\cdot,s), \bar{\bm{v}}\ped{r}(\cdot,s)\rangle_{L^2(\mathring{\mathscr{C}};\mathbb{R}^2)}= -\langle p(\cdot)\bm{f}(\cdot,s), \bar{\bm{v}}(\cdot,s)\rangle_{L^2(\mathring{\mathscr{C}};\mathbb{R}^2)}$. 
\end{proof}

\section{Details on contact area and pressure distribution modelling}\label{app:Patch}

A general way of modelling a steady-state contact area is using the superellipse
\begin{align}\label{eq:contactPatch}
\mathscr{C} = \Biggl\{ \bm{x} \in \mathbb{R}^2 \mathrel{\Bigg|} \biggl(\dfrac{x}{a}\biggr)^n + \biggl(\dfrac{y}{b}\biggr)^m \leq 1\Biggr\},
\end{align}
where $a \in \mathbb{R}_{>0}$ and $b \in \mathbb{R}_{>0}$ denote the contact area semilength and semiwidth, respectively. For $n=m=2$, Eq.~\eqref{eq:contactPatch} obviously yields an elliptical contact patch, and it may also approximate arbitrarily well a rectangular shape for large $n, m \in 2\mathbb{N}$. Similarly, a steady-state contact pressure may be modelled as
\begin{align}
p(\bm{x}) = p_0\bar{p}(\bm{x}), \quad \bm{x}\in \mathscr{C},
\end{align}
where
\begin{align}
p_0 \triangleq \dfrac{F_z}{\iint_{\mathscr{C}} \bar{p}(\bm{x}) \dif \bm{x}},
\end{align}
and $\bar{p}\in C^0(\mathscr{C};\mathbb{R}_{\geq 0})$ denotes a shape function.
In particular, for a constant pressure distribution, the following expressions may be deduced:
\begin{align}\label{eq:constantPress}
p_0 = \dfrac{F_z}{4ab\mathrm{B}\biggl(\dfrac{1}{n},\dfrac{1}{m}+1 \biggr)}, \quad \textnormal{and} \quad \bar{p}(\bm{x}) = 1, \quad \bm{x}\in \mathscr{C},
\end{align}
where $\mathrm{B}(\cdot,\cdot)$ denotes the Beta function, defined as
\begin{align}
\mathrm{B}(p,q)= \mathrm{B}(p,q) = \dfrac{\mathrm{\Gamma}(p)\mathrm{\Gamma}(q)}{\mathrm{\Gamma}(p+q)},
\end{align}
in which $\mathrm{\Gamma}(\cdot)$ indicates the Gamma function.
For a parabolic pressure distribution, it is instead possible to infer the following expressions:
\begin{align}\label{eq:parabolicPressure}
p_0= \dfrac{F_z n(m+1)}{4m a b\mathrm{B}\biggl(\dfrac{1}{n}, \dfrac{m+1}{m}+1\biggr)}, \quad \textnormal{and} \quad \bar{p}(\bm{x}) = 1-\biggl(\dfrac{x}{a}\biggr)^n - \biggl(\dfrac{y}{b}\biggr)^m, \quad \bm{x} \in \mathscr{C}.
\end{align}
For instance, the numerical results reported in Sect.~\ref{sect:numer} for a rectangular contact area with a parabolic pressure distribution were generated using $n=m=18$ in Eqs.~\eqref{eq:contactPatch} and~\eqref{eq:parabolicPressure}. 

The function $\bar{p}(\cdot)$ appearing in Eqs.~\eqref{eq:constantPress} and~\eqref{eq:parabolicPressure} may be modified to account for asymmetric trends towards the leading edge. Possible choices are, respectively, 
\begin{align}\label{eq:constModified}
\bar{p}(\bm{x}) = \exp\Bigl(-\gamma\bigl(x_{\mathscr{L}}(y)-x\bigr)\Bigr), \quad \bm{x} \in \mathscr{C},
\end{align}
and
\begin{align}\label{eq:parabolicModified}
\bar{p}(\bm{x}) = \Biggl[1-\biggl(\dfrac{x}{a}\biggr)^n - \biggl(\dfrac{y}{b}\biggr)^m\Biggr]\exp\Bigl(-\gamma\bigl(x_{\mathscr{L}}(y)-x\bigr)\Bigr), \quad \bm{x} \in \mathscr{C},
\end{align}
where $\gamma \in \mathbb{R}_{\geq 0}$ is an appropriately selected parameter, and $x_{\mathscr{L}}(y)$ denotes an explicit parametrisation of the leading edge:
\begin{align}
x_{\mathscr{L}}(y) = a\Biggl[1-\biggl(\dfrac{y}{b}\biggr)^m\Biggr]^{\frac{1}{n}}, \quad y \in [-b,b].
\end{align}
Although the corresponding expressions for $p_0 \in \mathbb{R}_{>0}$ are not particularly neat, Eqs.~\eqref{eq:constModified} and~\eqref{eq:parabolicModified} may ensure the fulfillment of inequality~\eqref{eq:condP} (e.g., for~\modref{stdmodel}), whilst also capturing the sign reversal of the steady-state vertical moment $M_z$ at large slip values.

\section{Additional considerations}\label{app:details} 
Appendices~\ref{sect:details} and~\ref{app:param} provide additional details about the implementation of the models and their parametrisation.
\subsection{Implementational and computational details}\label{sect:details}
The numerical results presented in Sect.~\ref{sect:numer} were generated in MATLAB\textsuperscript{\textregistered} R2024a on an Intel(R) Core(TM) Ultra 7 155U (1.70 GHz) with a 32 GB RAM. ~\modref{stdmodel} and~\ref{linmodel2} were implemented using a finite-difference Upwind scheme, whereas~\modref{semilinmodel} employed a fixed-point iteration of~\modref{linmodel2}. In particular, using~\modref{stdmodel} with a $50\times50$ spatial meshgrid and 300 gridpoints for the slip variable, the elapsed time needed to generate Figs.~\ref{fig:GoughRect} and~\ref{fig:GoughEll} was 25.23 s and 8.3 s, respectively. Using~\modref{semilinmodel} with the same settings, Figs.~\ref{fig:RectSpins} and~\ref{fig:EllSpins} required 81.43 s and 64.95 s, respectively. The computational cost of the steady-state models is therefore relatively low, with an average of 0.054 s for every combination of slip and spin inputs in the worst scenario. The transient simulations performed in Sect.~\ref{sect:relax} were instead computationally expensive, requiring approximately 50 s each, which is unacceptable for applications demanding real-time or close-to-real-time performance. In this context, model-order-reduction strategies similar to those proposed in \cite{Tsiotras1,Tsiotras3,LuGreSpin} may be explored for potential integration with larger, system-level dynamic simulations of complex mechanical systems.

\subsection{Considerations on model parametrisation}\label{app:param}
\modref{stdmodel},~\ref{semilinmodel}, and~\ref{linmodel2} depend on several parameters that must be identified as functions of the material properties and the operating conditions, including the applied vertical load and the frictional characteristics at the contact interface. In this context, a comprehensive set of recommendations for model parametrisation can be found in \cite{Deur1}; the present appendix is limited to providing a set of practical guidelines aimed at simplifying the parametrisation process.

Specifically, the normalised micro-stiffness and micro-damping coefficients contained in the matrices $\mathbf{\Sigma}_0$ and $\mathbf{\Sigma}_1$ may be identified using standard optimisation techniques by considering the linearised forms of \modref{stdmodel},~\ref{semilinmodel}, and~\ref{linmodel2} obtained for small slip and spin inputs. In this regime, the models reduce to linear brush formulations equivalent to those presented in \cite{Meccanica2,SphericalWheel}. Importantly, the resulting equations are independent of friction-related parameters, which enables a simpler, two-stage identification procedure. In this context, it is worth noting that both $\mathbf{\Sigma}_0$ and $\mathbf{\Sigma}_1$ may be identified solely from steady-state experiments, since $\mathbf{\Sigma}_1$ multiplies the total time derivatives appearing in Eqs.~\eqref{eq:fxt} and~\eqref{eq:fxs}.

Once these parameters have been identified, the coefficients governing frictional behaviour may be optimised by considering highly nonlinear operating conditions associated with large slip values. In practice, it is common to assume that the matrix $\mathbf{M}(\bm{v}\ped{r})$ admits a factorisation of the type $\mathbf{M}(\bm{v}\ped{r}) = \theta \bar{\mathbf{M}}(\bm{v}\ped{r})$, where $\bar{\mathbf{M}} \in C^0(\mathbb{R}^2;\mathbf{Sym}2(\mathbb{R}))$ represents a matrix of nominal friction coefficients, and the scalar parameter $\theta \in \mathbb{R}_{>0}$ captures additional functional dependencies. For instance, in tyre-road interactions, $\theta$ may account for variations in friction arising from tread temperature, asphalt conditions, or third-body layer effects.

Finally, geometric parameters related to the contact patch may be inferred directly from measurements of the contact area, whilst the total vertical load is typically known from steady-state experiments. In the absence of direct measurements, a rectangular contact patch may be assumed for tyres and an elliptical one for railway wheels; alternative choices should be justified based on the geometry of the contacting bodies. With regard to the pressure distribution, a uniform vertical pressure is the simplest assumption, although non-uniform distributions, mildly decreasing along the rolling direction, may be adopted when passivity requirements must be satisfied.

\newpage
\textbf{Errata} \newline

Errata to "Romano, L. Two-dimensional FrBD friction models for rolling contact. Nonlinear Dyn 114, 444 (2026). https://doi.org/10.1007/s11071-026-12298-x". 

\begin{enumerate}
\item The statement of Theorem 3.1 was modified as follows:

Suppose that $\mathring{\mathscr{C}}\subset \mathbb{R}^2$ is bounded, with boundary $\partial \mathscr{C}$ piecewise $C^1$. Then, for all $\tilde{\mathbf{\Sigma}} \in C^0(\mathscr{C}\times[0,S];\mathbf{M}_2(\mathbb{R}))$ and $\tilde{\bm{h}} \in C^0([0,S];L^2(\mathring{\mathscr{C}};\mathbb{R}^2))$ as in Eq. (45), and ICs $\bm{z}_0 \in L^2(\mathring{\mathscr{C}};\mathbb{R}^2)$, the PDE (47) admits a unique \emph{mild solution} $\bm{z} \in C^0([0,S];L^2(\mathring{\mathscr{C}};\mathbb{R}^2))$. Additionally, if $\tilde{\mathbf{\Sigma}} \in C^1(\mathscr{C}\times[0,S];\mathbf{M}_2(\mathbb{R}))$, $\tilde{\bm{h}} \in C^1([0,S];L^2(\mathring{\mathscr{C}};\mathbb{R}^2))$, and the IC $\bm{z}_0 \in \mathscr{D}(\mathscr{A})$, with $\mathscr{D}(\mathscr{A}) \triangleq \{\bm{\zeta} \in L^2(\mathring{\mathscr{C}};\mathbb{R}^2) \mathrel{|} \pd{\bm{\zeta}}{x} \in  L^2(\mathring{\mathscr{C}};\mathbb{R}^2), \; \eval[0]{\bm{\zeta}}_{\mathscr{L}} = \bm{0}\} $, the solution is \emph{classical}, that is, $\bm{z} \in C^1([0,S];L^2(\mathring{\mathscr{C}};\mathbb{R}^2)) \cap C^0([0,S];\mathscr{D}(\mathscr{A}))$.

\item The statement of Theorem 3.2 was modified as follows:

Suppose that $\bar{\bm{V}} \in C^1(\mathscr{C};\mathbb{R}^2)$ reads as in Eq. (60), and that $\mathring{\mathscr{C}}\subset \mathbb{R}^2$ is bounded, with boundary $\partial \mathscr{C}$ piecewise $C^1$. Then, for all $\tilde{\mathbf{\Sigma}}_{\varphi} \in C^0(\mathscr{C}\times[0,S];\mathbf{M}_2(\mathbb{R}))$ and $\tilde{\bm{h}} \in C^0([0,S];L^2(\mathring{\mathscr{C}};\mathbb{R}^2))$ as in Eq. (56), and ICs $\bm{z}_0 \in L^2(\mathring{\mathscr{C}};\mathbb{R}^2)$, the PDE (59) admits a unique mild solution $\bm{z} \in C^0([0,S];L^2(\mathring{\mathscr{C}};\mathbb{R}^2))$. Additionally, if $\tilde{\mathbf{\Sigma}}_{\varphi} \in C^1(\mathscr{C}\times[0,S];\mathbf{M}_2(\mathbb{R}))$, $\tilde{\bm{h}} \in C^1([0,S];L^2(\mathring{\mathscr{C}};\mathbb{R}^2))$, and the IC $\bm{z}_0 \in \mathscr{D}(\mathscr{A})$, with $\mathscr{D}(\mathscr{A}) \triangleq \{\bm{\zeta} \in L^2(\mathring{\mathscr{C}};\mathbb{R}^2) \mathrel{|} (\bar{\bm{V}}\cdot\nabla_{\bm{x}})\bm{\zeta} \in L^2(\mathring{\mathscr{C}};\mathbb{R}^2), \;  \eval[0]{\bm{\zeta}}_{\mathscr{L}} = \bm{0}\}$, the solution is classical, that is, $\bm{z} \in C^1([0,S];L^2(\mathring{\mathscr{C}};\mathbb{R}^2)) \cap C^0([0,S];\mathscr{D}(\mathscr{A}))$.

\item In the proof of Theorem 3.2, Eq. (84b) was modified as
\begin{align*}
\mathscr{D}(\mathscr{A}) & \triangleq \Bigl\{\bm{\zeta} \in L^2(\mathring{\mathscr{C}};\mathbb{R}^2)\mathrel{\Big|} (\bar{\bm{V}}\cdot\nabla_{\bm{x}})\bm{\zeta} \in L^2(\mathring{\mathscr{C}};\mathbb{R}^2), \;  \eval[0]{\bm{\zeta}}_{\mathscr{L}} = \bm{0}\Bigr\}.
\end{align*}
The part "which is equivalent to $\bm{z} \in C^1([0,S];L^2(\mathring{\mathscr{C}};\mathbb{R}^2)) \cap C^0([0,S];H^1(\mathring{\mathscr{C}};\mathbb{R}^2))$ satisfying the BC (59b)" was removed.

\item In Definitions 4.1, 4.2, and 4.3, the condition "$\bm{z}_0 \in H^1(\mathring{\mathscr{C}};\mathbb{R}^2)$ satisfying the BC" was replaced by "$\bm{z}_0 \in \mathscr{D}(\mathscr{A})$".
\end{enumerate}

\end{document}